\documentclass[11pt, twoside, sort, a4paper]{article}
\usepackage{graphicx}
\usepackage{subcaption}
\usepackage{hyperref, bm}
\usepackage{mathtools}
\newcommand{\myblue}[1]{\color{black}{#1}}
\newcommand{\zblue}[1]{\color{black}{#1}}

\newcommand{\dangblue}[1]{\color{black}{#1}}
\newcommand{\imablue}[1]{\color{black}{#1}}
\newcommand{\imapurple}[1]{\color{black}{#1}}
\usepackage{geometry}
\usepackage{empheq}
 \usepackage{mathrsfs}
 \usepackage{tikz}
\newcommand{\myout}{\scalebox{0.7}{\text{out}}}

\newcommand{\rhoh}{\hat{\rho}}

\newcommand{\mysum}{\sideset{}{^{\boldsymbol{*}}}\sum}
\AtBeginDocument{
\addtolength{\abovedisplayskip}{-0.5ex}
\addtolength{\abovedisplayshortskip}{-0.5ex}
\addtolength{\belowdisplayskip}{-0.5ex}
\addtolength{\belowdisplayshortskip}{-0.5ex}
}
\makeatletter
\newcommand{\subalign}[1]{%
  \vcenter{%
    \Let@ \restore@math@cr \default@tag
    \baselineskip\fontdimen10 \scriptfont\tw@
    \advance\baselineskip\fontdimen12 \scriptfont\tw@
    \lineskip\thr@@\fontdimen8 \scriptfont\thr@@
    \lineskiplimit\lineskip
    \ialign{\hfil$\m@th\scriptstyle##$&$\m@th\scriptstyle{}##$\crcr
      #1\crcr
    }%
  }
}

\usepackage{amsmath, amsthm, amssymb}
\usepackage{caption}
\usepackage{tikz}
\usetikzlibrary{matrix, positioning, calc}
\usepackage{multirow, nccmath}
\usepackage{algorithmic}
\usepackage{algorithm}

\usepackage{comment}
\usepackage{color}
\usepackage[titletoc,title]{appendix}

\usepackage[running, mathlines]{lineno}

\newcommand{\EQ}{\begin{equation}}
\newcommand{\EN}{\end{equation}}
\newcommand{\EQS}{\begin{equation*}}
\newcommand{\ENS}{\end{equation*}}
\newcommand{\EQA}{\begin{eqnarray}}
\newcommand{\ENA}{\end{eqnarray}}
\newcommand{\EQAS}{\begin{eqnarray*}}
\newcommand{\ENAS}{\end{eqnarray*}}
\newcommand{\ds}{\displaystyle}

\usepackage{accents}
\newlength{\dhatheight}

\newcommand{\myin}{\scalebox{0.7}{\text{in}}}

\newcommand{\mymin}{\scalebox{0.55}{$\min$}}
\newcommand{\mymax}{\scalebox{0.55}{$\max$}}
\newcommand{\myx}{\scalebox{0.6}{$x$}}
\newcommand{\myy}{\scalebox{0.6}{$y$}}

\newcommand{\taus}{\tau_{m}}
\newcommand{\Bs}{\gbt^{\al}}

\newcommand{\x}{{\bf{x}}}

\newcommand{\xh}{{\bf{\hat{x}}}}

\newcommand{\al}{\alpha}
\newcommand{\Al}{\mathcal{A}}

\newcommand{\Lcal}{\mathcal{L}}

\newcommand{\Ocal}{\mathcal{O}}

\newcommand{\0}{{\mathbf{0}}}
\newcommand{\U}{{\bf{u}}}
\newcommand{\Ut}{{\Tilde{\bf{u}}}}

\newcommand{\V}{{\bf{v}}}
\newcommand{\Vt}{{\Tilde{\bf{v}}}}

\newcommand{\gb}{{\bf{g}}}
\newcommand{\gbt}{{\Tilde{\bf{g}}}}
\newcommand{\gbh}{{\widehat{\bf{g}}}}

\newcommand{\Rbb}{\mathbb{R}}

\newcommand{\Nbb}{\mathbb{N}}
 \newcommand{\Jbb}{\mathbb{J}}
\newcommand{\Nd}{\mathbb{N}^{\dagger}}

\newcommand{\Jd}{\mathbb{J}^{\dagger}}

\newcommand{\Oinf}{\Omega^{\scalebox{0.5}{$\infty$}}}
\newcommand{\C}[1]{{\mathcal{C}^{\infty}{(#1)}}}
\numberwithin{equation}{section}
\numberwithin{table}{section}
\numberwithin{figure}{section}

\newtheorem{definition}{Definition}
\newtheorem{theorem}{Theorem}
\newtheorem{lemma}{Lemma}
\newtheorem{remark}{Remark}

\numberwithin{definition}{section}
\numberwithin{theorem}{section}
\numberwithin{lemma}{section}
\numberwithin{remark}{section}
\numberwithin{assumption}{section}
\numberwithin{condition}{section}
\numberwithin{property}{section}
\numberwithin{proposition}{section}
\numberwithin{corollary}{section}
\numberwithin{algorithm}{section}

\def\r{\right}
\def\l{\left}

\makeatletter
\renewcommand*\env@matrix[1][c]{\hskip -\arraycolsep
  \let\@ifnextchar\new@ifnextchar
  \array{*\c@MaxMatrixCols #1}}
\makeatother

\makeatletter
\def\thm@space@setup{\thm@preskip=3pt
\thm@postskip=3pt}
\makeatother
\usepackage{titlesec}
\titlespacing*{\section}
  {0pt}{0.7\baselineskip}{0.2\baselineskip}
\titlespacing*{\subsection}
  {0pt}{0.7\baselineskip}{0.2\baselineskip}

\usepackage[normalem]{ulem}
\usepackage{cancel}

\DeclareMathOperator{\T}{{\scalebox{0.6}{T}}}
\begin{document}

\title{A monotone piecewise constant control integration approach for
the two-factor uncertain volatility model}

\author{
Duy-Minh Dang\thanks{School of Mathematics and Physics, The University of Queensland, St Lucia, Brisbane 4072, Australia,
email: \texttt{duyminh.dang@uq.edu.au}
}
\and
Hao Zhou \thanks{School of Mathematics and Physics, The University of Queensland, St Lucia, Brisbane 4072, Australia,
email: \texttt{h.zhou3@uq.net.au}
}
}
\date{\today}
\maketitle
\begin{abstract}
Option contracts on two underlying assets within uncertain volatility models have their worst-case and best-case prices determined by a two-dimensional (2D) Hamilton-Jacobi-Bellman (HJB) partial differential equation (PDE) with cross-derivative terms.
This paper introduces a novel \mbox{``decompose and integrate, then optimize''} approach to tackle this HJB PDE. Within each timestep, our method applies piecewise constant control, yielding a set of independent linear 2D PDEs, each corresponding to a discretized control value. Leveraging closed-form Green's functions, these PDEs are efficiently solved via 2D convolution integrals using a monotone numerical integration method.
The value function and optimal control are then obtained by synthesizing the solutions of the individual PDEs. For enhanced efficiency, we implement the integration via Fast Fourier Transforms, exploiting the Toeplitz matrix structure.
{\imapurple{The proposed method is $\ell_{\infty}$-stable}}, consistent in the viscosity sense, and converges to the viscosity solution of the HJB equation. Numerical results show excellent agreement with benchmark solutions obtained by finite differences, tree methods, and Monte Carlo simulation, highlighting its robustness and effectiveness.

\vspace{.1in}

\vspace{.1in}

\end{abstract}

\section{Introduction}
\label{sec:intro}
The uncertain volatility model is an approach in quantitative finance where the instantaneous volatility of a risky asset is allowed to vary within a specified range \cite{lyons95, avellaneda95, BJMJ2022}. This stands in contrast to the more traditional approaches where volatility is often assumed to be either deterministic (as in the Black-Scholes model) or stochastic (as in the Heston model \cite{heston93} or the SABR model \cite{hagan2002managing}).
While stochastic volatility models can deliver a more detailed depiction of volatility's dynamic evolution and its interaction with asset prices, uncertain volatility models are particularly well-suited for worst-case scenario analysis.
Specifically, although the price of a financial contract is no longer unique under an uncertain volatility model,
for risk management, especially for sellers, the primary concern often lies in the worst-case scenario, which corresponds to the contract's maximum value.  Conversely, for the buyers, the worst-case scenario corresponds to the minimum potential value of a contract.
It is worth noting that the worst-case scenario for the seller of a contract is essentially the buyer's best-case scenario,  and vice versa.
The maximum and minimum value of a contract can be formulated as solution to a Hamilton-Jacobi-Bellman (HJB) equation, which
needs to be solved numerically \cite{pooley2003, MaForsyth2015, dokuchaev1998pricing, avellaneda1999combinatorial, smith2002american, pooley2003two, jaroszkowski2022valuation, RF2016}.

Provable convergence of numerical methods for (multi-dimensional) HJB equations are typically built upon
the framework established by  Barles and Souganidis in \cite{barles-souganidis:1991}.
This framework requires numerical methods to be (i) $\ell_{\infty}$-stable, (ii) consistent, and (iii) monotone (in the viscosity sense),
provided that a strong comparison result holds. Among these requirements, monotonicity is often  the most challenging to achieve. Non-monotone schemes could produce numerical solutions that fail to converge to viscosity solutions, resulting in a violation of the no-arbitrage principle (see, for example,  \cite{Oberman2006, pooley2003, warin2016}, among many other publications).

To the best of our knowledge, numerical techniques for HJB PDEs are predominantly dominated by finite difference (FD) methods. 
At each timestep, these methods typically involve discretizing the temporal and spatial partial derivatives in the HJB equation using finite difference (FD) schemes. It~is well-known that explicit time-stepping methods, while computationally simple, are subject to the Courant–Friedrichs–Lewy (CFL) condition, which imposes stringent restrictions on the timestep size to ensure stability of FD schemes \cite{courant1928partiellen, courant1967partial}. To circumvent this constraint, fully implicit timestepping is often employed in conjunction with a positive coefficient discretization method \cite{wang08, forsyth2007numerical}.
This combination ensures the monotonicity of the numerical schemes. The optimal control is subsequently determined by solving the resulting nonlinear discretized equations, often via variants of policy iteration \cite{Bok2009}. During this process, a local optimization problem at each grid point is addressed in every policy iteration. This conventional approach is succinctly termed “discretize, then optimize” \cite{RF2016}. Importantly, the positive coefficient discretization method provides a sufficient condition ensuring the convergence of policy iteration, regardless of the initial iterate. As such, this condition must be satisfied at each policy iteration.

We highlight that multi-dimensional HJB PDEs, including those from two-factor uncertain volatility models \cite{MaForsyth2015}, pose significant challenges due to cross derivative terms when the correlation between the two underlying risky assets is non-zero. At each policy iteration, construction of a monotone finite difference scheme via a positive coefficient discretization method is often addressed using a local coordinate rotation of the computational stencil. Originally developed for explicit wide stencil schemes in \cite{bonnans2003consistency, DJ12}, this method was refined in \cite{MaForsyth2015} for a fully implicit timestepping, circumventing timestep stability restrictions. However, as noted in \cite{MaForsyth2015}, this approach adds a significant computational overhead. For further details of numerical techniques for one- and two-dimensional HJB PDEs
resulting from one-factor and two-factor uncertain volatility models (with an uncertain correlation between the two  underlying risky assets), we refer the reader to \cite{pooley2003} and \cite{MaForsyth2015}, respectively.

In this paper, we present a streamlined approach to tackle the two-dimensional (2D) HJB PDE stemming from two-factor uncertain volatility models. Moving beyond the conventional ``discretize, then optimize'',  we introduce a ``decompose and integrate, then optimize'' approach.
In each timestep, we employ a piecewise constant control technique \cite{krylov99}, which yields a set of independent 2D linear PDEs, each corresponding to a fixed control value. Rather than discretizing the temporal and spatial derivatives of these PDEs, we utilize a Green's function representation to express each solution at the next time point as a
convolution integral.
{\imapurple{This integral is then evaluated with an unconditionally monotone
numerical integration scheme, thereby bypassing derivative-based discretization,
circumventing explicit–scheme CFL restrictions, and yielding an $\ell_{\infty}$-stable method.}}
The optimal value function and control are then obtained by synthesizing the solutions of these linear PDEs, significantly simplifying the process compared to policy iteration and avoiding the aforementioned challenges associated with positive coefficient FD discretization of cross-derivative terms.
This  approach is in line with recent developments in monotone and $\epsilon$-monotone
numerical integration methods for control problems in finance, which also merit attention
\cite{ForsythLabahn2017, online23, lu2024semi, LuDang2023, ZhouDang2025, du2024fourier, zhang2023monotone}. We note a recent study \cite{RF2016} that also utilizes a piecewise constant control technique. However, this work remains anchored in the finite difference framework and incorporates a switching system, thereby necessitating interpolation when searching for the optimal control.

The main contributions of our paper are outlined below.
\begin{itemize}

\item[(i)]
The maximum and minimum value of an option contract under a two-factor uncertain volatility model with uncertain correlation
is presented as an HJB PDE posed on an finite definition domain consisting of an interior and boundary sub-domains
with appropriate boundary conditions.

\item[(ii)]

We develop a monotone piecewise constant control integration scheme for the HJB equation that, at each timestep, solves a set of independent linear 2D PDEs corresponding to discretized controls. Leveraging the known closed-form Fourier transforms of the associated Green's functions, we derive explicit expressions for these functions. Using these expressions, we approximate the solutions of the linear PDEs via 2D convolution integrals evaluated with a monotone numerical integration method. These solutions are then combined to approximate the value function and optimal control, thereby capturing the nonlinearity of the HJB equation.

Our scheme not only simplifies the optimization process compared to policy iteration but also avoids the usual complications with positive coefficient FD discretization of cross derivative terms. The availability of the Green's functions in closed form enables a systematic and quantifiable approach for determining computational domain sizes, marking a significant advantage over the heuristic or trial-and-error methods common in FD and tree techniques.
Furthermore, the Green's function's ``cancellation property'' \cite{garronigreenfunctionssecond92} effectively mitigates the impact of errors in artificial boundary conditions. These combined factors ensure that our method significantly enhances the numerical solution's accuracy and reliability.

\item[(iii)]
Utilizing the Toeplitz matrix structure, we present an efficient implementation of our monotone piecewise constant control integration scheme using FFTs and circulant convolution. The implementation process includes expanding the inner summation’s convolution kernel into a circulant matrix, followed by expanding the kernel for the double summation to achieve a circulant block arrangement. This allows the circulant matrix-vector product to be efficiently computed as a circulant convolution using 2D FFTs.

\item[(iv)]
We mathematically demonstrate that the proposed monotone scheme is also
{\imapurple{$\ell_{\infty}$-stable}} and consistent in the viscosity sense, proving its pointwise convergence to the viscosity solution of the 2D HJB PDE as the discretization parameter approaches zero.

\item[(v)] Extensive numerical results show remarkable agreement with benchmark solutions
from monotone FD, tree-grid methods, and Monte Carlo simulation, underscoring the effectiveness of our approach. In particular, for convex payoffs, our method achieves high accuracy
with a single time step, and even in general settings with multiple time steps, it significantly outperforms unconditionally monotone FD schemes in both accuracy and run time, as reported in the literature. Notably, for general payoffs, we often observe experimentally
first-order convergence, significantly exceeding the $1/6$ rate proved in \cite{krylov99}
via purely probabilistic techniques.

\end{itemize}
Although our focus is specifically on monotone piecewise constant control integration methods for two-factor uncertain volatility models with uncertain correlation, our comprehensive and systematic approach could serve as a numerical and convergence analysis framework for the development of similar piecewise constant control monotone integration methods for other HJB PDEs arising in finance.

The remainder of the paper is organized as follows.
In Section~\ref{sc:formulation}, we briefly describe the two-factor uncertain volatility model and present a 2D HJB PDE.
We then define a localized problem for this HJB equation, including conditions for boundary sub-domains.
A simple and easy-to-implement monotone piecewise constant control integration scheme via a composite 2D quadrature rule is described in Section~\ref{section:num}. In Section~\ref{sc:conv}, we mathematically establish convergence the proposed piecewise constant control monotone integration scheme
to the viscosity solution of the 2D HJB PDE. Numerical results  are given in Section~\ref{sec:Numeri_exp}.
Section~\ref{sc:conclude} concludes the paper and outlines possible future work.

\section{Formulation}
\label{sc:formulation}
Let $T>0$ be a finite investment horizon. For each $t \in [0, T]$, we denote by  $X_t$ and $Y_t$ the prices at time $t$ of two distinct underlying assets.
In this paper, for brevity, we occasionally employ the subscript/superscript $z \in \{x, y\}$ to indicate that the discussion pertains to quantities related to the respective underlying assets. We assume that the risk-neutral dynamics of the process $\{Z_t\}_{t \in [0, T]}$, where $Z_t$ can be either $X_t$ or $Y_t$, follow
\EQA
\label{eq:2SDE_GBM}
d Z_{ t} = r Z_t  dt + \sigma_z Z_t dW_t^{z}, \quad Z_0 > 0 \text{ given,} \quad Z_t \in \{X_t,Y_t\}, \quad t \in (0, T].
\ENA
Here, $r>0$ is the risk-free interest rate; $\sigma_z>0$, $z\in \{x,y\}$, respectively are the instantaneous volatility for the associated underlying asset; $\{W_t^{z}\}_{t \in [0, T]}$ are correlated Brownian motions, with $dW^{x}_tdW^{y}_t=\rho dt$, where $-1\le \rho\le 1$ is the correlation parameter. In the uncertain volatility model, the instantaneous volatility  $\sigma_z$, $z\in\{x,y\}$, in \eqref{eq:2SDE_GBM} are uncertain, but are assumed to lie within a known range \cite{Lyons1995}. That is, $\sigma_z \in [\sigma^{z}_{\mymin},\sigma^{z}_{\mymax}]$, $z\in\{x,y\}$,
where $0< \sigma^{z}_{\mymin} < \sigma^{z}_{\mymax}$ are pre-determined and fixed constants.
In addition, the correlation between the two underlying assets is also permitted to be uncertain, lying within a known range, i.e.\ $\rho\in[\rho_{\mymin},\rho_{\mymax}]$, where $-1\le \rho_{\mymin}\le  \rho_{\mymax} \le 1$
are also pre-determined and fixed  constants.
In this setting, since the instantaneous volatilities $\sigma_{z}$, $z\in\{x,y\}$ and the correlation $\rho$ are uncertain,
the price of an option is no longer unique. However, for hedging purposes, we can determine the worst-case prices for the long or short positions.
These prices are essentially the hedging costs for the associated positions.

For the underlying asset processes  $\{X_t, Y_t\}$, $t \in [0, T]$, defined in \eqref{eq:2SDE_GBM},
we let $(x^{\prime},y^{\prime})$ be the state of system. We denote by $v^{\prime}(x^{\prime},y^{\prime}, t)$ the time-$t$ worst-case price of the short or long position in a European option contract with time-$T$  payoff given by function $p(x^{\prime},y^{\prime})$. By dynamic programming, $v^{\prime}(x^{\prime},y^{\prime}, t)$ is shown to satisfy the HJB PDEs
\begin{subequations}
\label{eq:UVM_HJB}
\begin{empheq}[left={
0 = \empheqlbrace}]{alignat=8}
& \big(- v^{\prime}_{t} -\sup_{\alpha \in \Al^{\prime}}\Lcal^{\prime}_{\alpha} v^{\prime}\big), ~
\text{or} ~
\big(-v^{\prime}_{t} -\inf_{\alpha \in \Al^{\prime}}\Lcal^{\prime}_{\alpha} v^{\prime}\big),
                 &\quad(x^{\prime},y^{\prime},t) &\in  \Rbb^+ \times \Rbb^+ \times [0, T),
\label{eq:UVM_HJB_a}
\\
&v^{\prime}(x^{\prime},y^{\prime},t) - p(x^{\prime},y^{\prime}), &\quad (x^{\prime},y^{\prime}, t) &\in  \Rbb^+ \times \Rbb^+ \times \{T\}.
\label{eq:UVM_HJB_b}
\end{empheq}
\end{subequations}
In \eqref{eq:UVM_HJB}, the $\sup_{\alpha}$ and $\inf_{\alpha}$ correspond to the
worst-case for the short and for the long positions, respectively;
$\alpha$ is the control, where $\alpha = (\sigma_x,\sigma_y,\rho)$;
the differential operator $\Lcal^{\prime}_\alpha(\cdot)$, where the subscript indicates
its dependence $\alpha$, is defined as
\EQA
\label{eq:L_UVM_d}
\Lcal^{\prime}_{\alpha} v^{\prime} = \frac{(\sigma_x)^2{(x^{\prime})}^2}{2}v^{\prime}_{x^{\prime}x^{\prime}}  + rx^{\prime}v^{\prime}_{x^{\prime}}+
\frac{(\sigma_y)^2(y^{\prime})^2}{2}v^{\prime}_{y^{\prime}y^{\prime}}+ ry^{\prime}v^{\prime}_{y^{\prime}}+\rho \sigma_x\sigma_yx^{\prime}y^{\prime}v^{\prime}_{x^{\prime}y^{\prime}} - rv^{\prime}.
\ENA
The admissible control set, denoted by $\Al^{\prime}$, is given by
\begin{align}
\label{eq:adm_control}
\Al^{\prime}&=\Al_x\times\Al_y\times\Al_\rho, \text{ where }
\Al_x\equiv\l[\sigma^{\myx}_{\mymin},\sigma^{\myx}_{\mymax}\r],~\Al_y\equiv\l[\sigma^{\myy}_{\mymin},\sigma^{\myy}_{\mymax}\r],
~\Al_\rho\equiv[\rho_{\mymin},\rho_{\mymax}],
\\
~
& \qquad \qquad 0 < \sigma^{z}_{\mymin} < \sigma^{z}_{\mymax} < \infty,~z \in \{x,y\}, ~ -1\le \rho_{\mymin} \le  \rho_{\mymax} \le 1.
\nonumber
\end{align}

{\dangblue{
\begin{remark}[Restriction of control set $\Al^{\prime}$]
\label{adm_con_set}
The literature highlights that the optimal value for the objective function in \eqref{eq:UVM_HJB} can be accurately determined by considering only the boundary values within the 3D admissible optimal control set  $\Al^{\prime}$.
See, for example, \cite{MaForsyth2015}[Proposition~3.1] and \cite{kossaczky20192d}.
Specifically, it is established that the search for optimal control can be limited to a much smaller set $\Al$, defined as:
\EQ
\label{eq:adm_con}
\Al = \big((\{\sigma^{\myx}_{\mymin},\sigma^{\myx}_{\mymax}\}\times \Al_y)\cup (\Al_x\times \{\sigma^{\myy}_{\mymin},\sigma^{\myy}_{\mymax}\})\big)\times\{\rho_{\mymin},\rho_{\mymax}\},
\EN
where $\Al_x$ and $\Al_y$ are defined in \eqref{eq:adm_control}.
Consequently, we focus our analysis on the boundary set $\Al$, enhancing the efficiency of the proposed piecewise constant control scheme by eliminating the need to search across the entire 3D set $\Al^{\prime}$.

The aforementioned restriction presumes the existence of second-order partial derivatives, which, despite appearing restrictive, is consistent with the viscosity solution framework that utilizes smooth test functions.
Unlike traditional grid-based methods such as FD and tree-grid \cite{MaForsyth2015, kossaczky20192d}, which require discretizing the differential operator and may not always yield optimal values at $\Al$, our approach bypasses differential operator discretization.  This not only resolves related issues but also simplifies the optimization process, offering a more direct and user-friendly path to identifying optimal control values, highlighting the method's practicality and ease of implementation.
\end{remark}
}}
Let $\tau = T - t$, and we apply the change of variables $x = \ln(x^{\prime})\in(-\infty,\infty)$ and $y = \ln(y^{\prime})\in(-\infty,\infty)$. Let $\x = (x,y,\tau)$, and denote by $v(\x) \equiv v(x,y,\tau) = v^{\prime}(e^x, e^y, T -t)$. With these in mind, formulation \eqref{eq:UVM_HJB} becomes
\begin{subequations}
\label{eq:omega_inf_all}
\begin{empheq}[left={
0 = \empheqlbrace}]{alignat=8}
&\big(v_{\tau}-\sup_{\al \in \Al}   \Lcal_{\al} v\big),
\text{ or } \big(v_{\tau}-\inf_{\al \in \Al}   \Lcal_{\al} v\big),
                 &&\qquad \x \in \Rbb\times \Rbb \times (0, T],
\label{eq:omega_inf_all_a_h}
\\
&v(\x) - p\l(e^{x},e^{y}\r), &&\qquad\x \in  \Rbb \times \Rbb \times \{0\},
\label{eq:omega_inf_all_b_h}
\end{empheq}
\end{subequations}
where $(x,y,\tau) \in \Rbb \times \Rbb \times [0,T]$ and the differential operator $\Lcal_\alpha (\cdot)$
is given by
\EQA
\label{eq:LcalX}
\Lcal_{\al} v
=\frac{(\sigma_{x})^2}{2}v_{xx}  + \l(r-\frac{(\sigma_{x})^2}{2}\r)v_{x}+\frac{(\sigma_{y})^2}{2}v_{yy}+ \l(r-\frac{(\sigma_{y})^2}{2}\r)v_{y}+\rho \sigma_{x}\sigma_{y}v_{xy} - rv.
\ENA
Without loss of generality, we only consider the $\sup_{\alpha}$ problem, i.e.\ worst-case for the short position,  in the following discussion.
The theoretical analysis of this paper holds for the $\inf_{\alpha}$ problem as well.

\subsection{Localization and definition}
\label{ssec:LD}
For the problem statement and convergence analysis of numerical schemes, we define a localized two-factor uncertain volatility model pricing problem.
\begin{figure}[!ht]
\begin{minipage}{0.55\linewidth}
To this end,  with $x^{\dagger}_{\min}<x_{\min} < 0 < x_{\max}<x^{\dagger}_{\max}$, $y^{\dagger}_{\min}<y_{\min} < 0 < y_{\max}<y^{\dagger}_{\max}$, where $|x^{\dagger}_{\min}|$, $|x_{\min}|$, $|y^{\dagger}_{\min}|$,$|y_{\min}|$, $x_{\max}$, $x^{\dagger}_{\max}$, $y_{\max}$ and $y^{\dagger}_{\max}$ are chosen sufficiently large, we define the following sub-domains:
\begin{linenomath}
\begin{align}
\label{eq:sub_domain_whole}
\Omega
&= [x^{\dagger}_{\mymin}, x^{\dagger}_{\mymax}] \times[y^{\dagger}_{\mymin}, y^{\dagger}_{\mymax}]\times [0,T],
\nonumber
\\
\Omega_{\tau_0}
& = [x^{\dagger}_{\mymin}, x^{\dagger}_{\mymax}] \times[y^{\dagger}_{\mymin}, y^{\dagger}_{\mymax}]\times \{0\},
\nonumber
\\
\Omega_{\myin}  &=(x_{\mymin}, x_{\mymax}) \times(y_{\mymin}, y_{\mymax}) \times (0,T],
\\
\Omega_{\myout}
&=  \Omega\setminus\Omega_{\tau_0}\setminus\Omega_{\myin}.
\nonumber
\end{align}
\end{linenomath}
An illustration of the sub-domains for the localized problem corresponding to a fixed $\tau \in (0, T]$
is given in Figure~\ref{fig:domain}.
\end{minipage}
~~~
\hspace*{-0.5cm}
\begin{minipage}{0.45\linewidth}
\begin{center}
\begin{tikzpicture}[scale=0.45]
   \draw [dashed, line width=0.5pt] [stealth-stealth] (-4,1.5) -- (12,1.5);
   \draw [thick]  (1,5) -- (7,5);
   \draw [dashed, line width=0.5pt][stealth-stealth](3,-5) --(3,9);

   \draw [thick](7,-1) --(7,5);

   \draw [thick]  (1,-1) -- (1,5);

   \draw [thick]  (1,-1) -- (7,-1);

   \draw [thick](-2,-3.5) -- (-2,7.5);

   \draw [thick](-2,7.5) -- (10, 7.5);

   \draw [thick] (10,-3.5) -- (10,7.5);

   \draw [thick](-2,-3.5) -- (10,-3.5);
   \node at (4,3) {\scalebox{1.0}{$\Omega_{\myin}$}};
   \node at (-0.5,2.5) {\scalebox{1.0}{$\Omega_{\myout}$}};
      \node at (8.5,2.5) {\scalebox{1.0}{$\Omega_{\myout}$}};
   \node at (4.5,-2) {\scalebox{1.0}{$\Omega_{\myout}$}};
     \node at (4,6.5) {\scalebox{1.0}{$\Omega_{\myout}$}};
   \node [below] at (-3,3) {\scalebox{0.8}{$x_{\min}^{\dagger}$}};
   \node [below] at (11,3) {\scalebox{0.8}{$x_{\max}^{\dagger}$}};
   \node [right] at (3,-4) {\scalebox{0.8}{$y_{\min}^{\dagger}$}};
   \node [right] at (3,8) {\scalebox{0.8}{$y_{\max}^{\dagger}$}};
   \node [below] at (0,1.5) {\scalebox{0.8}{$x_{\min}$}};
   \node [below] at (8,1.5) {\scalebox{0.8}{$x_{\max}$}};
   \node [above] at (4,-1) {\scalebox{0.8}{$y_{\min}$}};
   \node [below] at (4,6) {\scalebox{0.8}{$y_{\max}$}};
\end{tikzpicture}
\end{center}
\vspace*{-0.5cm}
\caption{Spatial definition sub-domain at each $\tau$.}
\label{fig:domain}
\end{minipage}
\end{figure}

\noindent We now present equations for sub-domains defined in \eqref{eq:sub_domain_whole}.
\begin{itemize}
\item For $(x, y, \tau)\in \Omega_{\myin}$, we have \eqref{eq:omega_inf_all}.

\item For $(x, y, \tau)\in \Omega_{\tau_0}$,
we use the initial condition $v(x, y, 0) =  p(e^x, e^y)$.

\item For the outer boundary sub-domain $\Omega_{\myout}$,
boundary conditions are generally informed by financial reasonings or derived from the asymptotic behavior of the solution. In this study, we implement a straightforward Dirichlet condition based on discounted payoff as follows
\EQA
\label{eq:vxymax}
v(x, y, \tau) = p(e^x, e^y) e^{-r\tau}, \quad (x, y, \tau) \in \Omega_{\myout}.
\ENA
\end{itemize}
While more sophisticated boundary conditions might involve the asymptotic properties of the HJB equation \eqref{eq:omega_inf_all}) as
$z \to -\infty$ or $z \to \infty$, where $z \in \{x, y\}$, our observations indicate that these sophisticated boundary conditions do not significantly impact the accuracy of the numerical solution within  $\Omega_{\myin}$. This observation is largely due to
the so called ``cancellation property'' of the Green's function \cite{garronigreenfunctionssecond92}, which effectively mitigates the impact of approximation errors in artificial boundary condition  behavior on the solution in $\Omega_{\myin}$.
This will be  illustrated through numerical experiments in Subsection~\ref{sec:constant_pad}.

With ${\bf{x}} = (x, y, \tau)$, we let $Dv({\bf{x}}) = \l(v_{x},v_{y}, v_{\tau}\r)$ and
$D^2 v( {\bf{x}} ) = \l(v_{xx},v_{yy}, v_{xy}\r)$, and define
\EQA
\label{eq:Fomega_def}
F_{\Omega}\left({ \bf{x}},
                    v({\bf{x}}),
                    Dv({\bf{x}}),
                    D^2 v({\bf{x}})
                \right)
~=~
\left\{
\begin{array}{lllll}
    F_{\myin}\left( { \bf{x}},
                    v({\bf{x}}),
                    Dv({\bf{x}}),
                    D^2 v({\bf{x}})
                \right),
           &\qquad
           \mathbf{x} \in \Omega_{\myin},
\\
F_{\myout}\left( { \bf{x}},
                    v({\bf{x}}),
                    Dv({\bf{x}}),
                    D^2 v({\bf{x}})
                \right),
                    &\qquad
                    \mathbf{x} \in \Omega_{\myout},
\\
F_{\tau_0}\left( { \bf{x}},
                    v({\bf{x}}),
                    Dv({\bf{x}}),
                    D^2 v({\bf{x}})
                \right),
        &\qquad
    \mathbf{x} \in \Omega_{\tau_0},
\end{array}
\right.
\ENA
with operators
\EQA
F_{\myin}\left(\cdot\right) &=&
        v_\tau-\sup_{\al \in \Al}  \Lcal_{\al} v ,
\label{eq:Finn}
\\
F_{\myout}\left(\cdot\right) &=& v- p(e^x, e^y)e^{-r\tau},
\label{eq:fout}
\\
F_{\tau_0}\left( \cdot \right) &=&
       v-    p(e^x, e^y).
\label{eq:ftau0}
\ENA

\begin{definition}[Two-factor uncertain volatility pricing problem]
\label{def:uvm_def}
The  pricing problem for the two-factor uncertain volatility model is defined as
\EQA
\label{eq:uvm_def}
F_{\Omega}\left( { \bf{x}},
                    v({\bf{x}}),
                    Dv({\bf{x}}),
                    D^2 v({\bf{x}})
                \right) ~=~ 0,
\ENA
where the operator $F_{\Omega}(\cdot)$
is defined in \eqref{eq:Fomega_def}.
\end{definition}
We recall the notions of the upper semicontinuous (u.s.c.\ in short)
and the lower semicontinuous (l.s.c.\ in short) envelops of a function
$u: \mathbb{X} \rightarrow \mathbb{R}$, where $\mathbb{X}$ is a closed
subset of $\mathbb{R}^n$. They are respectively denoted by $u^*(\cdot)$
(for the u.s.c.\ envelop) and $u_*(\cdot)$ (for the l.s.c.\ envelop),
and are given by
\EQA
\label{eq:envelop}
u^*({\bf{\hat{x}}}) = \limsup_{
    \subalign{{\bf{x}} &\to {\bf{\hat{x}}}
\\
{\bf{x}}, {\bf{\hat{x}}} &\in\mathbb{X}
}}
u({\bf{x}})
\quad
(\text{resp.}
\quad
u_*({\bf{\hat{x}}}) = \liminf_{
    \subalign{{\bf{x}} &\to {\bf{\hat{x}}}
\\
{\bf{x}}, {\bf{\hat{x}}} &\in\mathbb{X}
}}
u({\bf{x}})
).
\ENA

\begin{definition}[Viscosity solution of equation \eqref{eq:uvm_def}]
\label{def:vis_def_common}
A~locally bounded function $v: \Omega \to \Rbb$ is a viscosity subsolution (resp.\ supersolution) of \eqref{eq:uvm_def} if for all test function $\phi \in  \mathcal{C}^{\infty}(\Omega)$
and for all points ${\bf{\hat{x}}} \in \Omega$ such that
$v^*-\phi$ has a \emph{global} maximum on $\Omega$ at ${\bf{\hat{x}}}$
and $v^*({\bf{\hat{x}}}) = \phi({\bf{\hat{x}}})$
(resp.\ $v_*-\phi$ has a \emph{global} minimum on $\Oinf$ at ${\bf{\hat{x}}}$
and $v_*({\bf{\hat{x}}}) = \phi({\bf{\hat{x}}})$), we have
\begin{eqnarray}
\label{eq:Def1}
\left(F_{\Omega}\right)_* \left({\bf{\hat{x}}}, \phi({\bf{\hat{x}}}), D\phi({\bf{\hat{x}}}), D^2 \phi({\bf{\hat{x}}})
             \right) &\leq  & 0,
    \\
    \big(
\text{resp.\ }
\quad
      \left(F_{\Omega}\right)^* \l(
              {\bf{\hat{x}}}, \phi({\bf{\hat{x}}}), D\phi({\bf{\hat{x}}}), D^2 \phi({\bf{\hat{x}}})
             \r) &\geq & 0,
\big)\nonumber
\end{eqnarray}
where the operator $F_{\Omega}(\cdot)$ is defined in \eqref{eq:Fomega_def}.

\end{definition}

%
\begin{remark}[Strong comparison result and convergence region]
\label{rm:strong}
Under standard conditions in viscosity-solution theory \cite{crandall1983viscosity, crandall-ishii-lions:1992, barles-souganidis:1991},
if the payoff function $p(e^{x},e^{y})$ is continuous and exhibits at most quadratic growth in $e^{x}$ and $e^{y}$, then the value function of the HJB problem \eqref{eq:omega_inf_all}---defined on the unbounded
domain $\mathbb{R}^2 \times [0, T]$---satisfies a strong comparison principle
\cite{Pham, guyon2010uncertain, MaForsyth2015}.{\imapurple{In terms of the original variables $x' = e^x$ and $y' = e^y$ used in \eqref{eq:UVM_HJB_b}, this corresponds to at most quadratic growth in $x'$ and $y'$. This condition is satisfied by all standard equity-option payoffs, such as plain-vanilla calls/puts, rainbow or basket options, spreads, and butterfly options, which grow at most linearly in the underlying prices. These include the payoffs considered in this work.}} Consequently, there is a unique continuous viscosity solution of \eqref{eq:omega_inf_all} in $\mathbb{R}^2 \times [0, T]$.

In the present paper, we focus on the finite interior sub-domain
$\Omega_{\myin}$ in \eqref{eq:uvm_def}, with Dirichlet boundary conditions on
$\Omega_{\myout}$ and initial conditions on $\Omega_{\tau_0}$.
Since a strong comparison result already holds on the original unbounded domain, it also remains valid locally within $\Omega_{\myin}$ {\imapurple{\cite[Theorem~3.1, pp.\ 372–374]{Ishii1987} and \cite[Chapter 2, Theorem~3.1, pp.\ 51-56]{Bardi1997}}}. In particular, there is a unique continuous viscosity solution of \eqref{eq:uvm_def} in $\Omega_{\myin}$.

Finally, we note that, in general, continuity of the solution across the boundary $\partial \Omega_{\myin}$ is not guaranteed, as loss of boundary data may occur when  $\tau \to 0$, $x \to \{x_{\min}, x_{\max}\}$, and $y \to \{y_{\min}, y_{\max}\}$. In all cases, the computed solution is interpreted as the limiting value approached at $\partial \Omega_{\myin}$ from the interior.
\end{remark}

\section{Numerical methods}
\label{section:num}
\subsection{Piecewise constant control}
\label{ssc:pcc}
A key component of our numerical scheme is a piecewise constant control time-stepping method applied over $\Omega_{\myin}$, which yields a set of independent linear 2D PDEs in the variables $x$ and $y$, each corresponding to a discretized control value. Unlike traditional methods that directly discretize the temporal and spatial derivatives of these PDEs, our approach avoids such discretization by leveraging Green's functions to represent each solution as a convolution integral. These PDEs are then solved using a monotone numerical integration scheme. The resulting solutions are combined using a $\max\{\cdot, \cdot\}$ operation, which preserves monotonicity and yields approximations of the value function and optimal control, thereby addressing the nonlinearity of the HJB equation.

To approximate the admissible control set with a discretized subset,  we recall from Remark~\ref{adm_con_set} that we search for the optimal control within the boundary set
$\Al$ given in \eqref{eq:adm_con}.
To this end, we first make an observation that the admissible control set $\Al$, as defined in \eqref{eq:adm_control}), is a compact set.
Therefore, it can be approximated arbitrarily well by a finite set \cite{rudin1953principles}. Specifically, for any discretization parameter $h>0$, there exists a finite partition $\Al_h$ of $\Al$ such that for any $\alpha \in \Al$, the distance to its nearest  point in $\Al_h$ is no greater than $h$. That is,
\EQ
\label{eq:compact}
\max_{\alpha \in \Al}  \min_{\alpha' \in \Al_h} \| \alpha - \alpha'\|_2 \le h.
\EN
Motivated by \eqref{eq:compact}, to address the two-factor uncertain volatility pricing problem in Defn~\eqref{def:uvm_def}, we propose an approach that involves approximating \(\Al\) with \(\Al_h\). Specifically, for $\Omega_{\myin}$,  instead of solving the HJB equation
$\ds v_\tau-\sup_{\al \in \Al}  \Lcal_{\al} v = 0$, we solve
$\ds v_{\tau} - \sup_{\al \in \Al_h} \Lcal_{\al} v = 0$.
In our convergence analysis, we will establish that, as $h \to 0$, this numerical solution converges to the viscosity solution of the pricing problem in Defn~\ref{def:uvm_def}, which is described in Defn~\ref{def:vis_def_common}.

We now elaborate the piecewise constant control for $\Omega_{\myin}$.
We let $\{\tau_m\}$, $m = 0, \ldots, M$, be an equally spaced partition in the $\tau$-dimension, where $\tau_{m} = m\Delta \tau$ and $\Delta \tau = T/M$.
With a fixed $\taus > 0$ such that $\tau_{m+1}\le T$, we consider the HJB equation
\EQ
\label{eq:hjb_pw}
v_{\tau}-\sup_{\al \in \Al_h}   \Lcal_{\al} v = 0,
\qquad (x,y,\tau) \in  \Rbb \times \Rbb \times (\tau_m, \tau_{m+1}],
\EN
where the differential operator $\Lcal_\alpha (\cdot)$ is defined in \eqref{eq:LcalX}.
Here, we note that, in \eqref{eq:hjb_pw}, the admissible control set $\Al$ is approximated by the finite discretized control set $\Al_h$,
with $h > 0$ being the discretization parameter.

For fixed $h$ and each $\alpha \in \Al_h$, we denote by $u(\cdot; \al) \equiv u(x, y, \tau; \al)$ the solution to the linear PDE in $(x, y, \tau)$ given by
\EQ
\label{eq:lpde}
 u_{\tau} - \Lcal_{\alpha} u  = 0,  \qquad (x,y,\tau) \in  \Rbb \times \Rbb \times (\tau_m, \tau_{m+1}].
 \EN
 where $\Lcal_\alpha (\cdot)$ is defined in \eqref{eq:LcalX}.
 The PDE is subject to a generic initial condition at time $\tau_m$
 given by
\EQ
\label{eq:dis_pide_term}
\hat{v}(x, y, \tau_m)
= \left\{
\begin{array}{lll}
v(x, y, \tau_m) & \qquad (x, y, \tau_{m+1}) \in \Omega_{\myin},
\\
v_{bc}(x, y, \tau_m)& \qquad (x, y, \tau_{m+1}) \in \Omega \setminus \Omega_{\myin},
\end{array}
\right.
\EN
where $v_{bc}(x, y, \tau_m)$ is the boundary conditions at time $\tau_m$ satisfying \eqref{eq:vxymax} in $\Omega_{\myout}$.

We denote by $g_\alpha(\cdot) \equiv g_\alpha(x, x', y, y'; \tau - \tau_m)$, $\alpha \in \Al_h$, the Green's function associated with the 2D linear PDE \eqref{eq:lpde} with the initial condition \eqref{eq:dis_pide_term}. Due to the spatial homogeneity of
the stochastic system \eqref{eq:2SDE_GBM}, the Green's function $g_{\alpha}(\cdot)$ simplifies  to $g_{\alpha}(x - x', y - y'; \tau - \tau_m)$. While the operator $\mathcal{L}_\alpha(\cdot)$ does not depend on $\Delta \tau$, the Green's function depends explicitly on $\Delta \tau = \tau_{m+1} - \tau_m$ when solving for the solution at $\tau = \tau_{m+1}$, resulting in $g_{\alpha}(x - x', y - y'; \Delta \tau)$.

Since we are only interested in the solution of the PDE \eqref{eq:lpde} at
$\tau = \tau_{m+1}$, for convenience, we introduce the following notational convention: unless otherwise stated, we refer to \( g_{\alpha}(x - x', y - y'; \Delta \tau) \) as the Green's function associated with the 2D linear PDE \eqref{eq:lpde} and the initial condition \eqref{eq:dis_pide_term}, reflecting the explicit dependence of the Green's function on $\Delta \tau$.

By the Green's function approach \cite{garronigreenfunctionssecond92, Duffy2015},
for fixed $\alpha \in \Al_h$,  the solution $u(x,y,\tau_{m+1})$ for $(x,y)\in \mathbf{D}$, where
\[
\mathbf{D}\equiv (x_{\min}, x_{\max}) \times (y_{\min}, y_{\max}),
\]
can be represented as the convolution integral of the Green's function $g_\alpha(\cdot;\Delta \tau)$ and the initial condition $\hat{v}(\cdot, \tau_m)$ as follows
\EQ
u(x,y,\tau_{m+1}; \al)= \iint_{\Rbb^2}g_{\alpha}\left(x- x',y- y';   \Delta \tau\right)\hat{v}( x',y', \taus) dx'dy',\qquad (x,y) \in\mathbf{D}, \quad
\alpha \in \Al_h.
\label{eq:bkinteg}
\EN
The solution $u(x,y,\tau_{m+1}; \al)$ for $(x, y)\notin \mathbf{D}$ are given by the boundary condition \eqref{eq:vxymax}.

For computational purposes, we truncate the infinite region of integration of \eqref{eq:bkinteg}
to $\mathbf{D}^{\dagger}$, where
\EQ
\label{eq:truncate_region}
\mathbf{D}^{\dagger}\equiv [x_{\min}^{\dagger}, x_{\max}^{\dagger}] \times [y_{\min}^{\dagger}, y_{\max}^{\dagger}].
\EN
Here, recall that  $z\in\l\{x,y\r\}$, $z^{\dagger}_{\min}< z_{\min}<0<z_{\max}< z^{\dagger}_{\max}$ and $|z^{\dagger}_{\min}|$ and $z^{\dagger}_{\max}$ are sufficiently large. This results in the approximation
\EQ
\label{eq:integral_truncated}
u(x,y, \tau_{m+1}; \al) \simeq \iint_{\mathbf{D}^{\dagger}}g_{\alpha}\left(x- x',y- y';   \Delta \tau\right)\hat{v}( x',y', \taus) dx'dy',\qquad (x,y) \in\mathbf{D},
\quad
\alpha \in \Al_h.
\EN
Finally, an approximation to the solution of the HJB~\eqref{eq:hjb_pw} for $(x, y, \tau_{m+1}) \in \Omega_{\myin}$ is given by
\EQ
\label{eq:pcc_v}
v(x,y, \tau_{m+1}) \simeq \max_{\alpha\in\Al_h}u(x,y, \tau_{m+1}; \al),\qquad (x,y) \in\mathbf{D}.
\EN
We conclude by noting that the errors arising from (i) approximating $\Al$ by $\Al_h$ and (ii) from truncating the infinite
integration domain in \eqref{eq:bkinteg} to a finite one in \eqref{eq:integral_truncated} are discussed subsequently.

\subsection{A closed-form representation of $\boldsymbol{g_{\al}\l(\cdot\r)}$ for $\boldsymbol{\Omega_{\myin}}$}
We now present a closed-form expression for the Green's function $g_{\al}(\cdot)$ of the linear PDE \eqref{eq:lpde}, where the control $\alpha\equiv(\sigma_x,\sigma_y,\rho) \in \Al$ is fixed.
To this end, we denote by $G_{\alpha}(\cdot;\Delta \tau)$ the Fourier transform of $g_\alpha(\cdot;\Delta \tau)$, i.e.\
\EQ
\label{eq:ft_pair}
 \left\{
\begin{array}{lll}
\mathscr{F}|g_{\alpha}(x,y;\cdot)| &=G_{\alpha}(\eta,\zeta;\cdot)&=\displaystyle\iint_{\Rbb^2}e^{-i(\eta x+\zeta y)}g_{\alpha}(x,y;\cdot)dxdy,
\\
\mathscr{F}^{-1}|G_{\alpha}(\eta,\zeta;\cdot)|&=g_{\alpha}(x,y;\cdot)&=\frac{1}{(2\pi)^2}\displaystyle\iint_{\Rbb^2}e^{i(\eta x+\zeta y)}G_{\alpha}(\eta,\zeta;\cdot)d\eta d\zeta.
\end{array}
\right.
\EN
{\imapurple{A closed-form expression for $G_{\alpha}(\eta,\zeta;\cdot)$ is given
in Equation~(6.4) of \cite{ruijter2012two}:}}
\EQA
\label{eq:G_closed}
&&G_{\alpha}(\eta,\zeta;\cdot)=\exp(\Psi(\eta,\zeta) \Delta \tau),
\\
&&\qquad \text{ with } \Psi(\eta,\zeta)= \bigg(-\frac{\sigma_{x}^{2}\eta^2}{2}-\frac{\sigma_{y}^{2}\zeta^2}{2}+(r-\frac{\sigma_x^2}{2})i\eta+(r-\frac{\sigma_y^2}{2})i\zeta-\rho\sigma_x\sigma_y\eta\zeta-r\bigg).
\nonumber
\ENA
\noindent {\dangblue{We now introduce a lemma providing a closed-form expression for the Green's function $g_{\alpha}(x,y;\Delta \tau)$.
\begin{lemma}
\label{lem:g_k}
Let $\Delta \tau > 0$ be fixed,
and  $g_{\alpha}(x,y;\Delta \tau)$ and $G_{\alpha}(\eta,\zeta;\Delta \tau)$ be a Fourier transform pair
defined in \eqref{eq:ft_pair}, and $G_{\alpha}(\eta,\zeta;\Delta \tau)$ be given in \eqref{eq:G_closed}.
When $|\rho|< 1$, $g_{\al}(x,y;\Delta \tau)$ can be expressed in the form of a ``scaled'' joint density as follows
\begin{align}
\label{eq:g_k}
&g_{\alpha}(x,y;\Delta \tau) =  e^{-r\Delta\tau} f_{\al}(x,y; \Delta \tau), \text{ where }
\\
&f_{\al}(x,y;\Delta \tau)= \frac{1}{2 \pi  \kappa_x \kappa_y   \sqrt{1-\rho^2}}\exp
        \bigg( \frac{-1}{2(1 - \rho^2)}\bigg[
          \big(\frac{x-\mu_x}{\kappa_x}\big)^2 -
          2\rho\big(\frac{x - \mu_x}{\kappa_x}\big)\big(\frac{y - \mu_y}{\kappa_y}\big) +
          \big(\frac{y - \mu_y}{\kappa_y}\big)^2
        \bigg]
       \bigg).
\nonumber
\\
 & \text{ with }  \quad
\mu_x = \bigg(\frac{\sigma_x^2}{2}-r\bigg)\Delta\tau, ~\kappa_x = \sigma_x \sqrt{\Delta \tau},~
\mu_y = \bigg(\frac{\sigma_y^2}{2}-r\bigg)\Delta\tau, ~ \kappa_y = \sigma_y \sqrt{\Delta \tau},
\label{eq:cond_den_o}
\end{align}

When $\rho = \pm 1$, $g_{\al}(x,y;\Delta \tau)$ is given by
\begin{align}
\label{eq:g_k_rho}
g_{\alpha}(x,y;\Delta \tau) =  e^{-r\Delta\tau} \frac{1}{\sqrt{2\pi} \kappa_x} \exp\left( -\frac{(x - \mu_x)^2}{2\kappa_x^2} \right)
\delta(y - (a + \rho bx)).
\end{align}
{\dangblue{Here, $\delta(\cdot)$ is a Dirac delta function,
and $a = \mu_y - \rho b \mu_x$ with $b = \frac{\sigma_y}{\sigma_x}$.
}}
\end{lemma}
\begin{proof}[Proof of Lemma~\ref{lem:g_k}]
When $|\rho| < 1$, applying inverse Fourier transform to $G_\al(\cdot)$, provided in \eqref{eq:G_closed},
we obtain  the expression for the Green's function $g_{\al}(x,y;\Delta \tau)$ given in \eqref{eq:g_k}.
When $\rho = \pm 1$, $f_{\al}(x,y;\Delta \tau)$ can be expressed in the form
$f_{\al}(x,y; \Delta \tau) =   \frac{1}{\sqrt{2\pi} \kappa_x} \exp\left( -\frac{(x - \mu_x)^2}{2\kappa_x^2} \right) ~\delta(y - (a + \rho bx))$,
where $a$ and $b$ are constants, with $b > 0$ \cite{glazunov2012note}.
We then solve for $a$ and $b$ by comparing the Fourier transform of $g_{\al}(x,y; \Delta \tau)$ in this case
with the closed-form expression of $G_\al(\cdot)$. This gives $a = \mu_y - \rho b \mu_x$ and $b = \frac{\sigma_y}{\sigma_x}$.
This completes the proof.
\end{proof}
\begin{remark}[$\rho= \pm 1$]
\label{rm:rho}
In our study, while we acknowledge the theoretical significance of the cases where
$\rho= \pm 1$, we have chosen not to explore this scenario in depth.
Such extreme correlation values, though mathematically interesting, are rarely encountered in practical applications
and financial modeling. Therefore, our focus remains predominantly on scenarios where the correlation coefficient lies strictly between $-1$ and $1$, which are more representative of the conditions commonly observed and of
greater relevance to practitioners. However, it is important to note that our piecewise constant control integration scheme can effectively manage the special case of $\rho= \pm 1$.

For computational purposes, approximating the Dirac delta function $\delta(y - (a \pm bx))$ in \eqref{eq:g_k_rho} by a suitable Gaussian function is necessary (refer to \cite{spanier1988atlas}[Chapter 10], for example, for more details of such approximations). Specifically, we approximate $\delta(\cdot)$ using a Gaussian $\delta_{\rhoh}(\cdot)$ with $\rhoh \to \pm 1^{\mp}$. Essential aspects of our scheme in this case are elaborated in Appendix~\ref{app:rho_ex}.
\end{remark}
}}
We now present a lemma on the boundary truncation error of the Green's function $g_{\al}(\cdot)$ defined in \eqref{eq:g_k} for the case $|\rho| < 1$.
\begin{lemma}
\label{lemma:truncation}
Suppose that $|\rho| < 1$, and let $\Delta \tau > 0$ be fixed.
Furthermore, suppose $|x^{\dagger}_{\mymin}|$, $x^{\dagger}_{\mymax}$, $|y^{\dagger}_{\mymin}|$, and $y^{\dagger}_{\mymax}$ are chosen such that
\[
\min\left\{|x^{\dagger}_{\mymin}|,~x^{\dagger}_{\mymax},~|y^{\dagger}_{\mymin}|,~y^{\dagger}_{\mymax}\right\}
> \max \left\{ \mu_x \pm \gamma, \mu_y \pm \gamma \right\},
\]
where $\gamma \gg 0$ is a fixed constant, and both $\mu_x = \left(\frac{\sigma_x^2}{2}-r\right)\Delta\tau$ and $\mu_y = \left(\frac{\sigma_y^2}{2}-r\right)\Delta\tau$
are defined in \eqref{eq:cond_den_o}. Then, for sufficiently small $\Delta \tau$, $g_{\al}(\cdot)$, as defined in \eqref{eq:g_k}
for a fixed  $\al \in \Al$, satisfies
\EQ
\label{eq:boundgn}
\iint_{\Rbb^2\setminus \mathbf{D}^{\dagger}}g_{\alpha}\left(x,y;\Delta \tau\right) dxdy
<
C \Delta \tau e^{-\frac{1}{2\Delta \tau}},
\quad
\mathbf{D}^{\dagger}\equiv [x_{\min}^{\dagger}, x_{\max}^{\dagger}] \times [y_{\min}^{\dagger}, y_{\max}^{\dagger}],
\EN
where $C$ is a bounded constant independently of $\Delta \tau$.
\end{lemma}
\begin{proof}[Proof of Lemma~\ref{lemma:truncation}]
Without loss of generality, we present a proof for the case $0\leq\rho<1$.
For subsequent use, let $\Phi(s) \equiv \int_{-\infty}^{s}\phi(z)dz$ and $\phi(z) \equiv (2\pi)^{-1/2}\exp(-z^2/2)$ respectively be the CDF and the probability density function of standard normal distribution.

For simplicity, we let
$w= \min\left\{|x^{\dagger}_{\mymin}|,~x^{\dagger}_{\mymax},~|y^{\dagger}_{\mymin}|,~y^{\dagger}_{\mymax}\right\}$.
With $z_x = \frac{x-\mu_x}{\kappa_x}$, $z_y= \frac{y-\mu_y}{\kappa_y}$,
we define the region $\mathbf{B}$ as follows
\EQ
\label{eq:regB}
\mathbf{B} =  [-b, b]\times [-b, b], \quad \text{where }
b = \min \big\{ \big|\frac{-w-\mu_x}{\kappa_x}\big|, \frac{w-\mu_x}{\kappa_x},
\big|\frac{-w-\mu_y}{\kappa_y}\big|, \frac{w-\mu_y}{\kappa_y}\big\}.
\EN
We have
\begin{align}
\label{eq:boundgd}
&\iint_{\Rbb^2\setminus \mathbf{D}^{\dagger}}g_{\alpha}\left(x,y;\Delta \tau\right) dxdy
\le e^{-r\Delta\tau}\iint_{\Rbb^2\setminus \mathbf{B}} \frac{\exp
        \left( -\frac{1}{2\left[1 - \rho^2\right]}\left[
          z_x^2 -
          2\rho z_x z_y +
         z_y^2
        \right]
       \right)}{2 \pi    \sqrt{1-\rho^2}} dz_xdz_y
\nonumber
\\
&\leq 2e^{-r\Delta\tau}P\big(z_x\geq b,~z_y\geq b\big)
\overset{\text{(i)}}{\leq} 2e^{-r\Delta\tau}(1+\rho)\Phi(-b)\Phi\bigg(\frac{-b(1-\rho)}{\sqrt{1-\rho^2}}\bigg)
\nonumber\\
&\overset{\text{(ii)}}{\leq }
\frac{e^{-r\Delta \tau}(1+\rho)^{3/2}}{\pi(1-\rho)^{1/2}} \times \frac{e^{-b^2/2}}{b^2}.
\end{align}
Here, (i) is due to an upper bound for the bivariate normal distribution in \cite{willink2005bounds};
in (ii), we apply the following fact: if $X\sim N(0,1)$, then $P(X>x)\leq \frac{1}{x\sqrt{2\pi}}\exp(-x^2/2)$.
{\dangblue{It is straightforward to see that $\frac{e^{-r\Delta \tau}(1+\rho)^{3/2}}{\pi(1-\rho)^{1/2}} \le \frac{ (1+\rho_{\max})^{3/2}}{\pi (1-\rho_{\max})^{1/2}}$.}}
Thus, for sufficiently small $\Delta \tau$, the condition  $w > \max \left\{ \mu_x \pm \gamma, \mu_y \pm \gamma \right\}$, where $\gamma \gg 0$ is fixed, implies the rhs of \eqref{eq:boundgd} is
bounded by $C \Delta \tau e^{-1/2\Delta \tau}$, where $C$ is a bounded constant independently of $\Delta \tau$.
This completes the proof.
\end{proof}

%

\begin{remark}[Boundary truncation error]
\label{rm:bd_e}
For the case $|\rho|<1$,  the boundary truncation error upper bound, as detailed in \eqref{eq:boundgd}, serves as a practical tool for selecting an appropriate definition domain, $\mathbf{D}^{\dagger}$, to ensure this truncation error remains below a predefined threshold $\epsilon>0$. To achieve this, we first identify a value of $b$ satisfying
\EQ
\label{eq:eps_b}
\frac{e^{-r\Delta \tau}(1+\rho)^{3/2}}{\pi(1-\rho)^{1/2}} \times \frac{e^{-b^2/2}}{b^2} \le
\frac{ (1+\rho_{\max})^{3/2}}{\pi (1-\rho_{\max})^{1/2}} \frac{e^{-b^2/2}}{b^2} < \epsilon, \quad  |\rho|<1.
\EN
Given  $b$, we then determine $w$ through equation \eqref{eq:regB})
by ensuring the following conditions are met:
$b \leq \big|\frac{-w-\mu_x}{\kappa_x}\big|$, $b \leq \frac{w-\mu_x}{\kappa_x}$, $b \leq \big|\frac{-w-\mu_y}{\kappa_y}\big|$ and $b \leq \frac{w-\mu_y}{\kappa_y}$.
Subsequently, $\mathbf{D}^{\dagger}$ is derived via
$w = \min\left\{|x^{\dagger}_{\mymin}|,~x^{\dagger}_{\mymax},~|y^{\dagger}_{\mymin}|,~y^{\dagger}_{\mymax}\right\}$.

It is worth noting that $\mu_x$ and $\mu_y$, as defined in \eqref{eq:cond_den_o},
scale linearly with $\Delta\tau$. Therefore, if
\[
w = \min\left\{|x^{\dagger}_{\mymin}|,~x^{\dagger}_{\mymax},~|y^{\dagger}_{\mymin}|,~y^{\dagger}_{\mymax}\right\}
 > \gamma + \max\{\,|\mu_x|,|\mu_y|\}
\]
is satisfied for some $\Delta \tau$, then $w > \max\{\,\mu_x \pm \gamma,\,\mu_y \pm \gamma\}$ holds for all smaller $\Delta\tau$. As a result, $\mathbf{D}^{\dagger}$ remains sufficiently large to ensure that the truncation error stays below the threshold $\epsilon$ without adjusting $w$ as $\Delta\tau$ is refined toward zero.

For the special case $\rho = \pm 1$, the upper bound~\eqref{eq:boundgd} degenerates due to division by $\sqrt{1-\rho^2}$ and is no longer valid. As discussed in Remark~\ref{rm:rho}, the Green’s function $ g_{\alpha}(x,y;\Delta\tau)$ in this scenario is given by \eqref{eq:g_k_rho}, which involves a Dirac delta function $ \delta\bigl(y - (a + \rho b\,x)\bigr)$. To handle this computationally, we use the Gaussian approximation $\delta_{\rhoh}(\cdot)$ with $ \rhoh \to \pm 1^{\mp} $ as described in Remark~\ref{rm:rho}. Consequently, the boundary truncation strategy outlined above remains applicable
(see Appendix~\ref{app:app_del}).

The methodological approach outlined above represents
a significant advantage over traditional finite difference methods,
which typically depend on heuristic strategies or trial-and-error for determining appropriate domain sizes. Our approach introduces a systematic and quantifiable method for determining domain size, significantly enhancing the accuracy and reliability of numerical solutions.
The efficacy of this systematic approach is demonstrated through numerical experiments detailed in Subsection~\ref{ssc:padding}.
\end{remark}

\subsection{Discretization}
We highlight that,  in approximating the 2D convolution integral \eqref{eq:integral_truncated} over the  finite integration domain $\mathbf{D}^{\dagger}$, it is necessary to obtain values of the Green's function $g_{\al}(x, y;\cdot)$,
at points $(x, y)$ outside $\mathbf{D}^{\dagger}$. To define these points, we let $z_{\mymax}^{\ddagger} = z_{\mymax} - z_{\mymin}^{\dagger}$ and $z_{\mymin}^{\ddagger} = z_{\mymin} - z_{\mymax}^{\dagger}$ for $z \in \{ x, y \}$.
Consequently, we need $g_{\al}(x, y;\cdot)$ at $(x, y) \in \mathbf{D}_{\myout}^{\dagger}$, where
\EQ
\label{eq:extra_dom}
\mathbf{D}^{\dagger}_{\myout} = \big([x_{\mymin}^{\ddagger}, x_{\mymax}^{\ddagger}] \times [y_{\mymin}^{\ddagger}, y_{\mymax}^{\ddagger}]\big) \setminus \mathbf{D}^{\dagger}, \qquad  z_{\mymin}^{\ddagger} = z_{\mymin} - z_{\mymax}^{\dagger} \text{ for } z \in \{x, y\}.
\EN
Although $\mathbf{D}^{\dagger}_{\myout}$ lies outside the pricing problem's definition domain, the availability of a closed-form expression for $g_{\al}(x, y;\cdot)$ ensures no issues for our numerical methods. Moreover, the value functions for $(x, y) \in \mathbf{D}_{\myout}^{\dagger}$ are not required for our convergence analysis. The role of $\mathbf{D}_{\myout}^{\dagger}$ is to ensure the well-definedness of an associated Green's function for the convolution integral, which is crucial for time advancement within $\Omega{\myin}$.

Without loss of generality, for convenience, we assume that  $|z_{\min}|$ and $z_{\max}$, where $z\in\l\{x,y\r\}$, are chosen sufficiently large so that
\EQA
\label{eq:w_choice_green_jump_form}
z^{\dagger}_{\min} = z_{\min} - \frac{z_{\max} - z_{\min}}{2},
~~~\text{and}~~~
z^{\dagger}_{\max} =  z_{\max} + \frac{z_{\max} - z_{\min}}{2}.
\ENA
With \eqref{eq:w_choice_green_jump_form} in mind, recalling $z_{\min}^{\ddagger}$ and $z_{\max}^{\ddagger}$, $z\in\{x,y\}$ as defined
in \eqref{eq:extra_dom} gives
\EQ
\label{eq:w_choice_green_jump_form_dd}
z_{\min}^{\ddagger} = z^{\dagger}_{\min} - z_{\max} =  -\frac{3}{2}\l(z_{\max} - z_{\min}\r),
~~~\text{and}~~~
z_{\max}^{\ddagger} = z^{\dagger}_{\max} - z_{\min}  = \frac{3}{2}\l(z_{\max} - z_{\min}\r).
\EN

We denote by $N$ (resp.\ $N^{\dagger}$ and $N^{\ddagger}$ ) the number of intervals of a uniform partition of $[x_{\mymin}, x_{\mymax}]$
(resp.\ $[x_{\mymin}^{\dagger}, x_{\mymax}^{\dagger}]$ and $[x_{\mymin}^{\ddagger}, x_{\mymax}^{\ddagger}]$).
For convenience, we typically choose $N^{\dagger} = 2N$ and $N^{\ddagger} = 3N$ so that only one set of $x$-coordinates is needed. Also, let $P_{x} = x_{\mymax} - x_{\mymin}$, $P_x^{\dagger} = x^{\dagger}_{\mymax} - x^{\dagger}_{\mymin}$, and $P_x^{\ddagger} = x^{\ddagger}_{\mymax} - x^{\ddagger}_{\mymin}$. We define $\Delta x = \frac{P_x}{N} = \frac{P_x^{\dagger}}{N^{\dagger}} =\frac{P_x^{\ddagger}}{N^{\ddagger}}$. We use an equally spaced partition in the $x$-direction, denoted by $\{x_{n}\}$, and is defined as follows
\EQA
\label{eq:grid_x}
    x_{n}& = &{\hat{x}}_{0} + n\Delta x;
    ~~
    n = -N^{\ddagger}/2, \ldots, N^{\ddagger}/2, ~~\text{where}~~
    \nonumber
    \\
    \Delta x& =& P_x/N ~=~ P_x^{\dagger}/N^{\dagger}=P_x^{\ddagger}/N^{\ddagger}, ~~\text{and}~~
    \\
    \nonumber
    \hat{x}_{0} &~=~& ( x_{\mymin} + x_{\mymax})/2 ~=~ (x^{\dagger}_{\mymin} + x^{\dagger}_{\mymax})/2~=~ (x^{\ddagger}_{\mymin} + x^{\ddagger}_{\mymax})/2.
    \nonumber
\ENA
Similarly, for the $y$-dimension, with  $J^{\dagger} = 2J$, $J^{\ddagger} = 3J$, $P_{y} = y_{\mymax} - y_{\mymin}$, $P_y^{\dagger} = y^{\dagger}_{\mymax} - y^{\dagger}_{\mymin}$, and $P_y^{\ddagger} = y^{\ddagger}_{\mymax} - y^{\ddagger}_{\mymin}$, we denote by $\{y_{j}\}$, an equally spaced partition in the $y$-direction defined as follows
\EQA
\label{eq:grid_y}
    y_{j}& = &{\hat{y}}_{0} + j\Delta y;
    ~~
    j = -J^{\ddagger}/2, \ldots, J^{\ddagger}/2, ~~\text{where}~~
    \nonumber
    \\
    \Delta y& =& P_y/J ~=~ P_y^{\dagger}/J^{\dagger}=P_y^{\ddagger}/J^{\ddagger}, ~~\text{and}~~
    \\
    \nonumber
    \hat{y}_{0} &~=~& ( y_{\mymin} + y_{\mymax})/2 ~=~ (y^{\dagger}_{\mymin} + y^{\dagger}_{\mymax})/2~=~ (y^{\ddagger}_{\mymin} + y^{\ddagger}_{\mymax})/2.
    \nonumber
\ENA
We use the same previously defined uniform partition
$\{\tau_m\}$, $m = 0, \ldots, M$, with $\tau_{m} = m\Delta \tau$ and $\Delta \tau = T/M$.\footnote{While it is straightforward to generalize the numerical method to non-uniform partitioning of the $\tau$-dimension, to prove convergence, uniform partitioning suffices.}


Regarding the control set $\Al$, defined in \eqref{eq:adm_con},
we let $Q_x$ and $Q_y$ respectively be the number of intervals of a uniform partition of $\Al_x =
\l[\sigma^{\myx}_{\mymin},\sigma^{\myx}_{\mymax}\r]$ and $\Al_y = \l[\sigma^{\myy}_{\mymin},\sigma^{\myy}_{\mymax}\r]$.
We denote by $\{\sigma^x_q\}$ and $\{\sigma^y_{q'}\}$ an equally spaced partition for
$\Al_x$ and $\Al_y$, respectively, each with a uniform interval length
$\Delta \sigma_z = \frac{\sigma^{z}_{\mymax}-\sigma_{\mymin}^z}{Q_z}$, where $z\in\{x,y\}$.
Consequently, the discretized  control set $\Al_h$ approximating $\Al$ is given by
\EQA
\Al_h = \l\{ \l(\{\sigma^x_{\mymin}, \sigma^x_{\mymax}\}\times
\{\sigma^y_{q'}\}\r)\cup \l(\{\sigma^x_q\}\times\{\sigma^y_{\mymin}, \sigma^y_{\mymax}\}\r)\r\}\times\{\rho_{\mymin},\rho_{\mymax}\}.
\ENA
For subsequent use, we denote by $Q$ the cardinality of the set $\Al_h$, assuming that both $\Al_x$ and $\Al_y$ are discretized using the same number of partitions.


As is common in the literature \cite{MaForsyth2015, chen08a}[Equation 4.1],
we introduce a single mesh-discretization parameter $h > 0$ to control
the refinement of temporal, spatial, and control discretizations
simultaneously, as specified in \eqref{eq:dis_parameter} below:
\EQA
\label{eq:dis_parameter}
\Delta x=  C_1 h, \quad
\Delta y = C_2 h,\quad
\Delta \tau = C_3 h,\quad
\Delta \sigma_x = C_4 h,
\quad
\Delta \sigma_y = C_5 h,
\ENA
where the positive constants $C_1$, $C_2$, $C_3$, $C_4$, and $C_5$ are independent of $h$.
\begin{remark}[Mesh discretization parameter $h>0$]
The assumption \eqref{eq:dis_parameter} provides a unified framework for analyzing
convergence by tying the different discretizations--in space, time, and control--to the
single parameter $h$. {\imapurple{Unlike the classical explicit-scheme CFL condition \cite{courant1928partiellen, courant1967partial}, which typically imposes
$\Delta \tau \sim \mathcal{O}\bigl(\max(\Delta x,\Delta y)^2\bigr)$ for stability,
the relation \eqref{eq:dis_parameter} merely scales
$\Delta x$, $\Delta y$, $\Delta \tau$, $\Delta\sigma_x$, and $\Delta\sigma_y$ 
linearly with $h$ through constants $C_1,\dots,C_5$, all of which are
independent of~$h$.}}
%

A key consequence is that the discretization of the control set $\mathcal{A}$ is also
refined in a controlled manner (since $\Delta \sigma_x$ and $\Delta \sigma_y$ shrink
with $h$). Hence, the resulting discretized set $\mathcal{A}_h$ satisfies the approximation
bound~\eqref{eq:compact}. This practice is standard in numerical methods for HJB equations
as noted earlier.

{\imapurple{Finally, as shown in Subsection~\ref{ssc:stability}, the proposed method is $\ell_{\infty}$-stable under \eqref{eq:dis_parameter}; no additional explicit-scheme CFL restriction of the form $\Delta\tau = \mathcal{O}\bigl(\max(\Delta x,\Delta y)^2\bigr)$ is required.}}
\end{remark}
For convenience, we let $\mathbb{M}=\l\{0, \ldots M-1\r\}$ and we also define the following index sets:
\begin{linenomath}
\postdisplaypenalty=0
\begin{alignat}{8}
\label{index_sets}
&\mathbb{N} &&= \l\{-N/2+1, \ldots N/2-1\r\}, \quad && \mathbb{N}^{\dagger} &&= \l\{-N, \ldots N\r\}, \quad &&&\mathbb{N}^{\ddagger} &&= \l\{-3N/2+1, \ldots 3N/2-1\r\},
\nonumber
\\
&\mathbb{J} &&= \l\{-J/2+1, \ldots J/2-1\r\}, \quad &&\mathbb{J}^{\dagger} &&= \l\{-J, \ldots J\r\}, \quad &&&\mathbb{J}^{\ddagger} &&= \l\{-3J/2+1, \ldots, 3J/2-1 \r\}.
\end{alignat}
\end{linenomath}
With $n \in \mathbb{N}^{\dagger}$, $j \in \mathbb{J}^{\dagger}$, and $m \in \{0, \ldots, M\}$,
we denote by $v_{n, j}^m$ (resp.\ $u_{n, j}^m$)
a numerical approximation to the exact solution $v(x_n, y_j, \tau_m)$ (resp.\ $u(x_n, y_j, \tau_m)$) at the reference node $(x_n, y_j, \tau_m) = {\bf{x}}_{n, j}^{m}$. We also denote by $(\al^*)_{n, j}^{m} \equiv (\sigma_{\myx}^*, \sigma_{\myy}^*,\rho^*)_{n, j}^{m}$ the optimal control obtained by a numerical method for this reference node.
For $m \in \mathbb{M}$, nodes ${\bf{x}}_{n, j}^{m+1}$ having
(i) $n \in \mathbb{N} \text{ and } j \in \mathbb{J}$, are in $\Omega_{\myin}$,
(ii) either $n \in \mathbb{N}^{\dagger}\setminus \mathbb{N}$ and $j \in \mathbb{J}^{\dagger}$ or
$n \in \mathbb{N}^{\dagger}$ and $j \in \mathbb{J}^{\dagger}\setminus \mathbb{J}$ are in $\Omega_{{\myout}}$.
For double summation, unless otherwise noted, we adopt the short-hand notation:  $\mysum_{d\in\mathbb{D}}^{q\in \mathbb{Q}}(\cdot):=\sum_{q\in\mathbb{Q}}\sum_{d\in\mathbb{D}}(\cdot)$.
Lastly, it's important to note that references to indices
$n \in \mathbb{N}^{\ddagger}\setminus \mathbb{N}^{\dagger}$ or $j \in \mathbb{J}^{\ddagger}\setminus \mathbb{J}^{\dagger}$
pertain to points within $\mathbf{D}^{\dagger}_{\myout}$ (as defined in \eqref{eq:extra_dom}).
As noted earlier, no numerical solutions are required for these points.

\subsection{Numerical schemes}
\label{sec:NS}
\subsubsection{Constructions of the scheme}
For $(x_{n},y_{j},\tau_0)\in \Omega_{\tau_0}$, we impose the initial condition \eqref{eq:ftau0} by
\EQA
\label{eq:tau0i}
v_{n,j}^0 &=& p(e^{x_{n}}, e^{y_j}),
\quad n \in \mathbb{N}^{\dagger}  \text{ and }  j \in \mathbb{J}^{\dagger}.
\ENA
For $(x_{n},y_{j},\tau_{m+1})\in \Omega_{{\myout}}$, we impose the boundary condition \eqref{eq:fout} as follow
\EQA
\label{eq:outi}
v_{n,j}^{m+1} = p(e^{x_n}, e^{y_j}) e^{-r\tau_{m+1}},
\quad
n \in \mathbb{N}^{\dagger}\setminus \mathbb{N} \text{ or } j \in \mathbb{J}^{\dagger}\setminus \mathbb{J}.
\ENA

For $(x_{n},y_{j},\tau_{m+1})\in \Omega_{{\myin}}$,
let $g^\alpha_{n-l,j-d}\equiv g_{\alpha}\l(x_n-x_l,y_j-y_d; \Delta\tau\r)$ with $n\in\mathbb{N}$,
$j\in \mathbb{J}$, $l\in \mathbb{N}^{\dagger}$ and $d\in\mathbb{J}^{\dagger}$.
Here, $g_{\alpha}(\cdot)$ is given by the closed-form expression in \eqref{eq:g_k} in Lemma~\ref{lem:g_k}, where $\al \in \Al_h$ is fixed.
When the role of $\Delta \tau$ is important, we explicitly write
$g^\alpha_{n-l,j-d}(\Delta \tau)$.

We let $u_{n,j}^{m+1, \al}$ be an approximation to the double integral \eqref{eq:integral_truncated}
at $x = x_{n}$, $y= y_{j}$ and $\tau_{m+1})$ obtained via a 2D composite quadrature rule.
It is computed by
\EQ
\label{eq:dissum}
u_{n,j}^{m+1, \al} = \Delta x \Delta y \mysum_{l\in\mathbb{N}^{\dagger}}^{d\in\mathbb{J}^{\dagger}}\varphi_{l,d}~g^\alpha_{n-l,j-d}~v^{m}_{l,d},
\quad n\in\mathbb{N} \text{ and } j\in \mathbb{J}.
\EN
Here, the coefficients $\varphi_{l,d}$ in \eqref{eq:dissum} are the weights of the composite quadrature rule.
Finally, $v_{n,j}^{m+1}$ is computed as follow
\EQ
\label{eq:vmax}
v_{n,j}^{m+1} =\max_{\alpha\in\Al_h}u_{n,j}^{m+1, \al}=\max_{\al\in\Al_h} \bigg\{\Delta x \Delta y \mysum_{l\in\mathbb{N}^{\dagger}}^{d\in\mathbb{J}^{\dagger}}\varphi_{l,d}~g^\alpha_{n-l,j-d}~v^{m}_{l,d}\bigg\},
\quad
n\in\mathbb{N} \text{ and } j\in \mathbb{J}.
\EN
By solving the optimization problem \eqref{eq:vmax},
we obtain the optimal control  $(\al^*)_{n, j}^{m+1} \equiv (\sigma_{\myx}^*, \sigma_{\myy}^*,\rho^*)_{n, j}^{m+1}$,
where
\EQ
\label{eq:optimal_control}
(\alpha^*)_{n, j}^{m+1} = \arg\max_{\alpha \in \mathcal{A}h} u_{n,j}^{m+1, \alpha}.
\EN
Unless otherwise stated,  2D composite trapezoidal quadrature
rule is used.

\begin{remark}[Rescaled weights and convention]
\label{rem:rescaled_kernel}
In the scheme \eqref{eq:dissum}, the weights $g_{n-l,j-d}^{\alpha}(\Delta \tau)$ are multiplied by the grid area $\Delta x\,\Delta y$. As $\Delta \tau \to 0$, the Green's function
$g_{\alpha}(\cdot, \Delta \tau)$ approaches a Dirac delta function, becoming increasingly peaked and unbounded. However, once $\Delta x ,\Delta y$ absorbed into $g_{n-l,j-d}^{\alpha}(\Delta \tau)$, a direct verification using the closed-form expression
for $g_{n-l,j-d}^{\alpha}(\Delta \tau)$ in \eqref{eq:g_k} confirms that the rescaled weights remain bounded.

To formalize this, we define the rescaled weights of our scheme as follows:
\begin{equation}
\label{eq:widetilde_g}
  \widetilde{g}_{n-l,j-d}^{\alpha}(\Delta \tau)
  \;:=\;
  \Delta x\,\Delta y\,  \odot g_{n-l,j-d}^{\alpha}(\Delta \tau), \quad
  n \in \mathbb{N},~ l \in \mathbb{N}^{\dagger},~
j \in \mathbb{J}, \text{and } d \in \mathbb{J}^{\dagger}.
\end{equation}
Here, $\odot$ indicates that $\Delta x\,\Delta y$ is absorbed into $g_{n-l,j-d}^{\alpha}(\Delta \tau)$, ensuring that $\widetilde{g}_{n-l,j-d}^{\alpha}(\Delta\tau)$ remains bounded as $\Delta \tau \to 0$.

\medskip
\noindent
\underline{Convention:}
For the rest of the paper, we adopt the convention of continuing to write
$\Delta x\,\Delta y\, g_{n-l,j-d}^{\alpha}(\Delta \tau)$
and
$\Delta x\,\Delta y\,
  \mysum_{l\in\mathbb{N}^{\dagger}}^{d\in\mathbb{J}^{\dagger}}
  (\cdot)~g_{n-l,j-d}^{\alpha}(\Delta\tau)~(\cdot)$
in our scheme, implementation descriptions, and subsequent analysis.
These expressions should respectively be understood as shorthand for
$\widetilde{g}_{n-l,j-d}^{\alpha}(\Delta \tau)$
and
$\mysum_{l\in\mathbb{N}^{\dagger}}^{d\in\mathbb{J}^{\dagger}}
  (\cdot)~\widetilde{g}_{n-l,j-d}^{\alpha}(\Delta \tau)~(\cdot)$,
where $\widetilde{g}_{n-l,j-d}^{\alpha}(\Delta \tau)$ is the rescaled weight defined in \eqref{eq:widetilde_g}.
The same convention applies to matrix- or vector-valued expressions involving $g_{n-l,j-d}(\Delta \tau)$.
\end{remark}

\subsection{Efficient implementation and algorithms}
In this section, we discuss an efficient implementation of the 2D discrete convolution
\eqref{eq:dissum} using FFT. For convenience, with $N^{\dagger}=2N$, $N^{\ddagger}=3N$, $J^{\dagger}=2J$ and $J^{\ddagger}=3J$,
we define/recall sets of indices:
$\Nbb^{\ddagger}=\{-N^{\ddagger}/2+1,\ldots,N^{\ddagger}/2-1\}$, $\Nbb^{\dagger}=\{-N^{\dagger}/2,\ldots,N^{\dagger}/2\}$, $\Nbb=\{-N/2+1,\ldots,N/2-1\}$, $\Jbb^{\ddagger}=\{-J^{\ddagger}/2+1,\ldots,J^{\ddagger}/2-1\}$, $\Jbb^{\dagger}=\{-J^{\dagger}/2,\ldots,J^{\dagger}/2\}$, $\Jbb=\{-J/2+1,\ldots,J/2-1\}$.

For a fixed $m$ and a fixed $\al$,  to write \eqref{eq:dissum} for all $n \in \Nbb$ and $j \in \Jbb$ into a matrix-vector multiplication form,  we adopt the following notation:
\begin{itemize}
  \item For a fixed $j\in \Jbb$ and a fixed  $\al\in \Al_h$ and $m\in \{1, \ldots, M\}$, let $\U_{j}^{m, \al}$ be a column vector of length $(N-1)$ defined by
   $\U_{j}^{m, \al}\equiv\l[u^{m, \al}_{-N/2+1,j},u^{m, \al}_{-N/2+2,j},\ldots, u^{m, \al}_{N/2-1,j}\r]^{\T}$;

  \item
  For a fixed $q\in \Jbb^{\ddagger}$ and $m\in \{0, \ldots, M-1\}$, let $\V_{q}^{m}$ be a column vector of length
  $(2N+1)$ defined by
  $\V_{q}^{m}\equiv\l[v^{m}_{-N^{\dagger}/2,q}~\varphi_{-N^{\dagger}/2,q},
  v^{m}_{-N^{\dagger}/2+1,q}~\varphi_{-N^{\dagger}/2+1,q},\ldots, v^{m}_{N^{\dagger}/2,q}~\varphi^{m}_{N^{\dagger}/2,q}\r]^{\T}$.

\item For a fixed $q\in\Jbb^{\ddagger}$ and a fixed $\al\in\Al_h$,
let $\gb_{q}^{\al}$ be a (non-square) matrix  of size  $(N-1)\times(2N+1)$,
representing the convolution kernel in the inner summation (over $n$), defined as follows
\EQ
\label{eq:gb}
\gb_{q}^{\al}=
\l[g^{\al}_{n-l,q}\r]_{n\in\Nbb,l\in\Nd} =
\begin{bmatrix}
    g^{\al}_{N/2+1, q} & g^{\al}_{N/2, q}& \dots & \dots &\dots &  g^{\al}_{-3N/2+1, q} \\
    g^{\al}_{N/2+2, q} & g^{\al}_{N/2+1, q} & \dots & \dots & \dots & g^{\al}_{-3N/2+2, q} \\
    \vdots & \vdots &\vdots& \vdots &\vdots  &\vdots \\
    g^{\al}_{3N/2-1, q} & g^{\al}_{3N/2-2, q} & \dots &  g^{\al}_{N/2+1, q} & \dots &  g^{\al}_{{\myblue{-N/2-1}}, q}
\end{bmatrix}.
\EN

\end{itemize}
In this setup, we can express the 2D discrete convolution \eqref{eq:dissum} for all
$n \in \Nbb$ and $j \in \Jbb$ into a matrix-vector product form as follows
\EQA
\label{eq:tmf}
\underbrace{\l[
 \begin{matrix}
    \U^{m+1, \al}_{-J/2+1}\\
    \U^{m+1, \al}_{-J/2+2}\\
    \vdots \\
    \U^{m+1, \al}_{J/2-1}
\end{matrix}
\r]}_{\U^{m+1, \al}}
=\Delta x \Delta y
\underbrace{\l[
 \begin{matrix}
    \gb^{\al}_{J/2+1} &\gb^{\al}_{J/2} &\ldots &\ldots &\ldots &\gb^{\al}_{-3J/2+1}\\
    \gb^{\al}_{J/2+2} &\gb^{\al}_{J/2+1} &\ldots &\ldots &\ldots &\gb^{\al}_{-3J/2+2}\\
    \vdots &\vdots&\vdots &  \vdots &\vdots&\vdots\\
    \gb^{\al}_{3J/2-1} &\gb^{\al}_{3J/2-2} &\ldots &\gb^{\al}_{J/2+1} &\ldots &\gb^{\al}_{-J/2-1}
\end{matrix}
\r]}_{\l[\gb^{\al}_{j-d}\r]_{j\in\Jbb,d\in\Jd}}
\underbrace{\l[
 \begin{matrix} 
    \V^{m}_{-J^{\dagger}/2}\,\,\,\,\,\,\,\\
    \V^{m}_{-J^{\dagger}/2+1}\\
    \vdots \\
    \V^{m}_{J^{\dagger}/2}\,\,\,\,\,\,\,
\end{matrix}
\r]}_{\V^{m}}.
\ENA
Here, $\U^{m+1, \al}$ is a column vector  of length $(N-1)(J-1)$;
the (block) matrix $\l[\gb_{j-d}^{\al}\r]_{j\in\Jbb,d\in\Jd}$, which represents the convolution kernel for the double summation,
is of size $(N-1)(J-1)\times (2N+1)(2J+1)$;
$\V^{m}$ is a column vector of length $(2N+1)(2J+1)$.

It is noteworthy that the non-square matrix $\l[\gb_{j-d}^{\al}\r]_{j\in\Jbb,d\in\Jd}$ is a Toeplitz matrix \cite{bryc2006spectral}, enabling efficient computation of \eqref{eq:tmf} using FFT and circular convolution. This technique, initially applied to 1D problems in \cite{zhang2023monotone}, is now adapted to the 2D case given by \eqref{eq:tmf}.
Our goal is to represent \eqref{eq:tmf} as a circulant matrix-vector product.
This involves expanding $\l[\gb_{j-d}^{\al}\r]_{j\in\Jbb,d\in\Jd}$ to a 2D ciculant matrix -
a block matrix where each block is circulant and the blocks are arranged in a circulant pattern.
More specifically, the process involves (i) expanding each block $\gb_{j-d}^{\al}$ to a circulant matrix, denoted by $\tilde{\gb}^{\al}_{j-d}$, and (ii) expanding $\l[\tilde{\gb}_{j-d}^{\al}\r]_{j\in\Jbb,d\in\Jd}$
to a 2D circulant matrix. Correspondingly,
the vector $\V^{m}$  {\imablue{is}} also expanded to conform with this format. Key steps of this expansion process are outlined below.

\begin{itemize}
    \item Expansion of blocks: For each matrix $\gb_{{\myblue{q}}}^{\al} = \l[g^{\al}_{n-l,{\myblue{q}}}\r]_{n\in\Nbb,l\in\Nd}$, ${\myblue{q}}\in \mathbb{J}^{\ddagger}$, of size $(N-1) \times (2N+1)$, we expand it into a circular matrix $\tilde{\gb}^{\al}_{{\myblue{q}}}$ of size $(3N-1) \times (3N-1)$. This expansion, detailed in \cite{zhang2023monotone}, results in the matrix
        \EQA
\label{eq:cir_matrix_1D}
\tilde{\gb}^{\al}_{{\myblue{q}}} =
\l[\begin{array}{c|c}
\tilde{\gb}_{-1,0}^{{\myblue{q}}, \al} & \tilde{\gb}_{-1,1}^{{\myblue{q}}, \al}\\
\hline
\gb_{{\myblue{q}}}^{\al} & \tilde{\gb}_{0,1}^{{\myblue{q}}, \al}\\
\hline
\tilde{\gb}_{1,0}^{{\myblue{q}}, \al} & \tilde{\gb}_{1,1}^{{\myblue{q}}, \al}
\end{array}\r], \qquad
\gb_{{\myblue{q}}}^{\al} = \l[g^{\al}_{n-l,{\myblue{q}}}\r]_{n\in\Nbb,l\in\Nd}.
\ENA
Here, $\tilde{\gb}_{-1,0}^{{\myblue{q}}, \al}$, $\tilde{\gb}_{1,0}^{{\myblue{q}}, \al}$, $\tilde{\gb}_{-1,1}^{{\myblue{q}}, \al}$, $\tilde{\gb}_{0,1}^{{\myblue{q}}, \al}$ and $\tilde{\gb}_{1,1}^{{\myblue{q}}, \al}$ are padding matrices of sizes $N\! \times\! (2N+1)$, $N\! \times\! (2N+1)$,
$N\! \times\!(N-2)$, $(N-1)\! \times\! (N-2)$, and $N\!\times\! (N-2)$, respectively.
These matrices are appropriately defined to ensure the circulant structure of $\tilde{\gb}^{\al}_{{\myblue{q}}}$.
Further details on these padding matrices are provided in Appendix~\ref{app:mats}.

    \item  Expansion of block matrix: We then substitute $\gb_{j-d}^{\al}$ with circulant block $\tilde{\gb}_{j-d}^{\al}$ in $\l[\gb^{\al}_{j-d}\r]_{j\in\Jbb,d\in\Jd}$.  The resulting block matrix $\l[\tilde{\gb}_{j-d}^{\al}\r]_{j\in\Jbb,d\in\Jd}$ is then expanded into a circulant matrix of size $(3N-1)(3J-1)\times(3N-1)(3J-1)$, denoted as $\gbt^{\al}$. Specifically, $\gbt^{\al}$ is constructed as follows:
    \EQA
    \label{cir_t}
  \gbt^{\al} =   \l[~~
    \begin{matrix}
\Bs_{-1,0} & \vline&\Bs_{-1,1} \\
\hline
\l[\tilde{\gb}^{\al}_{j-d}\r]_{j\in\Jbb,d\in\Jd} &\vline&\Bs_{0,1} \\
\hline
\Bs_{1,0} & \vline&  \Bs_{1,1}
\end{matrix}
~~\r].
    \ENA
Here,  $\Bs_{-1,0}$, $\Bs_{1,0}$, $\Bs_{-1,1}$, $\Bs_{0,1}$ and $\Bs_{1,1}$ components are (block) matrices
with dimensions $(3N-1)J\times (3N-1)(2J+1)$, $(3N-1)J\times (3N-1)(2J+1)$, $(3N-1)J\times(3N-1)(J-2)$, $(3N-1)(J-1)\times(3N-1)(J-2)$ and $(3N-1)J\times(3N-1)(J-1)$, respectively.
These  matrices are appropriately defined to ensure the circulant structure of $\gbt^{\al}$.
Further details on these padding matrices are provided in Appendix~\ref{app:mats}.

   \item Vector expansion: To conform with the circulant-maxtrix format, for each $q\in \Jd$,  we construct the augmented column vector ${\Vt}^{m}_{q}$ of length $(3N-1)$,
    by appending zeros to the column vector $\V^{m}_{q}$. This is defined as follows.
    \EQA
    \label{eq:aug_v}
    {\Vt}_{q}^{m} = \l[v^{m}_{-N^{\dagger}/2,q}~\varphi_{-N^{\dagger}/2,q},
        v^{m}_{-N^{\dagger}/2+1,q}~\varphi_{-N^{\dagger}/2+1,q},\ldots, v^{m}_{N^{\dagger}/2,q}~\varphi_{N^{\dagger}/2,q},0,0,\ldots,0\r]^{\T}.
    \ENA
    Then, we form the vector ${\Vt}^{m}$ of size $(3N-1)(3J-1)$ by appending zeros
    as follows:
        \EQA
    \label{eq:aug_vm}
    {\Vt}^{m} = \l[{\Vt}^{m}_{-J^{\dagger}/2},{\Vt}^{m}_{-J^{\dagger}/2+1},\ldots, {\Vt}^{m}_{J^{\dagger}/2},\0,\0,\ldots,\0\r]^{\T}.
    \ENA
    where $\0$'s are zero vectors of length $(3N-1)$.

    \item Circulant matrix-vector product:   Utilizing this setup, we express the matrix-vector product \eqref{eq:tmf} as a circulant matrix-vector product, which is used to compute an intermediate column vector of discrete solutions. This column vector, denoted by $\Ut^{m+1,  \al}$, has a length of $(3N-1)(3J-1)$ and is determined as follows:
    \EQA
    \label{eq:tmp}
        \Ut^{m+1, \al} = \Delta x \Delta y~ {\gbt}^\al ~ \Vt^{m}, \qquad\al\in\Al_h.
    \ENA
    Here, ${\gbt}^\al $ is the circulant matrix defined in \eqref{cir_t}, $\Vt^{m}$ is the (augmented) column vector  given by \eqref{eq:aug_vm}.
    We note that discrete solutions $u_{n, j}^{m+1, \al}$ for $\Omega_{\myin}$ are obtained by discarding the components in $\Ut^{m+1, \al}$ corresponding to indices $n \in \mathbb{N}^{\ddagger} \setminus \mathbb{N}$ or $j \in \mathbb{J}^{\ddagger} \setminus \mathbb{J}$.

    \end{itemize}
The circulant matrix-vector product in \eqref{eq:tmp} can be efficiently computed as a circulant convolution using 2D FFT.
To this end, we let ${\gbh}^{\al}_{1}$ be the first column ${\gbt}^{\al}$ defined in \eqref{cir_t} reshaped into
a $(3N-1)\times (3J-1)$ matrix as follows
\EQA
\label{eq:gpt_1}
\begin{aligned}
{\gbh}^{\al}_{1} = &\left[~~
\begin{matrix}
\l[\gbt^{\al}_{-J/2+1}\r]_1   & \ldots &  \l[\gbt^{\al}_{J/2}\r]_1
& \l[\gbt^{\al}_{J/2+1}\r]_1 & \ldots & \l[\gbt^{\al}_{3J/2-1}\r]_1
& \l[\gbt^{\al}_{-3J/2+1}\r]_1 & \ldots & \l[\gbt^{\al}_{-J/2}\r]_1
\end{matrix}
~~\right]~.
\end{aligned}
\ENA
Here, $\l[\gbt^{\al}_{q}\r]_1$, $q\in\Jbb^{\ddagger}$, denotes the first column of the matrix
$\gbt^{\al}_{q}$.
We  reshape the vector $\Vt^{m}$ into a $(3N-1)\times (3J-1)$ matrix, denoted by $\l[\Vt^{m}\r]$.
The circulant matrix-vector product in \eqref{eq:tmp} can be expressed as a 2D circular convolution product
\EQA
\label{eq:cir_p}
   \l[\Ut^{m+1, \al}\r] =  \Delta x \Delta y~ {\gbh}^\al_1 \ast \l[\Vt^{m}\r], \qquad\al\in\Al_h.
\ENA
Here, $\l[\Ut^{m+1, \al}\r]$ is a $(3N-1)\times (3J-1)$ matrix, representing the reshaped version of $\Ut^{m+1, \al}$
from \eqref{eq:tmp}.
The circular convolution product \eqref{eq:cir_p} is computed efficiently using FFT and inverse FFT (iFFT) as follows
\EQA
\label{eq:fft_ifft}
  \l[\Ut^{m+1, \al}\r] = {\text{FFT}}^{-1}\l\{\text{FFT}\l\{\l[\Vt^{m}\r]\r\} \circ \text{FFT}\l\{{{\Delta x \Delta y\,  {\gbh}^\al_1}}\r\}\r\}, \qquad\al\in\Al_h.
\ENA
Finally, we discard the components in $\l[\Ut^{m+1, \al}\r]$ corresponding to indices $n \in \mathbb{N}^{\ddagger} \setminus \mathbb{N}$ or $j \in \mathbb{J}^{\ddagger} \setminus \mathbb{J}$, obtaining
discrete solutions $u_{n, j}^{m+1, \al}$ for $\Omega_{\myin}$.

As explained in Remark~\ref{rem:rescaled_kernel}, the factor $\Delta x,\Delta y$ is incorporated into $g_{n-l,j-d}^{\alpha}$, yielding $\widetilde{g}_{n-l,j-d}^{\alpha}$.
Following our convention, for simplicity, we continue to write
$\Delta x \Delta y\, \l[\gb^{\al}_{j-d}\r]$ in \eqref{eq:tmf},
 $\Delta x \Delta y~ {\gbt}^\al$ in \eqref{eq:tmp}, and
$\Delta x\,\Delta y\, {\gbh}^\al_1$ in~\eqref{eq:cir_p}-\eqref{eq:fft_ifft}

The implementation \eqref{eq:fft_ifft} suggests that we compute
the rescaled weight components of $\Delta x \Delta y\, {\gbh}^\al_1$
only once for each $\al\in\Al_h$ through the closed-form expression in \eqref{eq:g_k}, and reuse them for the computation over all time intervals.
Putting everything together, the proposed numerical scheme for the two-factor uncertain volatility model pricing problem is presented in Algorithm \ref{alg:monotone} below.
\begin{algorithm}[H]
\caption{
A monotone piecewise constant control integration algorithm for a two-factor uncertain volatility model pricing problem defined in Definition~\ref{def:uvm_def}, {{where $h>0$ is fixed.}}
}

\begin{algorithmic}[1]
\label{alg:monotone}

\STATE
\label{alg:gf}
for each $\al\in\Al_h$, and for each $j\in\Nbb^{\ddagger}$,
compute rescaled weight matrices \mbox{$\Delta x \Delta y\, \gb_{j}^{\al} = \l[\Delta x \Delta y\,g^\al_{n-l,j}\r]_{n\in\Nbb,l\in\Nd}$} defined in \eqref{eq:gb}
using the closed-form expression \eqref{eq:g_k};

\STATE
\label{alg:gf1}
construct rescaled weight matrices $\Delta x \Delta y\, {\gbh}^\al_1$, $\al\in\Al_h$, using  $\l[\Delta x \Delta y\, g^\al_{n-l,j}\r]_{n\in\Nbb,l\in\Nd}$, $j\in\Nbb^{\ddagger}$, defined in \eqref{eq:gpt_1};

\STATE
\label{alg:initial}
initialize $v_{n, j}^{0}
=p(e^{x_n},e^{y_j})$,
$n\in\Nd, j\in \Jd$;

 \FOR{$m = 0, \ldots, M-1$}

  \FOR{$\al\in\Al_h$}
  \label{alg:forstart}

 \STATE
    \label{alg:step1}
    compute matrices of intermediate values $[\Ut^{m+1, \al}]$ using FFT and iFFT as per \eqref{eq:fft_ifft};

\STATE
\label{alg:dis}
 obtain vector of discrete solutions $\U^{m+1, \al} = \l[u_{n, j}^{m+1, \al}\r]_{n\in \Nbb, j \in \Jbb}$ by  discarding the components in $[\Ut^{m+1, \al}]$ corresponding to indices $n \in \mathbb{N}^{\ddagger} \setminus \mathbb{N}$ or $j \in \mathbb{J}^{\ddagger} \setminus \mathbb{J}$;

 \ENDFOR
 \label{alg:forend}

        \STATE
    \label{alg:step4}
    set $v_{n,j}^{m+1}= \max_{\al\in\Al_h}u_{n,j}^{m+1, \al}$
     with  $\displaystyle (\alpha^*)_{n, j}^{m+1} = \arg\max_{\alpha \in \mathcal{A}h} u_{n,j}^{m+1, \alpha}$,
      $n\in\Nbb$ and  $j\in\Jbb$,
      \\
    where $u_{n,j}^{m+1, \al}$ are from Line~\ref{alg:dis};
    \hfill $(\Omega_{\myin})$

        \STATE
    \label{alg:step5}
    compute $v_{n,j}^{m+1}$, $n\in\Nbb^{\dagger}\setminus\Nbb$ or $j\in\Jbb^{\dagger}\setminus\Jbb$, using \eqref{eq:outi}; \hfill $(\Omega_{\myout})$

\ENDFOR
\end{algorithmic}
\end{algorithm}
To set the stage for highlighting the key differences between our method and finite difference methods with policy iteration, we first provide a detailed explanation of Algorithm~\ref{alg:monotone} for a fixed $h > 0$. As noted in Subsection~\ref{ssc:pcc}, the core component of the algorithm is the piecewise constant control method combined with monotone numerical integration using Green's functions.

In Lines~\ref{alg:gf}-\ref{alg:gf1}, non-negative rescaled weight matrices $\Delta x \Delta y\, {\gbh}^\al_1$, $\al \in \Al_h$, are precomputed for a fixed
$\Delta \tau$ using the Green's functions of independent 2D PDEs,
each corresponding to a discretized control value $\al \in \mathcal{A}_h$.
Since the timestep size $\Delta \tau$ is fixed, these rescaled weight matrices need to be computed only once and can be reused across all timesteps.
In Lines~\ref{alg:forstart} to \ref{alg:forend}, the independent linear 2D PDEs for $a \in \mathcal{A}_h$ are solved to obtain the solutions at
$\tau_{m+1}$, $m \in \{ 0, \ldots, M-1\}$, using a numerical integration scheme implemented via FFT and iFFT. This scheme is monotone in the viscosity sense due to the non-negativity of the (rescaled) weights.
In Line~\ref{alg:step4}, the time-$\tau_{m+1}$ numerical solutions of these 2D PDEs, $u_{n,j}^{m+1, \al}$ for $\al \in \Al_h$, are combined using the $\texttt{max}(\cdot)$ operator to compute approximations of the value function (i.e.\ $v_{n,j}^{m+1}$) and the optimal control (i.e.\ $(\alpha^*)_{n,j}^{m+1}$) at the grid points, directly addressing the nonlinearity of the HJB equation within $\Omega_{\myin}$. Finally, in Line~\ref{alg:step5}, boundary conditions are applied to ensure proper handling of $\Omega_{\myout}$.

\begin{remark}[Comparison with finite differences and policy iteration]
The proposed approach, based on the piecewise constant control method and numerical integration using Green's functions, as outlined in Algorithm~\ref{alg:monotone}, differs fundamentally from conventional methods, such as finite differences combined with policy iteration. These distinctions are particularly significant in addressing the nonlinearity of HJB equations and overcoming associated computational challenges.

Conventional methods, often referred to as ``discretize, then optimize'', typically rely on finite difference schemes to approximate the temporal and spatial partial derivatives of the HJB equation. When explicit time-stepping is employed, these schemes are constrained by CFL conditions, which impose restrictions on the timestep size $\Delta \tau$ to ensure numerical stability \cite{courant1928partiellen, courant1967partial}. Alternatively, implicit time-stepping avoids these constraints but results in a system of nonlinear algebraic equations that must be solved iteratively at each timestep, typically via policy iteration
\cite{wang08, forsyth2007numerical}. In both cases, computing the optimal control and value function involves either stringent timestep restrictions or computationally expensive iterative procedures.

In contrast, our approach, succinctly described as ``decompose, integrate, then optimize'',
avoids direct discretization of partial derivatives in the HJB equation and proceeds in two steps. First, we discretize the control set $\mathcal{A}$ into a finite subset $\mathcal{A}_h$ and treat each discretized control as constant on each time sub-interval. This yields a set of independent linear 2D PDEs in $(x,y)$, each corresponding to a discretized control value.
Each PDE is solved using Green's functions, representing the solution at the next time point as a 2D convolution integral. This integral is evaluated using a numerical integration scheme that is monotone in the viscosity sense--no {\imapurple{additional explicit-scheme CFL restriction or nonlinear iterative solvers are required.}}

In the second step, we combine the solutions of these linear PDEs at each grid point
using a $\max\{\cdot,\cdot\}$ operation, which preserves monotonicity and addresses the
HJB equation's nonlinearity. This yields approximations to both the value function and the
optimal control without relying on policy iteration or other iterative methods. By bypassing
derivative-based discretization, the proposed method 
{\imapurple{avoids explicit-scheme CFL timestep restrictions}}
and offers a robust and efficient alternative for solving HJB equations.

\end{remark}


\begin{remark}[Complexity]
As noted earlier, the cardinality of $\Al_h$, denoted by $Q$,  is $Q = \Ocal(1/h)$.
Algorithm \ref{alg:monotone} involves, for $m=0,\ldots,M-1$, the following key steps:
\begin{itemize}

    \item Compute $u_{n,j}^{m+1, \al}$, $n\in\Nbb^{\dagger}$, $j\in\Jbb^{\dagger} $ for all $\al\in\Al_h$ via FFT algorithm. The complexity of this step is $\Ocal(QNJ\log(NJ)=\Ocal(1/h^3\cdot\log(1/h))$, where we take into account \eqref{eq:dis_parameter}.

    \item Finding the optimal control $(\al^*)_{n, j}^{m+1}$ for each node $\x_{n,j}^{m+1}$ by comparing $u_{n,j}^{m+1, \al}$ for all $\al \in \Al_h$ requires $\Ocal(1/h)$ complexity. Thus, with a total of $\Ocal(1/h^2)$ nodes, this gives a complexity $\Ocal(1/h^3)$.

    \item Therefore, the major cost of Algorithm \ref{alg:monotone} is determined by the step of FFT Algorithm. With $\Ocal(1/h)$ timesteps, the total complexity is $\Ocal(1/h^4\cdot\log(1/h))$.

    \end{itemize}

\end{remark}

\section{Convergence to viscosity solution}
\label{sc:conv}
In this section, we appeal to a Barles-Souganidis-type analysis \cite{barles-souganidis:1991} to rigorously study the convergence of our scheme in $\Omega_{\myin}$ as $h \to 0$ by verifying three properties: $\ell_\infty$-stability, monotonicity,
and consistency.
Our scheme consists of \eqref{eq:tau0i} (for $ {\Omega}_{\tau_0}$), \eqref{eq:outi} (for $ {\Omega}_{{\myout}}$),
and \eqref{eq:vmax} (for $\Omega_{{\myin}}$).

For subsequent use, we state several results below.
For $\Omega_{{\myin}}$, from Lemma~\ref{lem:g_k}, for a fixed $\alpha\in\Al$, we have $\iint_{\mathbb{R}^2}g_\alpha(x,y;   \Delta \tau)dxdy = e^{-r\Delta \tau}$, hence $ \iint_{\mathbb{D}^{\dagger}}g_\alpha(x,y;   \Delta \tau)dxdy\leq e^{-r\Delta \tau}<1$,
where $\mathbb{D}^{\dagger}$ is defined in~\eqref{eq:truncate_region}.
For $n \in \Nbb$ and $j \in \Jbb$ (i.e.\ $\Omega_{{\myin}}$),
we define
\EQS
\epsilon_{g}= \max_{\alpha,n,j}\epsilon^{\alpha}_{n,j},
\text{ where }~ \epsilon^{\alpha}_{n,j}\coloneqq \bigg| \iint_{\mathbb{D}^{\dagger}}g_\alpha(x_n- x',y_j- y';   \Delta \tau)dx'dy' - \Delta x \Delta y \mysum_{l\in\mathbb{N}^{\dagger}}^{d\in\mathbb{J}^{\dagger}}\varphi_{l,d} ~g^\alpha_{n-l,j-d}\bigg|.
\ENS
Here, at noted earlier, $\varphi_{l,d}$ are the weights of the 2D composite trapezoidal quadrature rule.
Using the definition of $\epsilon_g$ and the fact that $\varphi_{l,d}>0$, for any fixed $\alpha \in \Al$, we have
\EQ
\label{eq:gkbound}
0~\leq~ \Delta x \Delta y \mysum_{l\in\mathbb{N}^{\dagger}}^{d\in\mathbb{J}^{\dagger}}\varphi_{l,d} ~g^\alpha_{n-l,j-d}~<~1+\epsilon^{\alpha}_{n,j}~\leq~1+\epsilon_{g}< e^{\epsilon_{g}}.
\EN
To establish $\epsilon_g = \mathcal{O}(h^2)$, we begin by analyzing $\epsilon^{\alpha}_{n,j}$ for any $\alpha \in \Al_h$ and showing that it satisfies this order. Using the explicit form of $g_\alpha(x, y; \Delta \tau)$ provided in Lemma~\ref{lem:g_k} [Equation~\ref{eq:g_k}], we observe that $g_\alpha$ is smooth and that its second-order partial derivatives remain bounded for all
$\Delta \tau > 0$ (or for all $h > 0$ by \eqref{eq:dis_parameter}). Letting
$C_\alpha = \sup_{(x, y) \in \mathbb{D}^\dagger} \max\bigg(\bigg|\frac{\partial^2 g_\alpha}{\partial x^2}\bigg|, \bigg|\frac{\partial^2 g_\alpha}{\partial y^2}\bigg|\bigg)$, we note that $C_\alpha$ is a bounded constant independent of $h$.
Using the error formula for the composite trapezoidal rule applied to a smooth function over the
bounded rectangular domain $\mathbb{D}^\dagger$, with
$|\mathbb{D}^\dagger| := (y^{\dagger}_{\mymax}- y^{\dagger}_{\mymin})
(x^{\dagger}_{\mymax} - x^{\dagger}_{\mymin})$ denoting its area,
we obtain the bound for all $n \in \Nbb$, $j \in \Jbb$, and a fixed $\alpha \in \Al_h$:
\[
\epsilon^{\alpha}_{n,j} \leq \bigg(\frac{\Delta x^2 + \Delta y^2}{12}\bigg)\, C_\alpha \, |\mathbb{D}^\dagger|
~\overset{\text{(i)}}{=} \bigg(\frac{C_1^2 + C_2^2}{12}\bigg) \, C_\alpha \, |\mathbb{D}^\dagger| \, h^2.
\]
Here, in (i), $\Delta x = C_1 h$ and $\Delta y = C_2 h$ by \eqref{eq:dis_parameter}.
Since $\Al$ is compact, $C_\alpha$, as a continuous function of $\alpha$, is uniformly bounded over $\Al$. Consequently, $C^\alpha_{\mymax} = \sup_{\alpha \in \Al} C_\alpha$ is finite independently of $h$. From here, we obtain
\EQ
\label{eq:epsilon_g_bound}
\epsilon_g =  \max_{\alpha,n,j}\epsilon^{\alpha}_{n,j}\le
\bigg(\frac{C_1^2 + C_2^2}{12}\bigg) \, C^\alpha_{\mymax} \, |\mathbb{D}^\dagger| \, h^2.
\EN
Since $C_1$, $C_2$, $C^\alpha_{\mymax}$, and $|\mathbb{D}^\dagger|$ are bounded constants independently of $h$, it follows that for sufficiently small $h$, $\epsilon_g = \mathcal{O}(h^2)$.


\subsection{Stability}
\label{ssc:stability}
{\imapurple{Using the explicit bound for $\epsilon_g$ in \eqref{eq:epsilon_g_bound}, we now show that the scheme is $\ell_\infty$-stable.}}
\begin{lemma}[$\ell_\infty$-stability]
\label{lemma:stability}
Suppose the discretization parameter $h$ satisfies \eqref{eq:dis_parameter}.
Then our scheme, which consists of \eqref{eq:tau0i}, \eqref{eq:outi}, and \eqref{eq:vmax}, satisfies the bound
$\ds \sup_{h > 0} \left\| v^{m} \right\|_{\infty}<\infty$
for all $m = 0, \ldots, M$, as the discretization parameter $h \to 0$.
Here, we have $\left\| v^{m} \right\|_{\infty} = \max_{n, j} |v_{n, j}^{m}|$,
$n \in \mathbb{N}^{\dagger}$ and $j \in \mathbb{J}^{\dagger}$.
\end{lemma}
\begin{proof}[Proof of Lemma~\ref{lemma:stability}]
First, we note that, for any fixed $h>0$, as given by \eqref{eq:tau0i}, we have $||v^{0}||_{\infty}<\infty$, since $\myblue{\Omega}$ is a bounded domain. Therefore, we have $\sup_{h>0}||v^{0}||_{\infty}<\infty$. Motivated by this observation, to demonstrate $l_{\infty}$-stability of our scheme, we will show that, for a fixed $h>0$, at any $(x_n,y_j,\tau_m)$, $m=0,\ldots,M$, we have
\EQ
\label{eq:st1}
|v^{m}_{n,j}|<e^{m\epsilon_g}||v^{0}||_{\infty},~~m=0,1,\ldots,M.
\EN
To see why  \eqref{eq:st1} is bounded as $h \to 0$, note that, by \eqref{eq:dis_parameter},  $mh \leq Mh = M \Delta \tau / C_3 = T / C_3$.
Together with \eqref{eq:epsilon_g_bound}, this results in
\[
e^{m\epsilon_g} \le  e^{M\epsilon_g}
\le
\exp\bigg(\bigg(\frac{C_1^2 + C_2^2}{12}\bigg) \, C^\alpha_{\mymax} \, |\mathbb{D}^\dagger| \, (T / C_3)\, h\bigg)
\to  1, \quad \text{as  $h \to 0$}.
\]
 It is straightforward to show that \eqref{eq:tau0i} is $\ell_{\infty}$-stable, since $\max_{n,j}|v^{0}_{n,j}|\leq||v^{0}||_{\infty}$ for $(n,j)\in\mathbb{N}^{\dagger} \times \mathbb{J}^{\dagger}$, clearly satisfying \eqref{eq:st1}.
 Next, for equation \eqref{eq:outi}, we note  that, since  $|v_{n, j}^{m+1}| = |v_{n, j}^{m} e^{-r\Delta \tau}| < |v_{n, j}^{m}|$,
 by induction on $m$, we have $\max_{n,j}|v^{m}_{n,j}|\leq||v^{0}||_{\infty}$,
 for either $(n, j) \in (\mathbb{N}^{\dagger}\setminus \mathbb{N}) \times \mathbb{J}^{\dagger}$ or
 $\mathbb{N}^{\dagger}  \times (\mathbb{J}^{\dagger}\setminus \mathbb{J})$.

Now we focus on the main task, demonstrating $\ell_{\infty}$-stability for \eqref{eq:vmax} ($\Omega_{\myin}$)
through an induction proof on $m$.
For the base case $m = 1$, with a fixed $\alpha \in \Al_h$,
\EQ
\label{eq:u_1}
u^{1}_{n,j} = \Delta x \Delta y \mysum_{l\in\mathbb{N}^{\dagger}}^{d\in\mathbb{J}^{\dagger}}\varphi_{l,d}~g^\alpha_{n-l,j-d}~v_{l, d}^{0}.
\EN
Then, we have
\begin{align}
\label{eq:in_base}
|u^{1}_{n,j}|~\leq~ \Delta x \Delta y \mysum_{l\in\mathbb{N}^{\dagger}}^{d\in\mathbb{J}^{\dagger}}\varphi_{l,d} ~g^\alpha_{n-l,j-d}|v^{0}_{{{l,d}}}|
~\leq~ \Delta x \Delta y \mysum_{l\in\mathbb{N}^{\dagger}}^{d\in\mathbb{J}^{\dagger}}\varphi_{l,d} ~g^\alpha_{n-l,j-d}||v^{0}||_{\infty}~\leq~e^{1\epsilon_g}||v^{0}||_{\infty},
\end{align}
{\zblue{where the last inequality is due to \eqref{eq:gkbound}.
Since $v^{1}_{n,j} = \max_{\alpha}u_{n,j}^{1}$, we have
\[
|v_{n,j}^{1}| = |\max_{\alpha}u^{1}_{n,j}|
~\leq~ \max_{\alpha}|u^{1}_{n,j}|
~\leq~ e^{1 \epsilon_g}||v^{0}||_{\infty},
\]
as wanted for the base case.}}
For the hypothesis,  assume that \eqref{eq:st1} hold for $m=m'$, $1\leq m'\leq M-1$
\EQ
\label{eq:hyp}
|v^{m'}_{n,j}|<e^{m'\epsilon_g}||v^{0}||_{\infty},\qquad (n,j)\in\mathbb{N}\times\mathbb{J}.
\EN
In the induction step, we need show that \eqref{eq:st1} also holds for $m=m'+1$, i.e.\
\EQ
\label{eq:induc}
|v^{m'+1}_{n,j}|<e^{(m'+1)\epsilon_g}||v^{0}||_{\infty}.
\EN
To show \eqref{eq:induc}, recalling $u^{m'+1}_{n,j}$ from  \eqref{eq:dissum} gives
\begin{align}
\label{eq:indu_m1}
|u^{m'+1}_{n,j}|&~\leq~ \Delta x \Delta y \mysum_{l\in\mathbb{N}^{\dagger}}^{d\in\mathbb{J}^{\dagger}}\varphi_{l,d} ~g^\alpha_{n-l,j-d}|v_{l, d}^{m'}|
~\overset{\text{(i)}}{\leq}~ \Delta x \Delta y \mysum_{l\in\mathbb{N}^{\dagger}}^{d\in\mathbb{J}^{\dagger}}\varphi_{l,d} ~g^\alpha_{n-l,j-d}~e^{m'\epsilon_g}||v^{0}||_{\infty}\nonumber\\
&~\leq~ e^{\epsilon_g}e^{m'\epsilon_g}||v^{0}||_{\infty}~=~e^{(m'+1)\epsilon_g}||v^{0}||_{\infty}.
\end{align}
{\zblue{Here, (i) is due to the hypothesis \eqref{eq:hyp} together with the fact that
the scheme for ${\Omega}_{{\myout}}$, captured by equation \eqref{eq:outi}, is also $\ell_{\infty}$-stable as shown earlier.}}
Hence, $|v_{n,j}^{m'+1}| = |\max_{\alpha}u^{m'+1}_{n,j}|
~\leq~ e^{(m'+1)\epsilon_g}||v^{0}||_{\infty}$, proving \eqref{eq:induc} for $m=m'+1$. This concludes the proof.
\end{proof}

\subsection{Consistency}
While equations \eqref{eq:tau0i}, \eqref{eq:outi}, and \eqref{eq:vmax} are convenient for computation, they are not in a form
amendable for analysis. For purposes of verifying consistency, it is more convenient to rewrite them in
a single equation. To this end, for $(x_n, y_j, \tau_{m+1}) \in \Omega_{\myin}$, i.e.\ $n \in \Nbb$ and $j \in \Jbb$,
we define operator $\mathcal{C}_{n, j}^{m+1}(\cdot)$, where
\EQ
\label{eq:scheme_CD}
\begin{aligned}
\mathcal{C}_{n, j}^{m+1}(\cdot) \equiv
\mathcal{C}_{n, j}^{m+1}
\bigg(h, v_{n, j}^{m+1},
\left\{v_{l,d}^{m}\right\}_{\subalign{l\in \Nd\\d\in \Jd}}
 \bigg)
= \frac{1}{\Delta \tau}\bigg[v_{n,j}^{m+1}
-
\max_{\alpha \in \Al_h} ~ \bigg\{\Delta x \Delta y \mysum_{l \in \Nd}^{d \in \Jd} \varphi_{l, d}~
g_{n-l, j-d}^{\alpha}~v^{m}_{l,d}\bigg\} \bigg].
\end{aligned}
\EN
Using $\mathcal{C}_{n,j}^{m+1}(\cdot)$ defined
in \eqref{eq:scheme_CD}, our numerical scheme at  the reference node ${\bf{x}} = (x_n, y_j, \tau_{m+1})$
can be rewritten in an equivalent form as follows
\EQA
\label{eq:scheme_GF}
0=
\mathcal{H}_{n,j}^{m+1}
\bigg(h, v_{n, j}^{m+1},
\left\{v_{l,d}^{m}\right\}_{\subalign{l\in \Nd\\d\in \Jd}}
 \bigg)
\equiv
\left\{
\begin{array}{lllllllllllll}
\mathcal{C}_{n,j}^{m+1}
\left(\cdot\right),
&
\quad {\bf{x}} \in \Omega{\myin},
\\
v_{n,j}^{m+1} - p(e^{x_{n}}, e^{y_j}),
&
\quad {\bf{x}} \in \Omega_{\tau_0},
\\
v_{n,j}^{m+1} - p(e^{x_n}, e^{y_j}) e^{-r\tau_{m+1}},
&
\quad {\bf{x}} \in \Omega_{\myout},
\end{array}
\right.
\ENA
where the sub-domains are defined in \eqref{eq:sub_domain_whole}, and
$p(\cdot, \cdot)$ is the terminal condition.

To demonstrate the consistency in viscosity sense of \eqref{eq:scheme_GF}, we need an intermediate result
given in Lemma~\ref{lemma:error_smooth} below.
\begin{lemma}[Two dimensional - $\Omega_{\myin}$]
\label{lemma:error_smooth}
Let $\phi$ be a test function in $\C{\Omega}$. For fixed  $\alpha \in \Al$ and ${\bf{x}}_{n,j}^{m} \in \Omega$, where
$n, j\in\mathbb{N}$ and $m \in \{1, \ldots, M\}$, with ${\phi_{n,j}^{m}} = \phi\big({\bf{x}}_{n,j}^{m}\big)$,
and for sufficiently small $h$, we have
\EQA
\label{eq:error_analysis_smooth}
\Delta x \Delta y
\mysum_{l\in\mathbb{N}^{\dagger}}^{d\in\mathbb{J}^{\dagger}} \varphi_{l, d} ~
    g_{n-l, j - d}^{\alpha}~  \phi_{l,d}^{m}~
=~
\phi_{n,j}^{m} + \Delta \tau \l[ \mathcal{L}_{\alpha}\phi \r]_{n,j}^{m} + \mathcal{O}(h^2).
 \ENA
Here,  $\big[\mathcal{L}_{\alpha}\phi \big]_{n,j}^{m} = \big[\mathcal{L}_{\alpha}\phi\big]\big({\bf{x}}_{n,j}^{m}\big)$,
and the differential operator $\Lcal_{\al}$ are defined in \eqref{eq:LcalX}.
\end{lemma}
\begin{proof}[Proof of Lemma~\ref{lemma:error_smooth}]
Starting from the discrete convolution on the left-hand-side (lhs) of \eqref{eq:error_analysis_smooth},
we need to recover an associated convolution integral of the form \eqref{eq:bkinteg}
which is posed on an infinite integration region.
{\zblue{Since for an arbitrary fixed $\tau_m$, $\phi(x, y, \tau_m)$ is not necessarily in $L^1(\Rbb^2)$, standard mollification techniques can be used to obtain a mollifier $\chi(x, y, \tau_m) \in L^1(\Rbb^2)$ which agrees with  $\phi(x, y, \tau_m)$ on $\mathbf{D}^{\dagger}$ \cite{Johnlee}, and has bounded derivatives up to second order across $\Rbb^2$.}}
For brevity, instead of $\chi(x, y, \tau_{m})$, we will write $\chi(x, y)$, which is a smooth bivariate function of $(x, y) \in \mathbb{R}^2$. We have
\begin{align}
\label{eq:error_smooth_b}
\Delta x \Delta y
\mysum_{l\in\mathbb{N}^{\dagger}}^{d\in\mathbb{N}^{\dagger}} \varphi_{l, d}~
    g_{n-l, j - d}^{\alpha}~  \phi_{l,d}^{m}~
&\overset{\text{(i)}}{=}
\iint_{\mathbf{D}^{\dagger}} g_{\alpha}\l(x_n - x, y_j - y;\Delta \tau\r)~ \phi(x, y) ~dx~dy + \mathcal{O}(h^2)
\nonumber
\\
&\overset{\text{(ii)}}{=}  \iint_{\Rbb^{\myblue{2}}} g_{\alpha}\l(x_n - x, y_j - y;\Delta \tau\r)~ \chi(x, y) ~dx~dy
 + \mathcal{O}(h^2) +  \mathcal{O}\big(he^{-1/h}\big)
 \nonumber
\\
&\overset{\text{(iii)}}{=}  [\chi*g](x_n, y_j) + \mathcal{O}(h^2)
\nonumber
\\
&=
\mathcal{F}^{-1}\l[\mathcal{F}\left[\chi\right]\!(\eta, \zeta)~ G\left(\eta, \zeta; \Delta \tau\right)\r](x_n, y_j) + \mathcal{O}(h^2).
 \end{align}
{\zblue{Here, in (i), the $\mathcal{O}(h^2)$ is due to error in the composite trapezoidal
rule, noting that $\phi$ has bounded derivatives of all orders in $\Omega$ because $\Omega$ is a bounded domain; in (ii) the boundary truncation error is $\mathcal{O}\big(he^{-1/h}\big)$, due to Lemma~\ref{lemma:truncation}, and in (iii) $[\chi * g]$ denotes the convolution of $\chi(x, y)$ and $g_{\alpha}(x, y;\Delta \tau)$.}}

In \eqref{eq:error_smooth_b},  with $\Psi(\eta, \zeta)$ given in \eqref{eq:G_closed}, expanding $G(\eta, \zeta; \Delta \tau) = e^{\Psi(\eta, \zeta)\Delta \tau}$
using a Taylor series with the Lagrange form for the remainder gives
\EQ
\label{eq:taylor}
G(\eta, \zeta; \Delta \tau)~=~ 1 + \Psi(\eta, \zeta) \Delta \tau + \mathcal{R}(\eta, \zeta) \Delta \tau^2,
\quad
\mathcal{R}(\eta, \zeta) = \frac{\Psi(\eta, \zeta)^2 e^{\xi \Psi(\eta, \zeta) }}{2},
\quad \xi \in (0, \Delta \tau).
\EN
Therefore,
\EQA
\label{eq:error_smooth_1}
\l[\chi*g\r](x_n, y_j)
&=&
\mathcal{F}^{-1}\l[\mathcal{F}\left[\chi\right]\!(\eta,\zeta)~\l(1 + \Psi(\eta,\zeta)\Delta \tau + \mathcal{R}(\eta, \zeta) \Delta \tau^2)\r) \r]\l(x_n, y_j\r)
\nonumber
\\
&=& \chi(x_n, y_j) + \Delta \tau \mathcal{F}^{-1}\l[\mathcal{F}\left[\chi\right]\!(\eta, \zeta)~\Psi\left(\eta, \zeta\right)\r](x_n, y_j)
\nonumber
\\
&& \qquad + \Delta \tau^2  \mathcal{F}^{-1}\l[\mathcal{F}\left[\chi\right]\!(\eta, \zeta) ~\mathcal{R}(\eta, \zeta) \r](x_n, y_j).
\ENA
Here, the first term in \eqref{eq:error_smooth_1}, namely $\chi(x_n, y_j) \equiv \chi(x_n, y_j, \tau_m) $ is simply $\phi_{n, j}^{m}$
by construction of $\chi(\cdot)$.
For the second term in \eqref{eq:error_smooth_1}, we focus on  $\mathcal{F}\left[\chi\right]\!(\eta, \zeta)~\Psi\left(\eta, \zeta\right)$.
Recalling the closed-form expression for $\Psi(\eta, \zeta)$ in \eqref{eq:G_closed}, we obtain
\begin{align*}
\mathcal{F}[\chi](\eta, \zeta) \Psi(\eta, \zeta) &=
\mathcal{F}[\chi](\eta, \zeta) \bigg(-\frac{\sigma_x^2 \eta^2}{2}
   -\frac{\sigma_y^2 \zeta^2}{2}
   + \big(r - \frac{\sigma_x^2}{2}\big)i\eta
   + \big(r - \frac{\sigma_y^2}{2}\big)i\zeta
   - \rho\sigma_x\sigma_y\eta\zeta
   - r\bigg)
\\
&\overset{\text{(i)}}{=}
\mathcal{F}\bigg[ \frac{\sigma_{x}^{2}}{2} \chi_{xx} + \frac{\sigma_{y}^{2}}{2} \chi_{yy} + (r-\frac{\sigma_x^2}{2})\chi_{x} + (r-\frac{\sigma_y^2}{2})\chi_{y} + \rho\sigma_x\sigma_y\chi_{xy} - r\chi \bigg](\eta, \zeta)
\\
&\overset{\text{(ii)}}{=} \mathcal{F}\l[ \mathcal{L}_{\alpha} \chi\r](\eta, \zeta).
\end{align*}
Here, (i) follows from the differentiation properties of the Fourier transform, which state that for the smooth test function $\chi(x, y)$,  we have
\begin{align*}
\mathcal{F}[\chi_x](\eta, \zeta) &= i\eta \mathcal{F}[\chi](\eta, \zeta),\quad ~~~
\mathcal{F}[\chi_y](\eta, \zeta) = i\zeta \mathcal{F}[\chi](\eta,\zeta),
\\
\mathcal{F}[\chi_{xx}](\eta, \zeta) &= -\eta^2 \mathcal{F}[\chi](\eta, \zeta),
\quad
\mathcal{F}[\chi_{yy}](\eta, \zeta) = -\zeta^2 \mathcal{F}[\chi](\eta, \zeta), \quad
\mathcal{F}[\chi_{xy}](\eta, \zeta) = \eta \zeta \mathcal{F}[\chi](\eta, \zeta).
\end{align*}
The equality in (ii) follows directly from the definition of the operator $\mathcal{L}_{\alpha}(\cdot)$ in \eqref{eq:LcalX}.
Therefore, the second term in \eqref{eq:error_smooth_1} becomes
\EQA
\label{eq:im_2}
\Delta \tau \mathcal{F}^{-1}\l[\mathcal{F}\left[\chi\right]\!(\eta, \zeta)~\Psi\left(\eta, \zeta\right)\r](x_n, y_j)
= \Delta \tau \l[ \mathcal{L}_{\alpha} \chi \r] (\x_{n, j}^{m})
= \Delta \tau \l[ \mathcal{L}_{\alpha} \chi \r]_{n, j}^{m}.
\ENA
For the third term $\Delta \tau^2  \mathcal{F}^{-1}\l[\mathcal{F}\left[\chi\right]\!(\eta, \zeta) ~\mathcal{R}(\eta, \zeta) \r](x_n, y_j)$ in \eqref{eq:error_smooth_1}, we have
\begin{align}
\label{eq:err_g}
&\Delta \tau^2 \l| \mathcal{F}^{-1}\l[ \mathcal{F} [\chi](\eta,\zeta)~\mathcal{R}(\eta, \zeta)\r]\!(x_n, y_j)\r|
\nonumber
\\
& \qquad\qquad =
\frac{\Delta \tau^2}{(2\pi)^2} \bigg| \iint_{\Rbb^2} e^{i(\eta x_n+\zeta y_j)} \mathcal{R}(\eta,\zeta) \bigg[ \iint_{\Rbb^2}  e^{-i(\eta x+\zeta y)} \chi(x,y)~dx~dy  \bigg] d \eta d \zeta \bigg|
\nonumber
\\
&\qquad \qquad\leq
\Delta \tau^2  \iint_{\Rbb^2} \l| \chi(x,y)  \r|~dxdy~ \iint_{\Rbb^2}
 \l| \mathcal{R}(\eta,\zeta) \r|~d\eta d\zeta.
\end{align}
Noting $\ds \mathcal{R}(\eta,\zeta) = \frac{\Psi(\eta,\zeta)^2 e^{\xi\Psi(\eta,\zeta)}}{2}$,
as shown in \eqref{eq:taylor}, where a closed-form expression for $\Psi(\eta,\zeta)$ is given in \eqref{eq:G_closed},
we obtain
\[
|\mathcal{R}(\eta, \zeta)| = \frac{|(\Psi(\eta, \zeta))^2|}{2} \exp\big(\xi\big(-\frac{\sigma_{x}^{2}\eta^2}{2} - \frac{\sigma_{y}^{2}\zeta^2}{2} - \rho\sigma_x\sigma_y\eta\zeta - r\big)\big).
\]
The term $|(\Psi(\eta, \zeta))^2|$ can be written in the form
$|\Psi|^2 = \sum_{\substack{k+q=4 \\ k, q \geq 0}} C_{kq} \eta^k \zeta^q$, where $C_{kq}$ are bounded coefficients.
This is a quartic polynomial in $\eta$ and $\zeta$. Furthermore, the exponent of exponential term is bounded by
\[
-\frac{1}{2} \sigma_{\myx}^2 \eta^2 -\frac{1}{2} \sigma_{\myy}^2 \zeta^2- \rho \sigma_{\myx}\sigma_{\myy}\eta\zeta-r \le 
-\frac{1}{2} \sigma_{\myx}^2 \eta^2-\frac{1}{2} \sigma_{\myy}^2 \zeta^2 + |\rho| \sigma_{\myx}\sigma_{\myy}|\eta\zeta|
\]
For $|\rho|< 1$, we have $|\rho| \sigma_{\myx}\sigma_{\myy}|\eta\zeta| < \frac{1}{2}(\sigma_{\myx}^2 \eta^2 + \sigma_{\myy}^2 \zeta^2)$. Therefore,   $\iint_{\Rbb^2} \l| \mathcal{R}(\eta,\zeta) \r|~d\eta d\zeta$ is bounded since
\[
\iint_{\Rbb^2} |\eta|^k |\zeta|^q~e^{-\frac{1}{2} \sigma_{\myx}^2 \eta^2-\frac{1}{2} \sigma_{\myy}^2 \zeta^2 - \rho \sigma_{\myx}\sigma_{\myy}\eta\zeta}~d\eta~d\zeta, \quad k+q=4,~ k, q \geq 0,
\]
is also bounded.
Together with  $\chi(x,y) \in L^1(\mathbb{R}^2)$,  the rhs of \eqref{eq:err_g} is  $\mathcal{O}(\Delta \tau^2)$, i.e.\
\begin{align}
\label{eq:err_g2}
\Delta \tau^2 \l| \mathcal{F}^{-1}\l[ \mathcal{F} [\chi](\eta,\zeta)~\mathcal{R}(\eta, \zeta)\r]\!(x_n, y_j)\r|
= \mathcal{O}(\Delta \tau^2).
\end{align}
Substituting \eqref{eq:im_2} and \eqref{eq:err_g2} into \eqref{eq:error_smooth_1},
noting \eqref{eq:error_smooth_b} and $\chi(x, y) = \phi(x, y)$ for $(x, y) \in \mathbf{D}^{\dagger}$~gives
\EQAS
\label{eq:error_smooth_2}
\Delta x \Delta y
\mysum_{l\in\mathbb{N}^{\dagger}}^{d\in\mathbb{N}^{\dagger}} \varphi_{l, d}~
    g_{n-l, j - d}^{\alpha}~  \phi_{l,d}^{m}~
=~
\phi_{n,j}^{m} + \Delta \tau \l[ \mathcal{L}_{\alpha} \phi \r]_{n,j}^{m}
+ \mathcal{O}(h^2).
\ENAS
This concludes the proof.
\end{proof}
To establish the consistency in the viscosity sense of our scheme as presented in \eqref{eq:scheme_GF}, it is essential to first examine the local consistency. This requires revisiting the operator
$F_{\text{in}}(\cdot)$ defined in \eqref{eq:Finn}. In the context of a discretized control set $ \mathcal{A}_h $, we introduce a modified operator that aligns with the piecewise constant control approach.
\begin{definition}
For a given a discretization parameter $h>0$, we define the operator $F_{\myin}^h$ for each control value $\alpha \in \mathcal{A}_h \subseteq \mathcal{A}$  as follows:
\EQ
\label{eq:Fh}
F_{\myin}^h (\cdot) \coloneqq F_{\myin}(\cdot), ~ \alpha \in \Al_h\subseteq \Al,
\EN
\end{definition}
Building on this definition, Lemma~\ref{lemma:Fh} presents an important  result regarding the approximation error bound when implementing the piecewise constant control technique.
\begin{lemma}
\label{lemma:Fh}
For any $\x \in \Omega_{\myin}$, and for a test function $\phi \in \C{\Omega}$ and a constant $\xi$, we have
\EQ
\label{eq:approx_F}
\left|F_{\myin}(\x, \phi(\x), D \phi(\x), D^{2} \phi(\x))
- F^{h}_{\myin}(\x, (\phi + \xi) (\x), D (\phi+\xi)(\x), D^{2} (\phi+\xi))\right|
\le Ch + r\xi,
\EN
where $C>0$ is a bounded constant independently of $h$.
\end{lemma}
\begin{proof}[Proof of Lemma~\ref{lemma:Fh}]
By insertion and the triangle inequality, the lhs of \eqref{eq:approx_F} is bounded as follows
\begin{align}
\label{eq:fin_finh}
\l|F_{\myin}(\cdot) - F^{h}_{\myin}(\cdot)\r|
&\leq \bigg|\sup_{\alpha\in\Al_h}\Lcal_\alpha (\phi+\xi)-\sup_{\alpha\in\Al}\Lcal_\alpha (\phi+\xi)\bigg|+\bigg|\sup_{\alpha\in\Al}\Lcal_\alpha (\phi+\xi)-\sup_{\alpha\in\Al}\Lcal_\alpha (\phi)\bigg|
\nonumber\\
&= \bigg|\sup_{\alpha\in\Al_h}\Lcal_\alpha \phi-\sup_{\alpha\in\Al}\Lcal_\alpha \phi\bigg|+r\xi.
\end{align}
Due to the compactness of $\Al$, the supremum of $\Lcal_{\alpha}(\phi)$ is attainable at, say $ \alpha^{*}\equiv (\sigma_{\myx}^*,\sigma_{\myy}^*,\rho^*)\in \Al$.  By \eqref{eq:compact}, there exists $\al'^*\equiv(\sigma'^{*}_{\myx},\sigma'^{*}_{\myy},\rho'^{*})\in \Al_h$ with $\|\al'^*-\al^*\|_2\leq h$.
Therefore, the first term in \eqref{eq:fin_finh} becomes
$\l|\sup_{\alpha\in\Al_h}\Lcal_\alpha (\phi)-\sup_{\alpha\in\Al}\Lcal_\alpha (\phi)\r| = \l|\Lcal_{\alpha'^*} \phi-\Lcal_{\alpha^{*}} \phi\r| = \ldots$
\begin{align*}
\ldots&\overset{\text{(i)}}{=}\l|\frac{1}{2}((\sigma'^{*}_{\myx})^2-(\sigma_{\myx}^*)^2)(\phi_{\myx\myx}-\phi_{\myx})+\frac{1}{2}((\sigma'^{*}_{\myy})^2-(\sigma_{\myy}^*)^2)(\phi_{\myy\myy}-\phi_{\myy})+(\rho'^{*}\sigma'^{*}_{\myx}\sigma'^{*}_{\myy}-\rho^*\sigma_{\myx}^*\sigma_{\myy}^*)\phi_{\myx\myy}\r|
\nonumber\\
&\leq \frac{1}{2}\l|((\sigma'^{*}_{\myx})^2-(\sigma_{\myx}^*)^2)\r|(|\phi_{\myx\myx}|+|\phi_{\myx}|)+\frac{1}{2}\l|((\sigma'^{*}_{\myy})^2-(\sigma_{\myy}^*)^2)\r|\l(|\phi_{\myy\myy}|+|\phi_{\myy}|\r)+\l|(\rho'^{*}\sigma'^{*}_{\myx}\sigma'^{*}_{\myy}-\rho^*\sigma_{\myx}^*\sigma_{\myy}^*)\r||\phi_{\myx\myy}|
\nonumber\\
&\overset{\text{(ii)}}{\leq} Ch,
\end{align*}
where $C>0$ is a bounded constant independently of $h$.
Here, in (i), we first insert the $\Lcal_\al(\cdot)$ operator \eqref{eq:LcalX} and then combine similar terms; (ii) due to the
$\|\al'^*-\al^*\|_2\leq h$, together with the compactness of the admissible control set $\Al$ and the fact that the test function $\phi$ has continuous bounded derivatives in $\Omega$ since $\Omega$ is bounded. This concludes the proof.
\end{proof}

\noindent Below, we state the key supporting lemma related to local consistency of our numerical scheme \eqref{eq:scheme_GF}.
\begin{lemma} [Local consistency]
\label{lemma:consistency}
Suppose that (i) the discretization parameter $h$ satisfies \eqref{eq:dis_parameter}.
Then, for any test function $\phi \in \C{\myblue{\Omega}}$,
with  $\phi_{n, j}^{m} = \phi\l({\bf{x}}_{n, j}^{m}\r)$ and 
${\bf{x}} \coloneqq (x_n, y_j, \tau_{m+1}) \in \Omega$, and for a sufficiently small $h$,  we have
\begin{linenomath}
\postdisplaypenalty=0
\EQA
\label{eq:lemma_1}
\mathcal{H}_{n, j}^{m+1}
\bigg(h, \phi_{n, j}^{m+1} + \xi,
\l\{\phi_{l, d}^{m}+\xi \r\}_{\subalign{l\in \Nd\\d\in \Jd}}
\bigg)
=
\left\{
\begin{array}{llllllllllr}
F_{\myin}^h\l(\cdot\r)
&\!\!\!\!\!+~
c(\x)\xi
+ \mathcal{O}(h),
&&
{\bf{x}} \in \Omega_{\myin},
\\
F_{\myout}\l(\cdot\r)
&\!\!\!\!\!+~ c(\x)\xi,
&&
{\bf{x}} \in \Omega_{\myout};
\\
F_{\tau_0}\l(\cdot\r)
&\!\!\!\!\!+~ c(\x)\xi,
&&
{\bf{x}} \in \Omega_{\tau_0}.
\end{array}\right.
\ENA
\end{linenomath}
Here, $\xi$ is a constant, and $c(\cdot)$ is a bounded function
satisfying $|c({\bf{x}})| \le \max(r, 1)$ for all ${\bf{x}}~\in~\Omega$.
The  operators $F_{\myin}^h(\cdot)$, defined in \eqref{eq:Fh}, and
$F_{\myout}(\cdot)$, and $F_{\tau_0}(\cdot)$, respectively defined in
\eqref{eq:fout}-\eqref{eq:ftau0},
are functions of $\l({\bf{x}}, \phi\l({\bf{x}}\r),
D\phi\l({\bf{x}}\r), D^2 \phi\l({\bf{x}}\r)\r)$.
\end{lemma}
\begin{proof}[Proof of Lemma~\ref{lemma:consistency}]
{\zblue{Since  $\phi \in \C{\Omega}$ and $\Omega$ is bounded,
$\phi$ has continuous and bounded  derivatives of up to second-order
in $\Omega$.}}
We now show that the first equation of \eqref{eq:lemma_1} is true, that is,
\EQAS
\mathcal{H}_{n, j}^{m+1}(\cdot) &=& \mathcal{C}_{n, j}^{m+1} \l(\cdot\r)
= F_{\myin}^h\l({\bf{x}}, \phi\l({\bf{x}}\r)\r)
+ c(\x)\xi + \mathcal{O}(h)
\\
&&\qquad \qquad
\text{if}~
x_{\min} < x_n < x_{\max},~
y_{\min} < y_j < y_{\max},~
0 < \tau_{m+1} \le T.
\ENAS
where operators $\mathcal{C}_{n, j}^{m+1}(\cdot)$ is defined in \eqref{eq:scheme_CD}.
In this case,  operator $\mathcal{C}_{n, j}^{m+1}(\cdot)$ is written as follows
\begin{linenomath}
\postdisplaypenalty=0
\begin{align}
\label{eq:all_c*}
\mathcal{C}_{n, j}^{m+1} (\cdot) &= \frac{1}{\Delta \tau}\bigg[ \phi_{n, j}^{m+1} + \xi
-
\max_{\alpha \in \Al_h} ~ \bigg\{\Delta x \Delta y \mysum_{l \in \Nd}^{d \in \Jd} \varphi_{l, d}~
g_{n-l, j-d}^{\alpha}~(\phi^{m}_{l,d}+ \xi)\bigg\} \bigg]
\\
&= \frac{1}{\Delta \tau}\bigg[ \phi_{n, j}^{m+1}
-
\max_{\alpha \in \Al_h} ~ \bigg\{\Delta x \Delta y \mysum_{l \in \Nd}^{d \in \Jd} \varphi_{l, d}~
g_{n-l, j-d}^{\alpha}~\phi^{m}_{l,d}\bigg\}
+ \xi \bigg( 1 - \max_{\alpha \in \Al_h} \Delta x \Delta y \mysum_{l \in \Nd}^{d \in \Jd} \varphi_{l, d}~
g_{n-l, j-d}^{\alpha}\bigg)
\bigg]
\nonumber
\\
&\overset{\text{(i)}}{=}
\frac{\phi_{n, j}^{m+1} - \phi_{n, j}^{m}}{\Delta \tau}
- \max_{\alpha \in \Al_h} \l[ \mathcal{L}_{\alpha} \phi \r]_{n,j}^{m}
+ \frac{\xi}{\Delta \tau} \bigg( 1 - \max_{\alpha \in \Al_h} \bigg\{\Delta x \Delta y \mysum_{l \in \Nd}^{d \in \Jd} \varphi_{l, d}~
g_{n-l, j-d}^{\alpha}\bigg\}\bigg)
+ \mathcal{O}(h).
\nonumber
\end{align}
\end{linenomath}
Here,  (i) is due to use Lemma~\ref{lemma:error_smooth}.
Regarding the term $\frac{\xi}{\Delta \tau} \big( 1 - \max_{\alpha \in \Al_h} \{\cdot\}\big)$,
suppose that $\max_{\alpha \in \Al_h}\{\cdot\}$ is attainable at $\alpha'$. Then, $| 1 - \max_{\alpha \in \Al_h} \{\cdot\}| = \bigg|1 - \Delta x \Delta y \mysum_{l \in \Nd}^{d \in \Jd} \varphi_{l, d}~
g_{n-l, j-d}^{\alpha'}\bigg| \le \ldots$
\EQ
\label{eq:termq}
\ldots
\le
\bigg|1 - \iint_{\Rbb^2}g_{\alpha'}(x_n-x,y_j - y; \Delta \tau)dxdy\bigg|
+
\bigg|\iint_{\Rbb^2}g_{\alpha'}(\cdot,\cdot; \Delta \tau)dxdy -
\Delta x \Delta y \mysum_{l \in \Nd}^{d \in \Jd} \varphi_{l, d}~
g_{n-l, j-d}^{\alpha'}
\bigg|.
\EN
The first term of \eqref{eq:termq} is simply $1 - e^{-r\Delta \tau} = r \Delta \tau + \mathcal{O}(h^2)$,
noting $\iint_{\Rbb^2}g_{\alpha}(\cdot, \cdot; \Delta \tau) dxdy = e^{-r\Delta \tau}$ for any $\alpha \in \Al$.
The second term  of \eqref{eq:termq} is simply $\mathcal{O}(h^2) + \mathcal{O}(he^{-1/h})=\mathcal{O}(h^2)$ due to numerical integration error and boundary truncation error, as noted earlier. With this in mind, we have
\[
\frac{\xi}{\Delta \tau} \bigg( 1 - \max_{\alpha \in \Al_h} \Delta x \Delta y \mysum_{l \in \Nd}^{d \in \Jd} \varphi_{l, d}~
g_{n-l, j-d}^{\alpha}\bigg) =   r\xi + \mathcal{O}(h).
\]
Substituting this result into \eqref{eq:all_c*} gives
 \begin{linenomath}
\postdisplaypenalty=0
\begin{align*}
\mathcal{C}_{n, j}^{m+1} (\cdot) &=
\frac{\phi_{n, j}^{m+1} - \phi_{n, j}^{m}}{\Delta \tau}
- \max_{\alpha \in \Al_h} \l[ \mathcal{L}_{\alpha} \phi \r]_{n,j}^{m}
+ r\xi + \mathcal{O}(h)
\overset{\text{(i)}}{=}
\bigg[ \phi_{\tau} -   \max_{\alpha \in \Al_h} \mathcal{L}_{\alpha} \phi  \bigg]_{n,j}^{m+1} + r\xi
+ \mathcal{O} (h).
\end{align*}
\end{linenomath}
Here, in (i), we use $(\phi_{\tau})_{n,j}^{m} = (\phi_{\tau})_{n,j}^{m+1} + \mathcal{O}\l(h\r)$,
$(\phi_z)_{n,j}^{m} = (\phi_z)_{n,j}^{m+1} + \mathcal{O}\l(h\r)$, $z \in \{x, y\}$, and for the cross derivative term
$(\phi_{xy})_{n,j}^{m} = (\phi_{xy})_{n,j}^{m+1} + \mathcal{O}\l(h\r)$.
This proves the first equation in (\ref{eq:lemma_1}).
The remaining equations in (\ref{eq:lemma_1}) can be proved using similar arguments with the first equation,
and hence omitted for brevity. This concludes the proof.
\end{proof}
We now verify the consistency of the numerical scheme $\mathcal{H}_{n, j}^{m+1}(\cdot)$ as defined in \eqref{eq:scheme_GF}.
We first define the notion of consistency in the viscosity sense  below.
\begin{definition} [Consistency in viscosity sense]
\label{definition:consistency_viscosity}
Suppose the discretization parameter $h$ satisfies~\eqref{eq:dis_parameter}. The numerical scheme \eqref{eq:scheme_GF}
is consistent in the viscosity sense if, for all ${\bf{\hat{x}}} = (\hat{x}, \hat{y}, \hat{\tau}) \in \Omega$,
and for any $\phi\in \C{\myblue{\Omega}}$, with  $\phi_{n, j}^{m} = \phi\big({\bf{x}}_{n, j}^{m}\big)$ and {\bf{x}}~=~$(x_n, y_j, \tau_{m+1})$, we have both of the following
\EQA
\limsup_{\subalign{h \to 0, & ~  {\bf{x}} \to {\bf{\hat{x}}} \\ \xi &\to 0}}
\mathcal{H}_{n, j}^{m+1}
\bigg(h, \phi_{n, j}^{m+1}+ \xi,
\l\{\phi_{l, d}^{m}+\xi \r\}_{\subalign{l\in \Nd\\d\in \Jd}} \bigg)
\leq
\left(F_{\Omega}\right)^* \l(
              {\bf{\hat{x}}}, \phi({\bf{\hat{x}}}), D\phi({\bf{\hat{x}}}), D^2 \phi({\bf{\hat{x}}})\r),
\label{eq:consistency_viscosity_1}
\\
\liminf_{\subalign{h \to 0, & ~ {\bf{x}} \to {\bf{\hat{x}}} \\ \xi &\to 0}}
\mathcal{H}_{n, j}^{m+1}
\bigg(h, \phi_{n, j}^{m+1}+ \xi,
\l\{\phi_{l, k}^{m}+\xi \r\}_{\subalign{l\in \Nd\\d\in \Jd}} \bigg)
\geq
\left(F_{\Omega}\right)_*\l(
              {\bf{\hat{x}}}, \phi({\bf{\hat{x}}}), D\phi({\bf{\hat{x}}}), D^2 \phi({\bf{\hat{x}}})
             \r).
\label{eq:consistency_viscosity_2}
\ENA
Here, $\left(F_{\Omega}\right)^*(\cdot)$ and  $\left(F_{\Omega}\right)_*(\cdot)$ respectively are
the u.s.c. and the l.s.c. envelop of the operator $F_{\Omega}(\cdot)$ defined in \eqref{eq:Fomega_def}.
\end{definition}
\noindent Below, we state and prove the main lemma on consistency of the numerical scheme \eqref{eq:scheme_GF}.
\begin{lemma} [Consistency]
\label{lemma:consistency_viscosity}
Suppose the discretization parameter $h$ satisfies~\eqref{eq:dis_parameter}.
Then,  the numerical  scheme \eqref{eq:scheme_GF} is consistent with the two-factor uncertain volatility pricing problem \eqref{def:uvm_def} in $\Omega$ in the sense of Definition~\ref{definition:consistency_viscosity}.
\end{lemma}


\begin{proof}[Proof of Lemma~\ref{lemma:consistency_viscosity}]
We first prove \eqref{eq:consistency_viscosity_1}.
Let ${\bf{\hat{x}}} \equiv (\hat{x}, \hat{y},  \hat{\tau})$ be an arbitrary, but fixed, point
in  $\Omega$. Consider $h \to 0$. There exists sequences of $\{h_i\}$, $\{m_i\}$, $\{\x_i\}$, and $\{\xi_i\}$,
such that
\EQA
\label{eq:consistency_1_F}
\text{as}~~i \to \infty, ~~~ h_i \to 0,~ \xi_i \to 0, ~{\bf{x}}_i \equiv (x_{n_i}, y_{j_i}, \tau_{m_i+1}) \to {\bf{\hat{x}}} \equiv (\hat{x}, \hat{y},  \hat{\tau}),
\ENA
and
\EQ
\label{eq:consistency_2_F}
\limsup_{i \to \infty} \mathcal{H}_{n_i, j_i}^{m_i+1}\l(h_i, \phi_{n_i, j_i}^{m_i+1} +\xi_i, \big\{\phi_{l_i, d_i}^{m_i} +\xi_i\big\}\r)
=\limsup_{\subalign{h \to 0, &~  {\x} \to {\xh}
\\ \xi &\to 0}} \mathcal{H}_{n, j}^{m+1}\l(h, \phi_{n, j}^{m+1} +\xi, \big\{\phi_{l, d}^{m} +\xi\big\}\r).
\EN
Now, we consider the case $\xh \in \Omega_{\myin}$. According to the first equation of \eqref{eq:lemma_1} (Lemma~\ref{lemma:consistency}), we have
\EQ
\label{eq:consistency_viscosity_4}
 \mathcal{H}_{n_i,j_i}^{m_i+1} \l(h_i, \phi_{n_i, j_i}^{m_i+1} + \xi_i,
\l\{\phi_{l_i, k_i}^{m_i}+\xi_i \r\}\r)
=
F_{\myin}^{h_i} \l({\bf{x}}_i, \phi\l({\bf{x}}_i\r), D \phi\l({\bf{x}}_i\r),
D^2\phi\l({\bf{x}}_i\r)\r)
+ r\xi_i + \mathcal{O}(h_i)
\EN
Using \eqref{eq:approx_F} and \eqref{eq:consistency_viscosity_4} gives, for each $i$,
\EQ
\label{eq:H2F_F}
\big|F(\x_i, \phi(\x_i), \cdot, \cdot)
-
 \mathcal{H}_{n_i,j_i}^{m_i+1} \big(h_i, \phi_{n_i, j_i}^{m_i+1} + \xi_i,
\big\{\phi_{l_i, k_i}^{m_i}+\xi_i \big\}\big)\big|
\le
C_i h_i + (r + c(\x_i)) \xi_i + \mathcal{O}(h_i)
\EN
Here, $C_i>0$ is a bounded constant and  $|c(\x_i)|\le \max(r, 1)$ for all $i$.
Thus, from \eqref{eq:H2F_F}, we have
\EQA
\label{eq:consistency_3_F}
 \mathcal{H}_{n_i,j_i}^{m_i+1} \big(h_i, \phi_{n_i, j_i}^{m_i+1} + \xi_i,
\big\{\phi_{l_i, k_i}^{m_i}+\xi_i \big\}\big)
\le
F\big(\x_i, \phi(\x_i), \cdot, \cdot\big)
+ C_i h_i + (r + c(\x_i)) \xi_i + \mathcal{O}(h_i)
\ENA
Combining \eqref{eq:consistency_2_F} and \eqref{eq:consistency_3_F}, with continuity of
$F(\cdot)$, we obtain
\begin{align*}
&\limsup_{\subalign{h \to 0, &~  {\x} \to {\xh}
\\ \xi &\to 0}} \mathcal{H}_{n, j}^{m+1}\l(h, \phi_{n, j}^{m+1} +\xi, \big\{\phi_{l, d}^{m} +\xi\big\}\r)
= \limsup_{i \to \infty} \mathcal{H}_{n_i, j_i}^{m_i+1}\l(h_i, \phi_{n_i, j_i}^{m_i+1} +\xi_i, \big\{\phi_{l_i, d_i}^{m_i} +\xi_i\big\}\r)
\\
&\qquad \qquad \qquad \qquad \leq
\limsup_{i \to \infty}
F\left(\x_i, \phi(\x_i), D\phi(\x_i), D^{2} \phi(\x_i)\right)
 +  \limsup_{i \to \infty} (C_ih_i + (r + c(\x_i)) \xi_i)
\nonumber
\\
&\qquad \qquad \qquad\qquad  =F^* \l(\xh, \phi(\xh), D\phi(\xh), D^2 \phi(\xh) \r).
\end{align*}
This proves \eqref{eq:consistency_viscosity_1} for $\xh \in \Omega_{\myin}$.
The case \eqref{eq:consistency_viscosity_1} for other sub-domains as well as
the case \eqref{eq:consistency_viscosity_2} can be proved in a similar fashion.
This concludes the proof.
\end{proof}

\subsection{Monotonicity}
We present a result on the monotonicity of scheme \eqref{eq:scheme_GF}.
\begin{lemma}{(Monotonicity)}
\label{lemma:mon}
Scheme \eqref{eq:scheme_GF} satisfies
 \EQA
\label{eq:mon}
\mathcal{H}^{m+1}_{n,j}\l(h,  v^{m+1}_{n,j}, \l\{u^{m}_{l,d}\r\}\r)\leq\mathcal{H}^{m+1}_{n,j}\l(h, v^{m+1}_{n,j}, \l\{z^{m}_{l,d}\r\}\r)
\ENA
for bounded $\l\{u^{m}_{l,d}\r\}$ and $\l\{z^{m}_{l,d}\r\}$ having $\l\{u^{m}_{l,d}\r\}\geq \l\{z^{m}_{l,d}\r\}$,
where the inequality is understood in the component-wise sense.
\end{lemma}
\begin{proof}[Proof of Lemma~\ref{lemma:mon}]
Since scheme \eqref{eq:scheme_GF} is defined case-by-case, to establish \eqref{eq:mon}, we will show that each
case satisfies \eqref{eq:mon}. It is straightforward that the scheme satisfies \eqref{eq:mon} in
${\Omega}_{\tau_0}$) and $ {\Omega}_{{\myout}}$. Now we establish
that  $\mathcal{C}_{n, j}^{m+1}\l(\cdot\r)$, as defined in \eqref{eq:scheme_CD}
for $\Omega_{\myin}$, also satisfies \eqref{eq:mon}. We have
\begin{align}
\label{eq:mon_p}
&\mathcal{C}^{m+1}_{n,j}\l(h,  v^{m+1}_{n,j}, \l\{u^{m}_{l,d}\r\}\r)-\mathcal{C}^{m+1}_{n,j}\l(h, v^{m+1}_{n,j}, \l\{z^{m}_{l,d}\r\}\r)\nonumber
\\
&\quad=\frac{1}{\Delta \tau}\bigg[\max_{\al} \Delta x \Delta y \mysum_{l\in\mathbb{N}^{\dagger}}^{d\in\mathbb{J}^{\dagger}}\varphi_{l, d} ~g^\al_{n-l,j-d}~z^{m}_{l,d}-\max_{\al} \Delta x \Delta y \mysum_{l\in\mathbb{N}^{\dagger}}^{d\in\mathbb{J}^{\dagger}}\varphi_{l, d} ~g^\al_{n-l,j-d}~u^{m}_{l,d}\bigg]\nonumber
\\
&\quad \overset{\text{(i)}}{\le} \frac{1}{\Delta \tau}\l[\max_{\al} \Delta x \Delta y \mysum_{l\in\mathbb{N}^{\dagger}}^{d\in\mathbb{N}^{\dagger}}\varphi_{l, d} ~g^\al_{n-l,j-d}~\l(z^{m}_{l,d}-u^{m}_{l,d}\r)\r] \leq 0.
\end{align}
Here, (i) is due to the fact that,  $\max_{\alpha \in \Al} f_1(\alpha) - \max_{\alpha \in \Al} f_2(\alpha) \le
\max_{\alpha} (f_1(\alpha) - f_2(\alpha))$ for  two real-valued functions $f_1, f_2$ of $\alpha$.
This concludes the proof.
\end{proof}

\label{sssec:FRC}
\begin{theorem} [Convergence to viscosity solution in $\Omega_{\myin}$]
\label{thm:convergence}
Suppose that all the conditions for Lemmas~\ref{lemma:stability}), \ref{lemma:consistency}
and \ref{lemma:mon} are satisfied. Our scheme \eqref{eq:scheme_GF} converges in $\Omega_{\myin}$ to the unique continuous viscosity solution of the two-factor uncertain volatility model pricing problem given in Definition \eqref{def:vis_def_common}.
\end{theorem}
\begin{proof}[Proof of Theorem~\ref{thm:convergence}]
Our scheme is $\ell_\infty$-stable (Lemma \ref{lemma:stability}),
and consistent in the viscosity sense (Lemma \ref{lemma:consistency})
and monotone (Lemma~\ref{lemma:mon}). Since a strong comparison holds in $\Omega_{\myin}$ (Remark~\ref{rm:strong}),
by \cite{barles-souganidis:1991}, convergence in $\Omega_{\myin}$ to the unique continuous vicosity solution of the HJB equation
is ensured.
\end{proof}

\section{Numerical experiments}
\label{sec:Numeri_exp}
This section presents the selected numerical results of our monotone piecewise constant control integration method (MPCCI)
applied to the two-factor uncertain volatility model pricing problem.
The modelling parameters for the tests carried out are given in Table~\ref{tab:parameter02}, reproduced from \cite{MaForsyth2015}[Table~3]. We note that specific ranges for $\sigma_{\myx}$, $\sigma_{\myy}$, and the correlation coefficient  $\rho$ are given.
\subsection{Preliminary}
Prior to initiating our experiments, it is essential to define a sufficiently large computational domain, guided by
the boundary truncation error bound provided in Lemma~\ref{lemma:truncation}. Specifically, we follow steps outlined
in Remark~\ref{rm:bd_e} to determine $x^{\dagger}_{\mymin}$, $x^{\dagger}_{\mymax}$, $y^{\dagger}_{\mymin}$ and $y^{\dagger}_{\mymax}$. In particular, in \eqref{eq:eps_b}, $\epsilon =  10^{-10}$ is used.
With the model parameters given Table~\ref{tab:parameter02}, this procedure gives $x_{\min} = \ln(X_0)-1.2$, $x_{\max} = \ln(X_0)+1.2$, $y_{\min} = \ln(Y_0)-1.2$, $y_{\max} = \ln(Y_0)+1.2$.
Furthermore, for $z \in \{x, y\}$, the values of $z_{\min}^{\dagger}$, $z_{\max}^{\dagger}$, $z_{\min}^{\ddagger}$ and $z_{\max}^{\ddagger}$ are determined as in \eqref{eq:w_choice_green_jump_form}-\eqref{eq:w_choice_green_jump_form_dd}.
Extensive testing indicates that larger intervals have negligible impact on numerical solutions, whereas smaller domains exhibit minor variations. These findings are numerically validated in Subsection~\ref{ssc:padding}
Unless noted otherwise, the specifics of mesh size and timestep refinement levels utilized in all experiments are detailed in Table~\ref{tab:step01}.

All numerical experiments were performed on a system equipped with an Intel Core i7-11700 CPU (11th Gen, 2.50 GHz, 8 cores / 16 threads) and 32GB of RAM (dual-channel, 3200 MHz). The system operates on Windows 11 (64-bit) with a 512GB SSD. No GPU acceleration was used. The implementation was carried out in MATLAB R2022b (Version 9.13) with the Statistics and Machine Learning Toolbox.

\begin{minipage}{0.55\textwidth}
\strut\vspace*{-\baselineskip}\newline
\center
\begin{tabular}{cl|cl}
\hline
Parameter                         & Value/    &Parameter                         & Value/ \\
                                & Range    &                         & Range  \\
\hline
$T$                               & 0.25 (years) & $X_0$                     & 40
\\
$r$                               & 0.05  &$Y_0$                     & 40
\\
$\sigma_{x}$                      & [0.3, 0.5] & $K$                               & 40
\\
$\sigma_{y}$                      & [0.3, 0.5] &$K_1$                             & 34
\\
$\rho$                            & [0.3, 0.5] & $K_2$                             & 46
\\
\hline
\end{tabular}
\captionof{table}{Model parameters used in numerical experiments for two-factor uncertain volatility model-
reproduced from \cite{MaForsyth2015}~Table~3.}
\label{tab:parameter02}
\end{minipage}
\begin{minipage}{0.45\textwidth}
\strut\vspace*{-\baselineskip}\newline
\center
\begin{tabular}{crrrr}
\hline
Level & $N$ &  $J$                & $M$  &$Q$    \\
      & ($x$) & ($y$)               & ($\tau$) & ($\alpha$)  \\ \hline
0     & $2^{7}$ & $2^{7}$          & 50    & 8 \\
1     & $2^{8}$ & $2^{8}$          & 100   & 24   \\
2     & $2^{9}$ & $2^{9}$        & 200     & 56   \\
3     & $2^{10}$& $2^{10}$        & 400    & 120  \\
4     & $2^{11}$&$2^{11}$        & 800     & 248   \\
\hline
\end{tabular}
\captionof{table}{Grid and timestep refinement levels for numerical tests.
}
\label{tab:step01}
\end{minipage}

\vspace*{+0.25cm}
Our MPCCI numerical prices are verified against those produced by:
(i) closed-form solutions (for certain European rainbow options),
(ii) FD methods reported in the literature, particularly the unconditionally monotone FD method of \cite{MaForsyth2015},
(iii) tree-grid (TG) methods of \cite{kossaczky20192d},
and (iii) Monte Carlo (MC) simulation. The Monte Carlo validation is carried out in two steps
\begin{enumerate}
 \item Step 1: we solve the two-factor uncertain volatility pricing problem using the proposed MPCCI on a fine computational grid (comprising of $2^{10}$ $x$-nodes, $2^{10}$ $y$-nodes, and  400 timesteps).
     At each time-$\tau_m$, we store the optimal controls or all pair of discrete  states $(x_n, y_j)$,
     denoted as $\big\{(\al^*)_{n, j}^{m}\big\} \equiv \big\{(\sigma_{\myx}^*, \sigma_{\myy}^*,\rho^*)_{n, j}^{m}\big\}$,
     where $n\in  \mathbb{N}^{\dagger}$, $j \in \mathbb{J}^{\dagger}$, and $m = 0, \ldots, M$.

 \item Step 2: we conduct Monte Carlo simulations of the 2D dynamics \eqref{eq:2SDE_GBM} from $t=0$ to $t=T$, following the stored MPCCI-computed optimal controls. For a given pair of simulated values of $X$ and $Y$,
     linear interpolation, if necessary, is used to determine the control.
     Specifically, for the $\gamma$-th $X$ and $Y$ simulated values at time-$\tau_m$, denoted by $\hat{X}_{m}^{(\gamma)}$ and $\hat{Y}_{m}^{(\gamma)}$,
     and given  $x_{n'} \le \hat{X}_m^{(\gamma)}\le x_{n'+1}$ and  $y_{j'} \le \hat{Y}_m^{(\gamma)}\le y_{j'+1}$, we interpolate the optimal control
    $(\al^*)_{n', j'}^{m}$, $(\al^*)_{n'+1, j}^{m}$, $(\al^*)_{n', j'+1}^{m}$, and $(\al^*)_{n'+1, j'+1}^{m}$
    to determine the volatilities $\tilde{\sigma}_x^m$ and $\tilde{\sigma}_y^m$  and the correlation coefficient $\tilde{\rho}_m$ for the interval $[\tau_m, \tau_{m+1}]$. The Euler-Maruyama discretization is then applied for each the interval $[\tau_m, \tau_{m+1}]$ as follows:
     \begin{align*}
     \hat{X}_{m+1}^{(\gamma)} &= \hat{X}_m^{(\gamma)} \l(1 + r \Delta t + \tilde{\sigma}_x^m  \sqrt{\Delta t}~\xi_x^{(\gamma)}\r),
     \\
    \hat{Y}_{m+1}^{(\gamma)} &= \hat{Y}_m^{(\gamma)} \l(1+  r \Delta t + \tilde{\sigma}_y^m  \sqrt{\Delta t}~(\tilde{\rho}_m \xi_x^{(\gamma)} + \sqrt{1- (\tilde{\rho}_m)^2}~\xi_y^{(\gamma)})\r),
     \end{align*}
     where $\xi_x^{(\gamma)}$ and $\xi_y^{(\gamma)}$ are independent standard normal random variables.
     The option value is approximated by $\frac{e^{-rT}}{\Gamma} \sum_{\gamma=1}^{\Gamma} p\big(\hat{X}_{M}^{(\gamma)}, \hat{Y}_{M}^{(\gamma)}\big)$,
     with $p(\cdot, \cdot)$ as the payoff function, using a total of  $\Gamma= 10^6$ simulation paths.

 \end{enumerate}
\subsection{Validation examples}
\subsubsection{European call options}
\label{sssec:Eu}
Our first test case evaluates a European call option on the maximum of two assets, as described in \cite{MaForsyth2015}.
The payoff function $p(e^x,e^y)$ is given by
\EQA
\label{eq:BSO_P1}
p(e^x,e^y) = \max(\max(e^x,e^y)-K,0), \quad K>0.
\ENA
We consider the worst-case value for the short position, for which an analytical solution exists, as noted in \cite{MaForsyth2015}. Specifically, since the payoff function \eqref{eq:BSO_P1} is convex and convexity is preserved \cite{janson2004preservation}, the worst-case price of the short position is attained at the fixed parameters $\sigma_x^* = \sigma^{x}{\mymax}$, $\sigma_y^* = \sigma^{y}{\mymax}$, and $\rho^* = \rho_{\mymin}$. The exact option price can be computed analytically using the closed-form expression from \cite{Stulz1982}. Using the parameters from Table~\ref{tab:parameter02},
where $\sigma_x^* = \sigma_y^* = 0.5$ and $\rho^* = 0.3$, we obtain the closed-form solution of 6.84769986, accurate to 8 decimal places.

\vspace*{-0.5cm}
\begin{minipage}{0.5\textwidth}
\strut\vspace*{-\baselineskip}
\flushleft
\begin{linenomath}
\begin{table}[H]
\center
\begin{tabular}{llll}
\hline
 Level     & Price        & Abs.\ error      & Ratio           \\
\hline
     \quad   0         &6.84492756 &2.77e-03 &            \\
    \quad    1         &6.84700690 &6.93e-04 & 4.0        \\
    \quad 2            &6.84752662 &1.73e-05  &4.0         \\
      \quad  3         &6.84765654 &4.33e-05  &4.0         \\
   \quad  4            &6.84768902 &1.08e-05  &4.0          \\
\hline
Ref.\ \cite{Stulz1982}&  6.84769986 & &
\\     \hline
 MC: 95\%-CI                 & \multicolumn{3}{l}{[6.8319, 6.8618]}   \\ \hline
\end{tabular}
\captionof{table}{Convergence study for a European call option on the maximum of two assets under the two-factor uncertain volatility model (\underline{worst-case}, composite trapezoidal rule). Payoff given by \eqref{eq:BSO_P1}. The reference value is the closed-form solution from \cite{Stulz1982} with $\sigma_x^* = \sigma_y^* = 0.5$ and $\rho^* = 0.3$.}
\label{tab:result01}
\end{table}
\end{linenomath}
\end{minipage}
\begin{minipage}{0.5\textwidth}
\vspace*{-0.9cm}
\begin{figure}[H]
\includegraphics[width=0.95\linewidth]{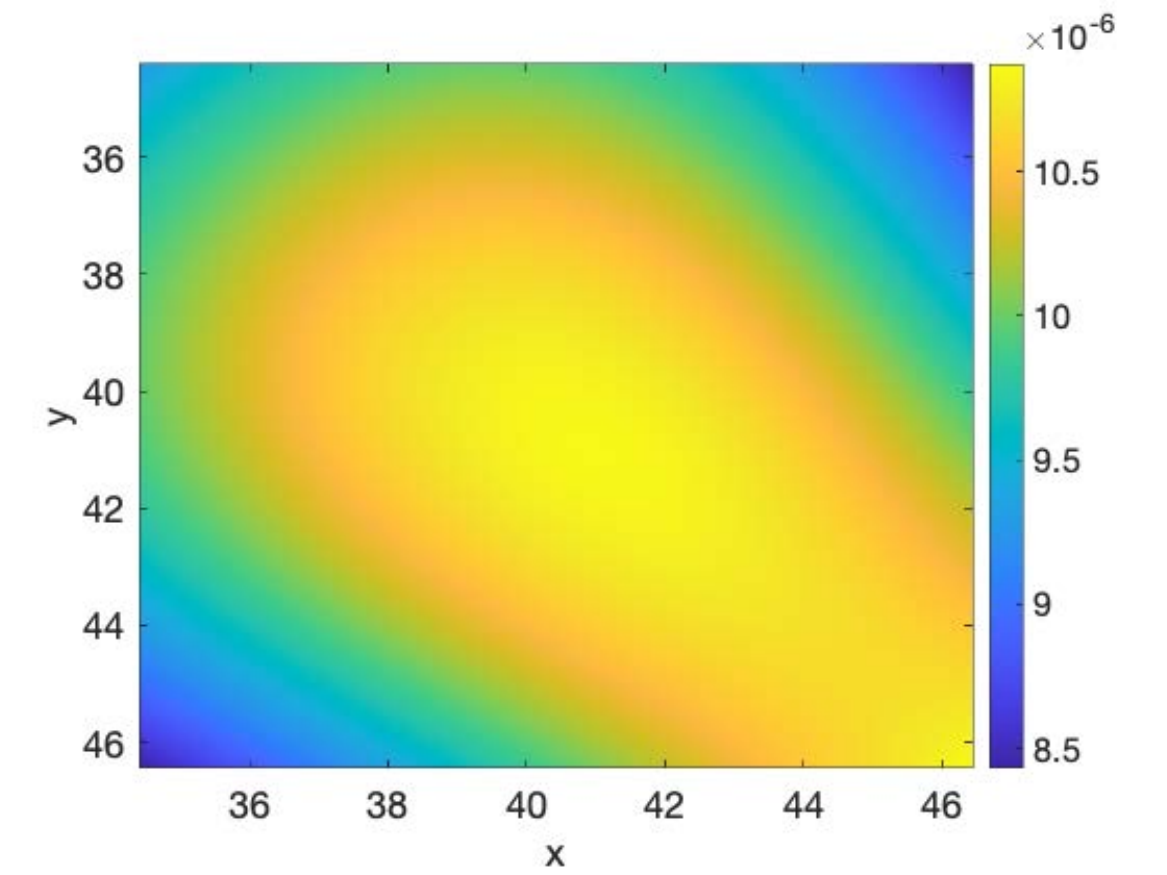}
    \caption{Absolute error on $\Omega$ associated with test case reported in Table~\ref{tab:result01}.}
    \label{fig:sub4}
\end{figure}
\end{minipage}

\vspace*{+0.25cm}
Despite knowing the optimal control, we discretize the admissible control set in our experiments for generality, i.e.\ $\Al_h$ is used and it contains the optimal control $\{(\alpha^*)\} = \{\sigma^{x}_{\mymax}, \sigma^{y}_{\mymax}, \rho_{\mymin}\}$ at all refinement levels. It is observed that the proposed MPCCI scheme accurately yielded the aforementioned optimal control for all $(x_n, y_j, \tau_m)$. Table~\ref{tab:result01} shows the convergence results for the time $t = 0$ option price at $e^x = X_0$, $e^y = Y_0$. To provide an estimate of the convergence rate of the proposed MPCCI~method, we compute the ``Abs.\ error'' as the absolute error between the exact option \cite{Stulz1982} and numerical option prices, and the ``Ratio'' as the ratio of successive absolute errors. These results  indicate excellent agreement with the analytic solution from \cite{Stulz1982}, as do the results from MC simulations.  Notably, the MPCCI method exhibits second-order convergence.

We further demonstrates the accuracy of the MPCCI method in the entire domain $\Omega$.
In Figure~\ref{fig:sub4}, we present the absolute error at time $t = 0$ on grid points obtained with refinement Level~4.
The absolute error, computed as
\[
\big| v(x_n, y_j, \tau_M) - v_{n, j}^M\big|, \quad n \in \Nbb^{\dagger}, \quad j \in \Jbb^{\dagger}, \quad\tau_M = T,
\]
 is very small across the computational domain, typically of the order of $10^{-5}$  or less, with higher errors concentrated near the strike $K = 40$, as expected.

\begin{minipage}{0.525\textwidth}
\strut\vspace*{-\baselineskip}
\flushleft
\begin{linenomath}
 \begin{table}[H]
\center
\begin{tabular}{llll}
\hline
Level     & Price        & Abs. error      & Ratio           \\
\hline
     \quad   0         &3.96880850 &4.80e-03 &            \\
    \quad    1         &3.97240621 &1.20e-03 &  4.0        \\
    \quad 2            &3.97330502 &3.00e-04  &4.0         \\
      \quad  3         &3.97352968 &7.49e-05  &4.0         \\
   \quad  4            &3.97358584 &1.87e-05  &4.0          \\
\hline
 MC: 95\%-CI                 & \multicolumn{3}{l}{[3.9657, 3.9840]}
\\     \hline
Ref.\ \cite{Stulz1982}&  3.97360457 & &
\\ \hline
\end{tabular}
\captionof{table}{Convergence study for a European call option on the maximum of two assets under the two-factor uncertain volatility model (\underline{best-case}, composite trapezoidal rule). Payoff given by \eqref{eq:BSO_P1}. The reference value is the closed-form solution from \cite{Stulz1982} with $\sigma_x = \sigma_y = 0.3$ and $\rho = 0.5$.
}
\label{tab:result01b}
\end{table}
 \end{linenomath}
\end{minipage}
\begin{minipage}{0.475\textwidth}
\vspace*{-0.75cm}
\begin{figure}[H]
\includegraphics[width=0.95\linewidth]{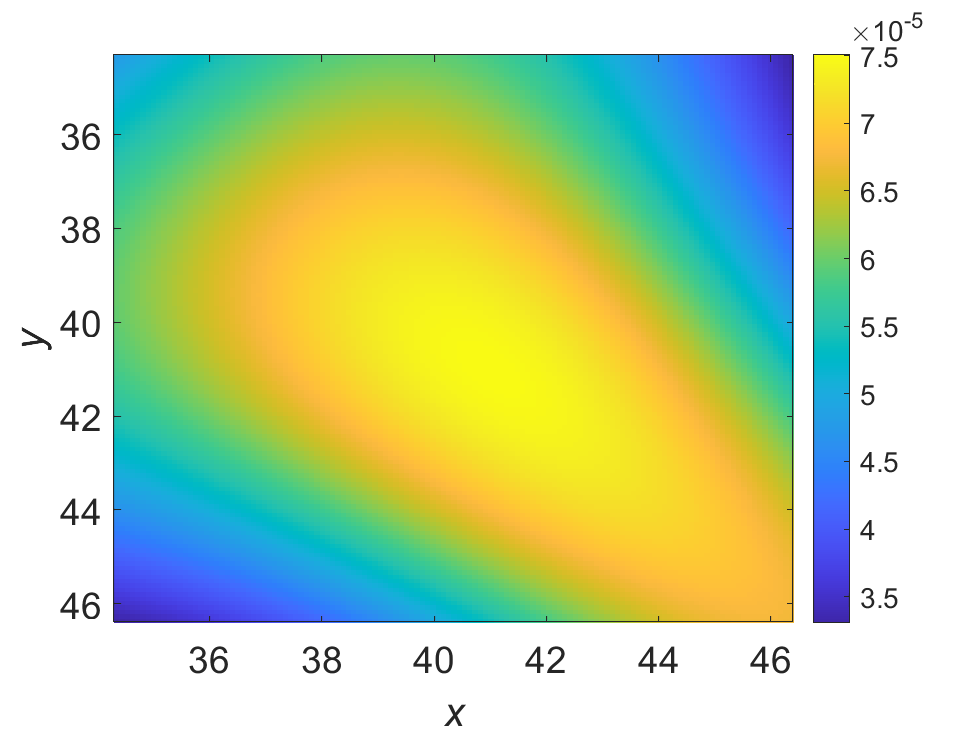}
    \caption{Absolute error on $\Omega$ associated with test case reported in Table~\ref{tab:result01b}.}
    \label{fig:sub4b}
\end{figure}
\end{minipage}

\vspace*{+0.25cm}
\noindent In Table~\ref{tab:result01b} and Figure~\ref{fig:sub4b}, we display the best-case results for the short position. The outcomes closely mirror those of the worst-case scenario, showing excellent agreement with the closed-form solution, exhibiting
 second-order of convergence.

\vspace*{+0.25cm}
\noindent \textbf{Discussion of observed convergence orders.}
While general uncertain volatility problems typically exhibit at most first-order convergence (as established in Lemma~\ref{lemma:consistency}), the results in Tables~\ref{tab:result01} and \ref{tab:result01b} demonstrate second-order convergence in scenarios with the convex payoff function \eqref{eq:BSO_P1}, namely the ``worst-case'' (resp.\ ``best-case'') European call option. In these cases, as noted earlier, theory indicates that the optimal control is constant throughout $\Omega_{\myin}$; for example, $\alpha^* = \{\sigma^x_{\max}, \sigma^y_{\max}, \rho_{\min}\}$ in the worst-case. Our numerical results show that the MPCCI method consistently identifies this global control correctly at each grid point in $\Omega_{\myin}$. Extensive tests across a wide range of parameter sets confirm this behavior. We conjecture that this may reflect an intrinsic property of the scheme when dealing with convex payoffs, though a formal proof is beyond the scope of this paper. In contrast, standard FD methods approximate the PDE's nonlinearity locally at each grid point and timestep, and often fail to converge to a single uniform control~\cite{MaForsyth2015}.

As a consequence of the fixed optimal control correctly identified at each grid point in $\Omega_{\myin}$, the PDE coefficients become constant, and therefore our scheme uses the same Green's function at each timestep to compute the 2D convolution. Due to the time additivity of the integral representation, the numerical solution at $\tau = T$ remains unchanged whether multiple timesteps or a single timestep ($M = 1$) is used. As a result, there is no timestep error, and the overall convergence rate depends solely on the quadrature used to evaluate the convolution, as well as the treatment of nonsmooth features in the payoff.

This dependency is clearly reflected in the observed convergence rates reported below.
Recall that using the composite trapezoidal rule with multiple timesteps yields $\mathcal{O}(h^2)$ convergence, as previously seen in Tables~\ref{tab:result01} and \ref{tab:result01b}. Corresponding results for $M = 1$ (a single timestep), reported in Tables~\ref{tab:result01b*} and \ref{tab:result02b} below, show virtually identical prices and the same second-order behavior.

\noindent\begin{minipage}{0.5\textwidth}
\strut\vspace*{-\baselineskip}
\flushleft
\small
\begin{linenomath}
\begin{table}[H]
\setlength{\tabcolsep}{2.5pt}
\center
{\color{black}
\begin{tabular}{clllr}
\hline
Level     & Price        & Abs.\      & Ratio & CPU    \\
$(M=1)$ &        & error      & &time           \\
\hline
   0         &6.84492757 &2.77e-03 &   &0.11s         \\
    1        &6.84700691 &6.93e-04 & 4.0 &0.36s        \\
 2          &6.84752663 &1.73e-04  &4.0 &1.52s        \\
  3        &6.84765655 &4.33e-05  &4.0 &6.44s        \\
 4          &6.84768903 &1.08e-05  &4.0  &27.44s        \\
\hline
Ref.\ \cite{Stulz1982}&  6.84769986 & &
\\     \hline
\end{tabular}
}
\captionof{table}{{\color{black}
Convergence study with a single timestep ($M = 1$) for the same test case reported in Table~\ref{tab:result01}, (worst-case, composite trapezoidal rule).}
\label{tab:result01b*}
}
\end{table}
\end{linenomath}
\end{minipage}
 \noindent\begin{minipage}{0.5\textwidth}
\strut\vspace*{-\baselineskip}
\flushleft
\small
\begin{linenomath}
\begin{table}[H]
\setlength{\tabcolsep}{2.5pt}
\center
{\color{black}
\begin{tabular}{clllr}
\hline
Level     & Price        & Abs.\      & Ratio &time          \\
$(M=1)$ &        & error      & &time           \\
   \hline
     0        &3.96880849  &4.80e-03  &         &0.10s        \\
       1        &3.97240621  &1.20e-03  &4.0     &0.35s        \\
       2        &3.97330502  &3.00e-04  &4.0     &1.51s        \\
       3        &3.97352968  &7.49e-05  &4.0     &6.20s        \\
       4        &3.97358585  &1.87e-05  &4.0     &26.35s       \\
    \hline
Ref.\ \cite{Stulz1982}& 3.97360457 & &
\\
\hline
\end{tabular}
}
\captionof{table}{{\color{black}Convergence study with a single timestep ($M = 1$) for the same test case reported in Table~\ref{tab:result01b}, (best-case, composite trapezoidal rule).}
}
\label{tab:result02b}
\end{table}
\end{linenomath}
\end{minipage}

Furthermore, replacing the trapezoidal rule with Simpson's rule leads to $\mathcal{O}(h^4)$ convergence (with a single time step), as shown in Tables~\ref{tab:result_simpson_1} and \ref{tab:result_simpson_2}, provided the payoff kinks are sufficiently resolved. In fact, the composite Simpson's rule yields significantly smaller absolute errors (on the order of $10^{-6}$ at the coarsest level down to $10^{-11}$ at the finest) than the trapezoidal rule ($10^{-3}$ to $10^{-5}$, respectively), despite using the same spatial grid resolution and number of time steps.
\noindent
\begin{minipage}{0.5\textwidth}
\strut\vspace*{-\baselineskip}
\flushleft
\small
\begin{linenomath}
\begin{table}[H]
\setlength{\tabcolsep}{3pt}
\centering
{\color{black}
\begin{tabular}{llll}
\hline
Level     & Price        & Abs.\      & Ratio \\
$(M=1)$ &        & error      &\\
\hline
\qquad 0 & 6.84770210 & 2.23e-06 & \\
\qquad 1 & 6.84770000 & 1.39e-07 & 16.04 \\
\qquad 2 & 6.84769987 & 8.70e-09 & 16.01 \\
\qquad 3 & 6.84769986 & 5.44e-10 & 16.00 \\
\qquad 4 & 6.84769986 & 3.40e-11 & 16.00 \\
\hline
Ref.\ \cite{Stulz1982} & 6.84769986 & & \\
\hline
MC: 95\%-CI & \multicolumn{3}{l}{[6.8319, 6.8618]} \\
\hline
\end{tabular}
}
\captionof{table}{{\color{black}
Convergence study with a single timestep ($M = 1$) for the same test case reported in Table~\ref{tab:result01} (worst-case, \underline{composite Simpson's rule}).}}
\label{tab:result_simpson_1}
\end{table}
\end{linenomath}
\end{minipage}
\hfill
\begin{minipage}{0.5\textwidth}
\strut\vspace*{-\baselineskip}
\flushleft
\small
\begin{linenomath}
\begin{table}[H]
\setlength{\tabcolsep}{3pt}
\centering
{\color{black}
\begin{tabular}{llll}
\hline
Level     & Price        & Abs.\      & Ratio\\
$(M=1)$ &        & error      & \\
\hline
\qquad 0 & 3.97361470 & 1.01e-05 & \\
\qquad 1 & 3.97360520 & 6.28e-07 & 16.12 \\
\qquad 2 & 3.97360461 & 3.92e-08 & 16.03 \\
\qquad 3 & 3.97360457 & 2.45e-09 & 16.01 \\
\qquad 4 & 3.97360457 & 1.53e-10 & 16.00 \\
\hline
Ref.\ \cite{Stulz1982} & 3.97360457 & & \\
\hline
MC: 95\%-CI & \multicolumn{3}{l}{[3.9657, 3.9840]} \\
\hline
\end{tabular}
}
\captionof{table}{{\color{black}Convergence study with a single timestep ($M = 1$) for the same test case reported in Table~\ref{tab:result01b} (best-case, \underline{composite Simpson's rule})}}
\label{tab:result_simpson_2}
\end{table}
\end{linenomath}
\end{minipage}

\vspace*{+0.25cm}
To achieve this, we partition the domain into subregions whose boundaries align with the payoff's non-differentiable lines (namely $x = \ln(K)$, $y = \ln(K)$, and $x = y$), so that the integrand is piecewise smooth within each subdomain. This alignment allows Simpson's rule to achieve its full fourth-order accuracy. This high-order performance stands in sharp contrast to standard FD-based methods for uncertain volatility problems, which rarely capture a global optimal control and generally remain limited to first-order accuracy, even for the same convex payoff functions.

\vspace*{+0.25cm}
\noindent \textbf{Accuracy and run-time comparison with \cite{MaForsyth2015}.}
We now compare our method against the unconditionally monotone FD approaches reported in \cite{MaForsyth2015}---to the best of our knowledge, the only published FD schemes for multi-dimensional uncertain volatility that guarantee unconditional monotonicity.
We note that the tree-grid method proposed in \cite{kossaczky20192d} is \emph{not} unconditionally monotone, and thus falls outside the scope of this comparison.

In Table~3 of \cite{MaForsyth2015}, for example, the ``hybrid scheme with rotation''—which appears to be among the more efficient schemes considered therein —on a $361\times361$ spatial grid with 100 timesteps, computes a worst-case option price of about $6.8542$ in $4{,}300.73\,\text{s}$ on a workstation with a 2.83\,GHz Intel Xeon CPU, yielding an absolute error of approximately $6.05\times10^{-3}$ relative to the exact solution $6.84769986$.
Refining the spatial grid to $721\times721$ with 200 timesteps
reduces the error to about $2.9\times10^{-3}$ (with a price of 6.8506) , but raises the run time drastically, up to $41{,}046.12\,\text{s}$.
Table~4 shows similar behavior for the ``pure wide-stencil scheme with rotation,''
with run times ranging from tens of thousands to over a hundred thousand CPU-seconds
(e.g.\ 116443.90\,\text{s} for the finest grid), yet still reporting errors on
the order of $10^{-2}$ to $10^{-3}$ for the worst-case option.

By contrast, our implementation in MATLAB on a standard desktop PC (Intel Core~i7 at 2.50\,GHz) achieves about $10^{-5}$ accuracy in a single timestep in under $30\,\text{s}$ using the composite trapezoidal rule, and about $10^{-11}$ accuracy in slightly more time with Simpson's rule (Tables~\ref{tab:result01b*}, \ref{tab:result02b}, \ref{tab:result_simpson_1}, and \ref{tab:result_simpson_2}).

We note, however, that single-step integration is only applicable in special cases,
such as convex payoffs, where the optimal control remains constant.
For general non-convex payoffs, multiple time steps are required. Even then, our method remains significantly faster.

More specifically, although the payoffs in Tables~\ref{tab:result01} and \ref{tab:result01b} are convex,
we employed multiple timesteps to demonstrate the method's runtime under a multi-step variant.
At the finest refinement levels in these tables (with $M=800$ timesteps), the total runtime
is about 21{,}000\,\text{s}---roughly half the 41{,}046.12\,\text{s} reported for the ``hybrid
scheme with rotation''---while achieving an error on the order of $10^{-5}$ (rather than only $10^{-3}$ with the hybrid FD scheme).

Although hardware, language, and code-optimization differences make a direct match approximate, the disparity in both accuracy and computational cost is striking. Our scheme bypasses the policy iteration and complex stenciling that FD methods require, offering a robust and efficient alternative for this class of problems.

Finally, we observe that the convergence behavior of the FD methods reported in
\cite{MaForsyth2015} is often erratic, with apparent orders varying inconsistently
across refinement levels (e.g.\ from 5.6 to 2.6, or from 1.3 to 1.5). In contrast,
our MPCCI method exhibits smooth and consistent convergence behavior across all
refinement levels, indicating that it offers greater stability in convergence trends—
a significant practical advantage.

\subsubsection{Butterfly options}
In the second test, we consider a butterfly option on the maximum of two assets.
For this option, the payoff function $p(e^x,e^y)$ is given by
\begin{align}
\label{eq:BSO_P2}
p(e^x,e^y) = \max\big(\max(e^x,e^y)-K_1,0\big)&-2\max\big(\max(e^x,e^y)-(K_1+K_2)/2,0\big)
\nonumber
\\
                    &+\max\big(\max(e^x,e^y)-K_2,0\big), \quad K_1, K_2>0.
\end{align}
For the butterfly payoff function \eqref{eq:BSO_P2}, a closed-form expression for the option price is unknown.
To estimate the convergence rate of the proposed MPCCI method, we calculate the ``Change'' as the difference in values from coarser to finer grids and the ``Ratio'' as the ratio of changes between successive grids.
We compare our prices against reference prices obtained by a FD method with pure wide stencil  rotation  developed in \cite{MaForsyth2015}, and by a tree-grid (TG) method of \cite{kossaczky20192d}.

Tables~\ref{tab:result02} and \ref{tab:result03} display the numerical prices for both the worst-case and best-case scenarios (short position), indicating that our method approximates first-order convergence. The comparison with reference prices shows minimal differences: against FD method prices, the discrepancies are around $6\times 10^{-3}$ for the worst-case and $1\times 10^{-3}$ for the best-case scenarios. When compared with TG prices, the differences are $2\times 10^{-3}$ and $4\times 10^{-3}$, respectively, highlighting the MPCCI method's precision.

\begin{figure}[htbp]
\centering
\begin{minipage}[t]{0.5\textwidth}
\centering
\begin{tabular}{lllll}
\hline
Level & Price & Change & Ratio \\ \hline
0 & 2.65092717 & & \\
1 & 2.66374754 & 0.0128 & \\
2 & 2.67280480 & 0.0091 & 1.42 \\
3 & 2.67793762 & 0.0051 & 1.76 \\
4 & 2.68070303 & 0.0028 & 1.86 \\ \hline
MC: 95\%-CI & \multicolumn{3}{l}{[2.6735, 2.6832]} \\
\hline
FD\ \cite{MaForsyth2015}&  2.6744 & &
\\
\hline
TG\ \cite{kossaczky20192d}&  2.6784 & &
\\
\hline
\end{tabular}
\captionof{table}{
Convergence study for a butterfly option (\underline{worst-case}, composite trapezoidal rule) under a two-factor uncertain volatility model  - payoff function in \eqref{eq:BSO_P2}. Reference prices: by FD method is 2.6744 \cite{MaForsyth2015} (finest level in Table~6 therein, Pure wide stencil (with rotation)), by TG method is 2.6784 \cite{kossaczky20192d} (finest level in Table~3 therein)}
\label{tab:result02}
\end{minipage}\hfill
\begin{minipage}[t]{0.5\textwidth}
\centering
\begin{tabular}{lllll}
\hline
Level & Price & Change & Ratio \\ \hline
0 & 0.94015237 & & \\
1 & 0.92418409 & -0.0138 & \\
2 & 0.91794734 & -0.0062 & 2.56 \\
3 & 0.91473085 & -0.0032 & 1.94 \\
4 & 0.91308945 & -0.0016 & 1.96 \\  \hline
MC : 95\%-CI & \multicolumn{3}{l}{[0.9120, 0.9182]}
\\
\hline
FD\ \cite{MaForsyth2015}&  0.9148 & &
\\
\hline
TG\ \cite{kossaczky20192d}&  0.9173  & &
\\
\hline
\end{tabular}
\captionof{table}{
Convergence study for a butterfly option (\underline{best-case}, composite trapezoidal rule) under a two-factor uncertain volatility model  - payoff function in \eqref{eq:BSO_P2}. Reference prices: (i) by FD method is 0.9148 \cite{MaForsyth2015} (finest level in Table~8 therein, Pure wide stencil (with rotation)), (ii) by TG method is 0.9173 \cite{kossaczky20192d} (finest level in Table~4 therein).}
\label{tab:result03}
\end{minipage}
\end{figure}

\subsection{Impact of spatial domain sizes}
\label{ssc:padding}
{\dangblue{In this subsection, we numerically validate the adequacy of our selected spatial domain for the experiments. We revisit the scenarios from Tables~\ref{tab:result01}, \ref{tab:result02}, and \ref{tab:result03},
this time doubling the lengths of the spatial domains. Specifically, we extend the spatial domain boundaries to
$x_{\min} = \ln(X_0)-2.4$, $x_{\max} = \ln(X_0)+2.4$, $y_{\min} = \ln(Y_0)-2.4$, $y_{\max} = \ln(Y_0)+2.4$,
with the number of intervals $N$ and $J$ also doubled to maintain the same $\Delta x$ and $\Delta y$.

The numerical prices from this extended domain, shown in Table~\ref{tab:result04}, are virtually identical with those obtained from the original smaller domain (reproduced under columns marked Tab.~\ref{tab:result01}, Tab.~\ref{tab:result02}, and Tab.~\ref{tab:result03}). This  indicates that enlarging the spatial computational domain further has a negligible effect on the numerical prices. Additionally, for a comprehensive analysis,
we conducted tests on smaller spatial domains with the boundaries set to $x_{\min} = \ln(X_0)-0.9$, $x_{\max} = \ln(X_0)+0.9$, $y_{\min} = \ln(Y_0)-0.9$, and $y_{\max} = \ln(Y_0)+0.9$, using the same $\Delta x$ and $\Delta y$ as in previous tests. The prices, presented in Table~\ref{tab:result08}, show slight discrepancies (from the fourth decimal digits) when compared to those
obtained from original domain size.

These findings affirm the adequacy of our computational domain, whose size was carefully chosen based on the upper bound for the boundary truncation error of the Green's function provided in \eqref{eq:boundgd}. This approach balances the need for demonstrating theoretical convergence and computational efficiency
in our analysis.
}}
\label{sec:lar_z}
\begin{table}[htb!]
\center
\begin{tabular}{lllllll}
\hline
\multirow{3}{*}{Level} & \multicolumn{6}{c}{Two-factor uncertain volatility model}          \\ \cline{2-7}
 & \multicolumn{2}{c}{European} & \multicolumn{2}{c}{Butterfly (worst)}    & \multicolumn{2}{c}{Butterfly (best)}                                   \\ \cline{2-7}
                & Price        & Price    & Price        & Price  & Price        & Price
                       \\
                        &         & (Tab.~\ref{tab:result01})    &        & (Tab.~\ref{tab:result02})  &        & (Tab.~\ref{tab:result03})
                            \\ \hline
0                      & 6.84492758    &  6.84492756        & 2.65092717    &   2.65092717      & 0.94015237    &   0.94015237                                                            \\
1                      & 6.84700691     &  6.84700690       & 2.66374754    &  2.66374754   & 0.92418409    &  0.92418409                   \\
2                      & 6.84752663     &  6.84752662    & 2.67280480     & 2.67280480  & 0.91794734     & 0.91794734    \\
3                      & 6.84765655     &  6.84765654     & 2.67793762    & 2.67793762   & 0.91473085   & 0.91473085    \\
4                      & 6.84768903     &  6.84768902    & 2.68070303      & 2.68070303   & 0.91308945    & 0.91308945      \\ \hline
\end{tabular}
\caption{Prices obtained using a \underline{larger} spatial computational domain:
$x_{\min} = \ln(X_0)-2.4$, $x_{\max} = \ln(X_0)+2.4$, $y_{\min} = \ln(Y_0)-2.4$, $y_{\max} = \ln(Y_0)+2.4$,
in comparison with prices in Table~\ref{tab:result01}, \ref{tab:result02}, \ref{tab:result03}
obtained with the original smaller domain $z_{\min} = \ln(Z_0)-1.2$, $z_{\max} = \ln(Z_0)+1.2$, for $z\in \{x,y\}$.
}
\label{tab:result04}
\end{table}

\label{sec:small_z}
\begin{table}[htb!]
\center
\begin{tabular}{lllllll}
\hline
\multirow{3}{*}{Level} & \multicolumn{6}{c}{Two-factor uncertain volatility model}          \\ \cline{2-7}
 & \multicolumn{2}{c}{European} & \multicolumn{2}{c}{Butterfly (worst)}    & \multicolumn{2}{c}{Butterfly (best)}                                   \\ \cline{2-7}
                & Price        & Price    & Price        & Price  & Price        & Price
                       \\
                        &         & (Tab.~\ref{tab:result01})    &        & (Tab.~\ref{tab:result02})  &        & (Tab.~\ref{tab:result03})
                            \\ \hline
0                      & 6.84490660    &  6.84492756        & 2.65092891    &   2.65092717      & 0.94014513    &   0.94015237                                                            \\
1                      & 6.84698614     &  6.84700690       & 2.66375021    &  2.66374754   & 0.92417329    &  0.92418409                   \\
2                      & 6.84750559     &  6.84752662    & 2.67280820     & 2.67280480  & 0.91793374     & 0.91794734    \\
3                      & 6.84763508     &  6.84765654     & 2.67794155    & 2.67793762   & 0.91471524   & 0.91473085    \\
4                      & 6.84766711     &  6.84768902    & 2.68070724      & 2.68070303   & 0.91307245    & 0.91308945      \\ \hline
\end{tabular}
\caption{
Prices obtained using a smaller computational domain: $x_{\min} = \ln(X_0)-0.9$, $x_{\max} = \ln(X_0)+0.9$.
$y_{\min} = \ln(Y_0)-0.9$, and $y_{\max} = \ln(Y_0)+0.9$.  Compare with prices in Table~\ref{tab:result01}, \ref{tab:result02}, \ref{tab:result03}, where $z_{\min} = \ln(Z_0)-1.2$, $z_{\max} = \ln(Z_0)+1.2$, for $z\in \{x,y\}$.
}
\label{tab:result08}
\end{table}

\subsection{Impact of boundary conditions}
\label{sec:constant_pad}
In this subsection, we numerically demonstrate that our straightforward approach of employing discounted payoffs for boundary sub-domains is adequate. We revisited previous experiments reported in Tables~\ref{tab:result01}, \ref{tab:result02}, and \ref{tab:result03}, introducing sophisticated boundary conditions based on the asymptotic behavior of the HJB equation \eqref{eq:omega_inf_all} as $z \to -\infty$ and $z \to \infty$ for $z\in\{x,y\}$ as proposed in \cite{MaForsyth2015}.
Specifically,  the HJB equation \eqref{eq:omega_inf_all} simplifies to the 1D forms shown in \eqref{eq:hjbxymin} when $x$ or $y$ tends to $-\infty$:
\EQ
\label{eq:hjbxymin}
\begin{aligned}
v_{\tau} - \sup_{\sigma_y \in \Al_y }\big\{ \big(r-(\sigma_{y})^2/2\big)v_{y}  +   (\sigma_{y})^2/2 v_{yy} \big\} + rv &= 0, \quad x \to -\infty,
\\
v_{\tau} - \sup_{\sigma_x \in \Al_x}\big\{ \big(r-(\sigma_{x})^2/2\big)v_{x} +   (\sigma_{x})^2/2v_{xx} \big\} +rv &= 0, \quad y \to -\infty.
\end{aligned}
\EN
As  $x, y \to -\infty$, the HJB equation \eqref{eq:omega_inf_all} simplifies to the ordinary differential equation $v_{\tau} + rv = 0$.

To adhere to these asymptotic boundary conditions, we choose a much large spatial domain: $x_{\min} = \ln(X_0)-9.6$, $x_{\max} = \ln(X_0)+9.6$, $y_{\min} = \ln(Y_0)-9.6$, $y_{\max} = \ln(Y_0)+9.6$,  and adjust the number of intervals $N$ and $J$ accordingly to maintain the same grid resolution ($\Delta x$ and $\Delta y$).
Employing the monotone piecewise constant control integration technique, tailored for the 1D case, we solved the 1D HJB equations in \eqref{eq:hjbxymin}. The ordinary differential equation $v_{\tau} + rv = 0$ is solved directly and efficiently. The scheme's convergence to the viscosity solution can be rigorously established in the same fashion as the propose scheme.

The resulting option prices, listed in Table~\ref{tab:result10}, are virtually identical with those from the original settings (under columns marked with Tab.~\ref{tab:result01}, Tab.~\ref{tab:result02}, and Tab.~\ref{tab:result03}).
These results confirm the effectiveness of our simple boundary conditions, demonstrating that they are both easy to implement and sufficient for the theoretical and practical demands of our numerical experiments.

\begin{table}[htb!]
\center
\begin{tabular}{lllllll}
\hline
\multirow{3}{*}{Level} & \multicolumn{6}{c}{Two-factor uncertain volatility model}          \\ \cline{2-7}
 & \multicolumn{2}{c}{European} & \multicolumn{2}{c}{Butterfly (worst)}    & \multicolumn{2}{c}{Butterfly (best)}                                   \\ \cline{2-7}
                & Price        & Price    & Price        & Price  & Price        & Price
                       \\
                        &         & (Tab.~\ref{tab:result01})    &        & (Tab.~\ref{tab:result02})  &        & (Tab.~\ref{tab:result03})
                            \\ \hline
0                      & 6.84492760    &  6.84492756        & 2.65092717    &   2.65092717      & 0.94015237    &   0.94015237                                                            \\
1                      & 6.84700690     &  6.84700690       & 2.66374754    &  2.66374754   & 0.92418410    &  0.92418409                   \\
2                      & 6.84752663     &  6.84752662    & 2.67280480      & 2.67280480  & 0.91794734     & 0.91794734    \\
3                      & 6.84765655     &  6.84765654     & 2.67793762   & 2.67793762   & 0.91473085   & 0.91473085    \\
4                      & 6.84768903     &  6.84768902    & 2.68070303     & 2.68070303   & 0.91308945    & 0.91308945      \\ \hline
\end{tabular}
\caption{Results using sophisticated boundary conditions~\eqref{eq:hjbxymin}. Compare with results in Table~\ref{tab:result01}, \ref{tab:result02}, \ref{tab:result03} where simple boundary conditions based on discounted payoffs are used.}
\label{tab:result10}
\end{table}

\section{Conclusion}
\label{sc:conclude}

In this paper, we have presented a novel and streamlined approach for solving 2D HJB PDEs arising from two-factor uncertain volatility models with uncertain correlation.


Departing from the traditional ``discretize, then optimize'' strategy, our
``decompose and integrate, then optimize'' method leverages a piecewise constant
control technique, which, over each timestep, yields a set of independent 2D linear
PDEs---each corresponding to a discretized control value---that are solved using
Green's function convolution. The resulting solutions are then combined to obtain
the value function and optimal control, effectively addressing the nonlinearity of
the HJB equation and significantly simplifying the optimization process.

Our main contributions include the development of a monotone piecewise constant
control numerical integration scheme that uses closed-form Green's functions to
evaluate these convolution integrals. This avoids discretizing spatial derivatives
and, in particular, simplifies the treatment of cross-derivative terms—an advantage
over conventional finite difference methods. We have also implemented our scheme
efficiently using FFT and circulant convolution, exploiting the Toeplitz structure
of the convolution kernels to accelerate both inner and double summations via 2D FFTs.

We have mathematically demonstrated the {\imapurple{$\ell_\infty$-stability}} and
consistency of the scheme in the viscosity sense, along with its pointwise convergence
to the viscosity solution of the HJB equation. Extensive numerical experiments show excellent agreement with benchmark solutions, while also demonstrating significantly improved accuracy and run time compared to unconditionally monotone finite difference methods, thereby highlighting the robustness and efficiency of our approach.

Although our focus has been on uncertain volatility models, the overall framework—
built on piecewise constant control, monotone integration, and Green's function methods—
is general and may be adapted to a broader class of HJB equations in finance, offering
a promising direction for future research.

\section*{Acknowledgments}
The authors are grateful to the two anonymous referees and the Associate Editor for their constructive comments and suggestions, which have significantly improved the quality of this work.


\section*{Appendices}
\appendix
\section{Special case $\boldsymbol{\rho = \pm 1}$}
\label{app:rho_ex}
\subsection{Approximation of $\boldsymbol{\delta(\cdot)}$}
\label{app:app_del}
In this appendix, we detail key elements of the proposed scheme for the case $\rho = \pm 1$ as highlighted in Remark~\ref{rm:rho}, and provide selected numerical results. The key challenge is that, for computational purposes, the Dirac delta function $\delta(y - (a + \rho bx))$, where $a = \mu_y - \rho b \mu_x$ with $b = \frac{\sigma_y}{\sigma_x}$, needs to be approximated. We focus on the case $\rho = 1$. The analysis for $\rho = -1$ follows similarly and is omitted. Using a conditional density approach, $\delta(y - (a + \rho bx))$ is approximated by a Gaussian (a conditional density) when the correlation coefficient is $\hat{\rho}$ with $\hat{\rho} \nearrow 1$ \cite{glazunov2012note}:
\EQ
\label{eq:vareptwo}
\delta(y - (a + \rho bx)) = \lim_{\hat{\rho} \nearrow 1} \delta_{\hat{\rho}} (y - (a + \rho bx)), \text{ where } \delta_{\hat{\rho}}(\gamma) = \frac{\exp\left(-\frac{\gamma^2}{2 \kappa_y^2 (1 - \rhoh^2)}\right)}{ \sqrt{ 2 \pi} \kappa_y\sqrt{1 -\rhoh^2}} , \text{ and } \kappa_y = \sigma_y \sqrt{\Delta \tau}.
\EN
For a fixed $\alpha$, recall the exact Green's function $g_\alpha(x, y; \Delta \tau)$ defined in \ref{eq:g_k_rho}. We define by $\hat{g}_\alpha(x, y; \Delta \tau)$ an approximation to $g_{\alpha}(x, y; \Delta \tau)$ obtained by replacing $\delta(y - (a + \rho bx))$ by $\delta_{\hat{\rho}}(y - (a + \rho bx))$. Formally,
\EQ
\label{gh}
\hat{g}_{\al}(x, y; \Delta \tau) = e^{-r\Delta \tau} \frac{1}{\sqrt{2\pi} \kappa_x} \exp\left( -\frac{(x - \mu_x)^2}{2\kappa_x^2} \right) {{\delta_{\rhoh} (y - (a +\rho bx))}},
\EN
where $\delta_{\rhoh}(\cdot)$ is defined in \eqref{eq:vareptwo}, and { {$a = \mu_y - \rho b \mu_x$}} and $b = \frac{\sigma_y}{\sigma_x}$. The function $\hat{g}_\alpha(x, y; \Delta \tau)$ is the weight function for our scheme. {\imablue{Using the same techniques as those employed in Lemma \ref{lemma:truncation},
we can show that, for each fixed $\hat{\rho}$, the bound \eqref{eq:boundgn} also holds for $\hat{g}_{\al}(x, y; \Delta \tau)$ defined in \eqref{gh}, i.e.\
\EQS
\iint_{\Rbb^2\setminus \mathbf{D}^{\dagger}}\hat{g}_{\alpha}\left(x,y;\Delta \tau\right) dxdy
<
C \sqrt{\Delta \tau} e^{-\frac{1}{2\Delta \tau}},
\quad
\mathbf{D}^{\dagger}\equiv [x_{\min}^{\dagger}, x_{\max}^{\dagger}] \times [y_{\min}^{\dagger}, y_{\max}^{\dagger}],
\ENS
where $C$ is a bounded constant independently of $\Delta \tau$ and $\rhoh$.}}
In this case, our scheme is monotone, and it is straightforward to show that it is $\ell_{\infty}$-stable.
The selection of $\rhoh$ is crucial for the scheme's consistency. Below, we show that choosing $\rhoh$ appropriately can achieve first-order consistency for the scheme.

\subsection{Consistency}
For the rest of the proof,  we let {\zblue{$C$}} be generic bounded constant independent of the discretization parameter $h$,
which may take different values from line to line.
We re-examine the proof of Lemma~\ref{lemma:error_smooth}, now utilising $\hat{g}_{\al}(x, y; \Delta \tau)$ from \eqref{gh} instead of $g_{\al}(x, y; \Delta \tau)$.
{\zblue{For a smooth test function $\phi$, and recalling the smooth function $\chi \in L^1(\Rbb^2)$
with bounded derivatives up to second-order in $\Rbb^2$, a mollified version of $\phi$, we have}}
\begin{align}
\label{eq:doublesum}
&\Delta x \Delta y
\mysum_{l\in\mathbb{N}^{\dagger}}^{d\in\mathbb{N}^{\dagger}} \varphi_{l, d}~
    \hat{g}_{n-l, j - d}^{\alpha}~  \phi_{l,d}^{m}~
= \iint_{\Rbb^{2}} \hat{g}_{\alpha}\l(x_n - x, y_j - y;\Delta \tau\r)~ \chi(x, y) ~dx~dy
 + \mathcal{O}(h^2) +  \mathcal{O}\big(he^{-1/h}\big)
 \nonumber
 \\
 & = \iint_{\Rbb^{2}} g_{\alpha}\l(x_n - x, y_j - y;\Delta \tau\r)~ \chi(x, y) ~dx~dy
 \nonumber
 \\
 &+ \iint_{\Rbb^{2}} (\hat{g}_{\alpha}(x_n - x, y_j - y;\Delta \tau) - g_{\alpha}(x_n - x, y_j - y;\Delta \tau))~ \chi(x, y) ~dx~dy
 + \mathcal{O}(h^2).
 \end{align}
We now focus on the error term (the second term) in \eqref{eq:doublesum}, expressed through substitutions as
 \begin{align}
 \label{eq:mid}
  \int_{\Rbb} \frac{e^{-r \Delta \tau}}{\sqrt{2\pi} \kappa_x} \exp\left( -\frac{(x - \mu_x)^2}{2\kappa_x^2} \right)
\bigg(\int_{\Rbb} (\delta_{\rhoh}{\zblue{(y - (a + \rho bx))}} - \delta(y - (a + \rho bx))) ~ \chi(x_n-x, y_j-y) dy\bigg) dx.
  \end{align}
Regarding the inner integral, we have $\int_{\Rbb} (\delta_{\rhoh}{\zblue{(y - (a + \rho bx))}} - \delta(y - (a + \rho bx))) ~ \chi(x_n-x, y_j-y) dy = \ldots$
\begin{align}
\label{eq:first}
\ldots=
\int_{\Rbb} \delta_{\rhoh}{\zblue{(y - (a + \rho bx))}} ~\chi(x_n-x, y_j-y)~dy - \chi(x_n - x, y_j - (a + \rho bx)).
\end{align}
Here, the second term in \eqref{eq:first} is due to from the sifting property of the Delta function.
Letting $\gamma = y_j -{\zblue{ (a + \rho bx)}}$ and applying a change of variables, the integral in \eqref{eq:first} is reformulated as
\begin{align}
\label{eq:reform}
\int_{\Rbb} \delta_{\rhoh}\big(y - \gamma\big) \chi(x_n -x, y) dy.
\end{align}
By Taylor's series expansion, we have
\begin{align*}
\chi(\cdot, y) = \chi(\cdot, \gamma + (y - \gamma))= \chi(\cdot, \gamma)
 + (y - \gamma) \frac{\partial \chi}{\partial y}(\cdot, \gamma) +
\frac{(y - \gamma)^2}{2} \frac{\partial^2 \chi}{\partial y^2}(\cdot, \gamma) + o((y - \gamma)^2).
\end{align*}
So $\ds \int_{\Rbb} \delta_{\rhoh}\big(y - \gamma\big) \chi(x_n-x, y) dy = \ldots$
\begin{align}
\label{eq:tay}
\ldots
= \int_{\Rbb} \delta_{\rhoh}\big(y - \gamma\big)
\left( \chi(x_n-x, \gamma) + (y - \gamma) \frac{\partial \chi}{\partial y}(x_n-x, \gamma)
+
\frac{(y - \gamma)^2}{2} \frac{\partial^2 \chi}{\partial y^2}(x_n-x, \gamma)
+ o((y - \gamma)^2)\right) dy.
\end{align}
Terms in \eqref{eq:tay} are further simplified as follows
\EQ
\label{eq:terms}
\begin{aligned}
\int_{\Rbb} \delta_{\rhoh}\big(y - \gamma\big)~\chi(\cdot, \gamma)~dy & = \chi(\cdot, \gamma) \int_{\Rbb} \delta_{\rhoh}\big(y - \gamma\big)~dy  = \chi(\cdot, \gamma),
\\
\int_{\Rbb} \delta_{\rhoh}\big(y - \gamma\big) (y - \gamma) \frac{\partial \chi}{\partial y}(\cdot, \gamma)~dy
&= \frac{\partial \chi}{\partial y}(\cdot, \gamma) \int_{\Rbb} \delta_{\rhoh}\big(y - \gamma\big) (y - \gamma) ~dy = 0,
\\
\int_{\Rbb} \delta_{\rhoh}\big(y - \gamma\big)
\frac{(y - \gamma)^2}{2} \frac{\partial^2 \chi}{\partial y^2}(\cdot, \gamma)~dy
& =
\frac{\partial^2 \chi}{\partial y^2}(x, \gamma) \int_{\Rbb} \delta_{\rhoh}\big(y - \gamma\big)
\frac{(y - \gamma)^2}{2} ~dy = \frac{\kappa_y^2 (1 - \rhoh^2)}{4}~ \frac{\partial^2 \chi}{\partial y^2}(\cdot, \gamma).
\end{aligned}
\EN
{\dangblue{Next, we substitute \eqref{eq:terms} into \eqref{eq:reform} which is the first term in \eqref{eq:first}, {\zblue{noting
 that the term  $\chi(\cdot, \gamma)$ in \eqref{eq:terms} is indeed $ \chi(x_n - x , y_j - (a + \rho bx))$
 and it  cancels with the term $\chi(x_n - x , y_j - (a + \rho bx))$ in \eqref{eq:first}.
Therefore, due to boundedness of the derivatives of $\chi(\cdot)$, the error term \eqref{eq:mid} becomes
$C\kappa_y^2 (1 - \rhoh^2)$.
Therefore, noting $\kappa_y = \sigma_y \sqrt{\Delta \tau}$, \eqref{eq:doublesum}  becomes
\begin{align}
\label{eq:doubleTwo}
\Delta x \Delta y
\mysum_{l\in\mathbb{N}^{\dagger}}^{d\in\mathbb{N}^{\dagger}} \varphi_{l, d}~
    \hat{g}_{n-l, j - d}^{\alpha}~  \phi_{l,d}^{m}~
=
\iint_{\Rbb^{2}} g_{\alpha}\l(x_n - x, y_j - y;\Delta \tau\r)~ \chi(x, y) ~dx~dy
+ C\Delta \tau (1 - \rhoh^2)+\Ocal(h^2).
\end{align}
Now we re-examine Lemma~\ref{lemma:consistency} with \eqref{eq:doubleTwo} in mind. Here,
since we need to achieve $C\Delta \tau (1 - \rhoh^2)/h \to 0$ as $h \to 0$,
$\rhoh$ needs to be such that $(1- \rhoh^2) \to 0$ as $h\to 0$.
A possible choice  is $\rhoh = \sqrt{1 - Ch}$, which gives
$(1 - \rhoh^2)  = \Ocal(h)$, and we obtain the same overall
$\Ocal(h)$ error as in Lemma~\ref{lemma:consistency} for scenarios $|\rho|<1$.
}}
}}

\subsection{Select numerical experiments}
\begin{figure}[H]
\centering
\begin{minipage}[t]{0.5\textwidth}
\centering
\begin{tabular}{lllll}
\hline
Level & Price & Error & Ratio \\ \hline
0 & 8.41173784 & 3.67e-03 & \\
1 & 8.41450094 & 9.07e-04 & 4.00 \\
2 & 8.41519144 & 2.16e-04 & 4.19 \\
3 & 8.41536400 & 5.15e-05 & 4.20\\
4 & 8.41540714 & 1.17e-05&  4.40\\
\hline
 Ref.\ \cite{Stulz1982}  &8.41540757 & &
\\       \hline
 MC: 95\%-CI     & \multicolumn{3}{l}{[8.3991, 8.4314]}   \\ \hline
\end{tabular}
\captionof{table}{
Convergence study for a European call option on the maximum of two risky assets under two-factor uncertain volatility model (worst case) with $\rho\in [-1, 1]$ - payoff function in \eqref{eq:BSO_P1}. The closed form solution is obtained using \cite{Stulz1982} with fixed parameters $\sigma_x^* = 0.5$, $\sigma_y^* = 0.5$ and $\rho^* = -1$.}
\label{tab:result07}
\end{minipage}\hfill
\begin{minipage}[t]{0.5\textwidth}
\centering
\begin{tabular}{lllll}
\hline
Level & Price & Error & Ratio \\ \hline
0 & 4.20586189 & 1.84e-03& \\
1 & 4.20724358 & 4.60e-04 & 4.00\\
2 & 4.20758890 & 1.15e-04 & 4.00 \\
3 & 4.20767522 & 2.86e-05 & 4.02 \\
4 & 4.20769680 & 7.02e-06 & 4.07 \\
\hline
 Ref.\ \cite{Stulz1982}  &4.20770382 & &
\\       \hline
 MC: 95\%-CI      & \multicolumn{3}{l}{[4.2026, 4.2208]}   \\ \hline
\end{tabular}
\captionof{table}{
Convergence study for a European call option on the maximum of two risky assets under two-factor uncertain volatility model (best case) with $\rho\in [-1, 1]$ - payoff function in \eqref{eq:BSO_P1}. The closed form solution is obtained using \cite{Stulz1982}
with fixed parameters $\sigma_x^* = 0.5$, $\sigma_y^* = 0.5$ and $\rho^* = 1$.
}
\label{tab:result08*}
\end{minipage}
\end{figure}
{\zblue{In the above, we show that a possible choice for $\rhoh$ is  $\sqrt{1 - Ch}$, where $C$ is a bounded constant
independently of $h$.  We now present a heuristic method to determine $C$.
In the numerical experiments for these special cases, we choose $\Delta y = 6 \kappa_y\sqrt{1-\rhoh^2}$, resulting in
$\rhoh = \sqrt{1-\frac{\Delta y^2 }{(6\kappa_y)^2}}$. To avoid the necessity for interpolation, for a given $\al\in\Al_h$, we adjust the partition for $y$-direction so that each of pair $(x_n,~ a+\rho b x_n)$, where $n \in \mathbb{N}$, aligns with the grid points.
In Tables~\ref{tab:result07} and \ref{tab:result08*}, we present the worst-case and best-case prices for the short position in the case of European rainbow options with payoff function in \eqref{eq:BSO_P1} with $\rho\in [-1, 1]$. Other parameters given in Table~\ref{tab:parameter02}, and the mesh size and timestep refinement levels are in Table~\ref{tab:step01}.
The closed-form solution is obtained using \cite{Stulz1982} with fixed parameters $\sigma_x = 0.5$, $\sigma_y = 0.5$ and
$\rho = \{-1, 1\}$.  It is evident that the numerical prices show excellent agreement with the closed-form solutions,
and also exhibit approximately second-order of convergence, which aligns with our explanations
in  Section~\ref{sssec:Eu}.
}}

\section{Details of padding matrices}
\label{app:mats}
In this appendix, we provide details of the padding matrices $\tilde{\gb}^{q,\al}_{-1,0}$, $\tilde{\gb}^{q,\al}_{1,0}$, $\tilde{\gb}^{q,\al}_{-1,1}$, $\tilde{\gb}^{q,\al}_{0,1}$ and $\tilde{\gb}^{q,\al}_{1,1}$ for the circulant matrix $\tilde{\gb}^{\al}_{q}$ defined in \eqref{eq:cir_matrix_1D}. These padding matrices are defined as follows
\EQAS
\tilde{\gb}^{q,\al}_{-1,0} &=&
\l[\begin{array}{cccccccc}
    g^{\al}_{-N/2+1,q} & g^{\al}_{-N/2,q} & \dots & g^{\al}_{-3N/2+1,q} & g^{\al}_{3N/2-1,q} & g^{\al}_{3N/2-2,q} & \dots & g^{\al}_{N/2,q} \\
    g^{\al}_{-N/2+2,q} & g^{\al}_{-N/2+1,q} & \dots & g^{\al}_{-3N/2+2,q} & g^{\al}_{-3N/2+1,q} & g^{\al}_{3N/2-1,q} & \dots & g^{\al}_{N/2+1,q}\\
    \vdots & \vdots & & \vdots & \vdots & \vdots & & \vdots \\
    g^{\al}_{N/2,q} & g^{\al}_{N/2-1,q} & \dots & g^{\al}_{-N/2,q} & g^{\al}_{-N/2-1,q} & g^{\al}_{-N/2-2,q} & \dots & g^{\al}_{3N/2-1,q}
\end{array}\r]_{N \times (2N+1)},
\ENAS
\EQAS
\tilde{\gb}^{q,\al}_{1,0} &=&
\l[\begin{array}{cccccccc}
    g^{\al}_{-3N/2+1,q} & g^{\al}_{3N/2-1,q} & g^{\al}_{3N/2-2,q} & \dots & g^{\al}_{N/2+1,q} & g^{\al}_{N/2,q} & \dots & g^{\al}_{-N/2,q} \\
    g^{\al}_{-3N/2+2,q} & g^{\al}_{-3N/2+1,q} & g^{\al}_{3N/2-1,q}& \dots & g^{\al}_{N/2+2,q} & g^{\al}_{N/2+1,q} & \dots & g^{\al}_{-N/2+1,q}\\
    \vdots & \vdots & \vdots & & \vdots & \vdots & & \vdots \\
    g^{\al}_{-N/2,q} & g^{\al}_{-N/2-1,q} & g^{\al}_{-N/2-2,q} & \dots & g^{\al}_{-3N/2+1,q} & g^{\al}_{3N/2-1,q}  & \dots & g^{\al}_{N/2-1,q}
\end{array}\r]_{N \times (2N+1)},
\ENAS
\EQAS
\tilde{\gb}^{q,\al}_{-1,1} &=&
\l[\begin{array}{cccc}
    g^{\al}_{N/2-1,q} & g^{\al}_{N/2-2,q} & \dots & g^{\al}_{-N/2+2,q} \\
    g^{\al}_{N/2,q} & g^{\al}_{N/2-1,q} & \dots & g^{\al}_{-N/2+3,q} \\
    \vdots & \vdots & & \vdots \\
    g^{\al}_{3N/2-2,q} & g^{\al}_{3N/2-3,q} & \dots & g^{\al}_{N/2+1,q}
\end{array}\r]_{N \times (N-2)},
\nonumber
\ENAS
\EQAS
\tilde{\gb}^{q,\al}_{0,1} &=&
\l[\begin{array}{ccccc}
    g^{\al}_{3N/2-1,q} & g^{\al}_{3N/2-2,q}& \dots & g^{\al}_{N/2+3,q} & g^{\al}_{N/2+2,q}\\
    g^{\al}_{-3N/2+1,q} & g^{\al}_{3N/2-1,q} & \dots & g^{\al}_{N/2+4,q} & g^{\al}_{N/2+3,q} \\
    \vdots & \vdots &  & \vdots & \vdots  \\
	g^{\al}_{-N/2-2,q} & g^{\al}_{-N/2-3,q} & \dots & g^{\al}_{-3N/2+2,q} & g^{\al}_{-3N/2+1,q}
\end{array}\r]_{(N-1) \times (N-2)},
\ENAS
\EQAS
\tilde{\gb}^{q,\al}_{1,1} &=&
\l[\begin{array}{cccc}
    g^{\al}_{-N/2-1,q} & g^{\al}_{-N/2-2,q}  & \dots & g^{\al}_{-3N/2+2,q} \\
    g^{\al}_{-N/2,q} & g^{\al}_{-N/2-1,q}  & \dots & g^{\al}_{-3N/2+3,q}\\
    \vdots & \vdots & & \vdots \\
    g^{\al}_{N/2-2,q} & g^{\al}_{N/2-1,q} & \dots & g^{\al}_{-N/2+1,q}
\end{array}\r]_{N \times (N-2)}.
\nonumber
\ENAS
Next, we provide details of padding matrices $\Bs_{-1,0}$, $\Bs_{1,0}$, $\Bs_{-1,1}$, $\Bs_{0,1}$ and $\Bs_{1,1}$
for the 2D circulant matrix $\tilde{\gb}^{\al}$ defined in \eqref {cir_t}. These matrices are defined as follows
\begin{align*}
\Bs_{-1,0} &=  \l[~~
 \begin{matrix}
    \tilde{\gb}^{\al}_{-J/2+1} &\tilde{\gb}^{\al}_{-J/2} &\ldots &\tilde{\gb}^{\al}_{-3J/2+1} & \tilde{\gb}^{\al}_{3J/2-1}&\tilde{\gb}^{\al}_{3J/2-2}&\ldots &\tilde{\gb}^{\al}_{J/2}\\
    \tilde{\gb}^{\al}_{-J/2+2} &\tilde{\gb}^{\al}_{-J/2+1} &\ldots &\tilde{\gb}^{\al}_{-3J/2+2} & \tilde{\gb}^{\al}_{-3J/2+1}&\tilde{\gb}^{\al}_{3J/2-1}&\ldots &\tilde{\gb}^{\al}_{J/2+1}\\
    \vdots &\vdots &  &\vdots &\vdots&\vdots&  &\vdots\\
    \tilde{\gb}^{\al}_{J/2} &\tilde{\gb}^{\al}_{J/2-1} &\ldots &\tilde{\gb}^{\al}_{-J/2} & \tilde{\gb}^{\al}_{-J/2-1}&\tilde{\gb}^{\al}_{-J/2-2}&\ldots &\tilde{\gb}^{\al}_{3J/2-1}
\end{matrix}
~~\r]_{(3N-1)J\times (3N-1)(2J+1)},
\end{align*}
\begin{align*}
\Bs_{1,0} &=  \l[~~
 \begin{matrix}
    \tilde{\gb}^{\al}_{-3J/2+1} &\tilde{\gb}^{\al}_{3J/2-1} &\tilde{\gb}^{\al}_{3J/2-2} &\ldots &\tilde{\gb}^{\al}_{J/2+1} & \tilde{\gb}^{\al}_{J/2}&\ldots &\tilde{\gb}^{\al}_{-J/2}\\
      \tilde{\gb}^{\al}_{-3J/2+2} &\tilde{\gb}^{\al}_{-3J/2+1} &\tilde{\gb}^{\al}_{3J/2-1} &\ldots &\tilde{\gb}^{\al}_{J/2+2} & \tilde{\gb}^{\al}_{J/2+1}&\ldots &\tilde{\gb}^{\al}_{-J/2+1}\\
    \vdots &\vdots   &\vdots& &\vdots&\vdots&  &\vdots\\
    \tilde{\gb}^{\al}_{-J/2} &\tilde{\gb}^{\al}_{-J/2-1} &\tilde{\gb}^{\al}_{-J/2-2} &\ldots & \tilde{\gb}^{\al}_{-3J/2+1}&\tilde{\gb}^{\al}_{3J/2-1}&\ldots &\tilde{\gb}^{\al}_{J/2-1}
\end{matrix}
~~\r]_{(3N-1)J\times (3N-1)(2J+1)},
\end{align*}
\begin{align*}
\Bs_{-1,1} &=  \l[~~
 \begin{matrix}
    \tilde{\gb}^{\al}_{J/2-1} &\tilde{\gb}^{\al}_{J/2-2} &\ldots &\tilde{\gb}^{\al}_{-J/2+2}\\
    \tilde{\gb}^{\al}_{J/2} &\tilde{\gb}^{\al}_{J/2-1} &\ldots &\tilde{\gb}^{\al}_{-J/2+3}\\
    \vdots &\vdots& &\vdots\\
    \tilde{\gb}^{\al}_{3J/2-2} &\tilde{\gb}^{\al}_{3J/2-3} &\ldots &\tilde{\gb}^{\al}_{J/2+1}
\end{matrix}
~~\r]_{(3N-1)J\times(3N-1)(J-2)},
\end{align*}
\begin{align*}
\Bs_{0,1} &=  \l[~~
 \begin{matrix}
    \tilde{\gb}^{\al}_{3J/2-1} &\tilde{\gb}^{\al}_{3J/2-2} &\ldots &\tilde{\gb}^{\al}_{J/2+3} &\tilde{\gb}^{\al}_{J/2+2}\\
    \tilde{\gb}^{\al}_{-3J/2+1} &\tilde{\gb}^{\al}_{3J/2-1} &\ldots &\tilde{\gb}^{\al}_{J/2+4} &\tilde{\gb}^{\al}_{J/2+3}\\
    \vdots &\vdots& &\vdots&\vdots\\
    \tilde{\gb}^{\al}_{-J/2-2} &\tilde{\gb}^{\al}_{-J/2-3} &\ldots &\tilde{\gb}^{\al}_{-3J/2+2} &\tilde{\gb}^{\al}_{-3J/2+1}
\end{matrix}
~~\r]_{(3N-1)(J-1)\times(3N-1)(J-2)},
\end{align*}
\begin{align*}
\Bs_{1,1} &=  \l[~~
 \begin{matrix}
    \tilde{\gb}^{\al}_{-J/2-1} &\tilde{\gb}^{\al}_{-J/2-2} &\ldots &\tilde{\gb}^{\al}_{-3J/2+2}\\
    \tilde{\gb}^{\al}_{-J/2} &\tilde{\gb}^{\al}_{-J/2-1} &\ldots &\tilde{\gb}^{\al}_{-3J/2+3}\\
    \vdots &\vdots& &\vdots\\
    \tilde{\gb}^{\al}_{J/2-2} &\tilde{\gb}^{\al}_{J/2-1} &\ldots &\tilde{\gb}^{\al}_{-J/2+1}
\end{matrix}
~~\r]_{(3N-1)J\times(3N-1)(J-1)}.\nonumber
\end{align*}

\end{document}